\documentclass[leqno]{article}
\usepackage{amsmath,amsthm,amssymb,verbatim}

\newtheorem{theorem}{Theorem}[section]
\newtheorem{corollary}[theorem]{Corollary}
\newtheorem{lemma}[theorem]{Lemma}

\newtheorem{proposition}[theorem]{Proposition}
\theoremstyle{remark}
\newtheorem{remark}{Remark}[section]
\newtheorem{example}{Example}[section]
\newtheorem{criterion}{Criterion}[section]

\begin{document}

\title{\textbf{Second-order Lagrangians admitting \\
a first-order Hamiltonian formalism}}
\author{\textsc{E. Rosado Mar\'{\i}a} \\
Departamento de Matem\'atica Aplicada \\
Escuela T\'ecnica Superior de Arquitectura, UPM \\
Avda.\ Juan de Herrera 4, 28040-Madrid, Spain \\
\emph{E-mail:\/} \texttt{eugenia.rosado@upm.es}
\and
\textsc{J. Mu\~{n}oz Masqu\'e} \\
Instituto de Tecnolog\'{\i}as F\'{\i}sicas y de la Informaci\'on, CSIC \\
C/ Serrano 144, 28006-Madrid, Spain \\
\emph{E-mail:\/} \texttt{jaime@iec.csic.es}}
\date{}

\maketitle

\begin{abstract}
\noindent Second-order Lagrangian densities admitting a first-order
Hamiltonian formalism are studied; namely, i) for each second-order
Lagrangian density on an arbitrary fibred manifold $p\colon E\to N$ the
Poincar\'e-Cartan form of which is projectable onto $J^1E$, by using a new
notion of regularity previously introduced, a first-order Hamiltonian
formalism is developed for such a class of variational problems; ii) the
existence of first-order equivalent Lagrangians are discussed from a local
point of view as well as global; iii) this formalism is then applied to
classical Einstein-Hilbert Lagrangian and a generalization of the BF theory.
The results suggest that the class of problems studied is a natural
variational setting for GR.
\end{abstract}

\bigskip

\noindent \textit{PACS 2010:\/} 02.30.Ik,
02.30.Xx,
02.40.Vh,
04.20.Fy,
11.10.Ef,
11.10.Kk.

\medskip

\medskip


\section{Preliminaries}

\subsection{Legendre \& Poincar\'e-Cartan forms}

Below, a fibred manifold $p\colon E\to N$ is considered
over a connected $n$-dimensional manifold $N$ oriented by
a volume form $v=dx^{1}\wedge \cdots \wedge dx^n $. The bundle
of $k$-jets of local sections of $p$ is denoted by
$p^{k}\colon J^{k}E\to N$, with
natural projections $p_{l}^{k}\colon J^{k}E\to J^{l}E$, $k\geq l$.

Every fibred coordinate system $(x^{j},y^{\alpha })$, $1\leq j\leq n$,
$1\leq \alpha \leq m=\dim E-n$, for the submersion $p$, induces a coordinate
system $(x^{j},y_{I}^{\alpha })$ ($I=(i_{1},\dotsc ,i_{n})$ being a
multi-index in $\mathbb{N}^n $ of order $|I|=i_{1}+\ldots +i_{n}\leq r$)
on $J^rE$ defined by,
\begin{equation*}
y_{I}^{\alpha }\left( j_{x}^rs\right) =\tfrac{\partial ^{|I|}(y^{\alpha
}\circ s)}{\partial (x^{1})^{i_{1}}\ldots \partial (x^n )^{i_{n}}}(x),
\end{equation*}
where $s$ is a local section of $p$. We also set
$(j)=(0,\dotsc ,0,\overset{(j}{1},0,\dotsc ,0)\in \mathbb{N}^n $,
$(jk)=(j)+(k)$, etc., and $y_{(j)}^{\alpha }=y_{j}^{\alpha }$.

The \textit{Legendre form} of a second-order Lagrangian density
$\Lambda =Lv$, defined on $p\colon E\to N$, $L\in C^{\infty }(J^{2}E)$,
is the $V^{\ast }(p^{1})$-valued $p^{3}$-horizontal $(n-1)$-form
$\omega _{\Lambda }$ on $J^{3}E$ locally given by (e.g.,
see \cite{Mu2}, \cite{MR}, \cite{SCr}),
\begin{equation*}
\omega _{\Lambda }=(-1)^{i-1}L_{\alpha }^{i0}v_{i}\otimes dy^{\alpha
}+(-1)^{i-1}L_{\alpha }^{ij}v_{i}\otimes dy_{j}^{\alpha },
\end{equation*}
where $v_{i}=dx^{1}\wedge \cdots \wedge \widehat{dx^{i}}\wedge
\cdots \wedge dx^n $, and
\begin{align}
L_{\alpha }^{ij}& =\tfrac{1}{2-\delta _{ij}}\tfrac{\partial L}{\partial
y_{(ij)}^{\alpha }},  \label{f11} \\
L_{\alpha }^{i0}& =\tfrac{\partial L}{\partial y_{i}^{\alpha }}
-\tfrac{1}{2-\delta _{ij}}D_{j}
\left( \tfrac{\partial L}{\partial y_{(ij)}^{\alpha }}\right) ,
\label{f12}
\end{align}
and $D_{j}=\tfrac{\partial }{\partial x^{j}}
+\sum_{|I|=0}^{\infty }\sum_{\alpha =1}^{m}y_{I+(j)}^{\alpha }
\tfrac{\partial }{\partial y_{I}^{\alpha }}$ denotes
the total derivative with respect to the coordinate $x^{j}$.
The \textit{Poincar\'{e}-Cartan form} (or P-C form for short)
attached to $\Lambda $ is the ordinary $n$-form on $J^{3}E$
given by $\Theta _{\Lambda }=(p_{2}^{3})^{\ast }\theta ^{2}
\wedge \omega _{\Lambda }+\Lambda $ (e.g., see \cite{Mu2},
\cite{SCr}), where $\theta ^{1}$, $\theta ^{2}$ are
the first- and second-order structure forms on $J^{1}E$,
$J^{2}E$, locally given by (cf.\ \cite{Mu1}, \cite{Saunders}),
$\theta ^{1}=\theta ^{\alpha }
\otimes \tfrac{\partial }{\partial y^{\alpha }}$,
$\theta ^{2}=\theta ^{\alpha }
\otimes \tfrac{\partial }{\partial y^{\alpha }}
+\theta _{h}^{\alpha }
\otimes \tfrac{\partial }{\partial y_{h}^{\alpha }}$,
respectively, and $\theta ^{\alpha }=dy^{\alpha }
-y_{k}^{\alpha }dx^{k}$, $\theta _{i}^{\alpha }
=dy_{i}^{\alpha}-y_{(ik)}^{\alpha }dx^{k}$, is the standard basis
of contact $1$-forms, and the exterior product
of $(p_{2}^{3})^{\ast }\theta ^{2}$ and the Legendre
form, is taken with respect to the standard pairing
$V(p^{1})\times _{J^{1}E}V^{\ast }(p^{1})\to \mathbb{R}$.
\subsection{Projecting onto $J^2E$ or $J^1E$}
The most outstanding difference with a first-order Lagrangian density is
that the Legendre and Poincar\'e-Cartan forms associated with a second-order
Lagrangian density are generally defined on $J^3E$, thus increasing by one
the order of the Lagrangian density $\Lambda $.

For certain second-order Lagrangian densities it is known that the P-C form
is projectable onto $J^{2}E$; e.g., see \cite{GM1}. More precisely, the P-C
form of a second-order Lagrangian projects onto $J^{2}E$ if and only if the
following system of PDEs holds (cf.\ \cite{DM}, \cite{GM1}):
\begin{equation*}
\tfrac{1}{2-\delta _{ib}}\tfrac{\partial ^{2}L}{\partial y_{ac}^{\beta}
\partial y_{ib}^{\alpha }}+\tfrac{1}{2-\delta _{ia}}
\tfrac{\partial ^{2}L}{\partial y_{bc}^{\beta }\partial y_{ia}^{\alpha }}
+\tfrac{1}{2-\delta _{ic}}\tfrac{\partial ^{2}L}{\partial y_{ab}^{\beta }
\partial y_{ic}^{\alpha }}=0,
\end{equation*}
for all indices $1\leq a\leq b\leq c\leq n$, $\alpha ,\beta =1,\dotsc ,m$.

More surprisingly, there exist second-order Lagrangians for which the
associated P-C form projects not only on $J^{2}E$ but also on $J^{1}E$.
Notably, this is the case of the Einstein-Hilbert Lagrangian in General
Relativity.

As is well known (e.g., see \cite[(1.3)]{GoS}, \cite[2.1]{MR}),
$p_{r-1}^r\colon J^rE\to J^{r-1}E$ admits an affine bundle structure modelled
over the vector bundle
\begin{equation}  \label{W^r}
W^r=(p^{r-1})^\ast S^rT^\ast N \otimes (p_0^{r-1})^\ast V(p)\to J^{r-1}E.
\end{equation}

\begin{proposition}[cf. \protect\cite{KS}, \protect\cite{EJ}]
\label{proposition1} The Poincar\'{e}-Cartan form attached to a Lagrangian
$L\in C^{\infty }(J^{2}E)$ projects onto $J^{1}E$ if and only if $L$ is an
affine function with respect to the affine structure of $p_{1}^{2}\colon
J^{2}E\to J^{1}E$, namely
\begin{equation}
L=L_{\alpha }^{ij}y_{(ij)}^{\alpha }+L_{0},\quad L_{\alpha }^{ji}=L_{\alpha
}^{ij}\in C^{\infty }(J^{1}E),L_{0}\in C^{\infty }(J^{1}E),  \label{affine}
\end{equation}
and the following equations hold:
\begin{equation}
\tfrac{\partial L_{\beta }^{ih}}{\partial y_{a}^{\alpha }}=\tfrac{\partial
L_{\alpha }^{ia}}{\partial y_{h}^{\beta }},\quad a,h,i=1,\dotsc ,n,\;\alpha
,\beta =1,\dotsc ,m.  \label{first_tris}
\end{equation}
\end{proposition}

The equations \eqref{first_tris} admit a variational meaning. The
Euler-Lagrange (or E-L for short) operator of an arbitrary second-order
Lagrangian can be written in terms of the coefficients of the P-C form (see
the formulas \eqref{f11}, \eqref{f12}) as follows:
\begin{align*}
\mathcal{E}_{\alpha }(L)& =\sum_{i\leq j}D_{i}D_{j}
\left( \tfrac{\partial L}{\partial y_{(ij)}^{\alpha }}
\right) -D_{i}
\left( \tfrac{\partial L}{\partial
y_{i}^{\alpha }}\right) +\tfrac{\partial L}{\partial y^{\alpha }} \\
& =\tfrac{\partial L}{\partial y^{\alpha }}-D_{i}\left( L_{\alpha
}^{i0}\right) ,\quad 1\leq \alpha \leq m.
\end{align*}
The E-L equations for an affine second-order Lagrangian $L$, given as in the
formula \eqref{affine}, are of third order and they are of second order if
and only if the equations \eqref{first_tris} hold
(cf.\ \cite[Proposition 2.2]{EJ}).

As the projection $p_{r-1}^r\colon J^rE\to J^{r-1}E$ admits an
affine-bundle structure, a natural vector-bundle isomorphism is obtained,
\begin{equation}
I^r\colon (p_{r-1}^r)^{\ast }W^r=(p^r)^{\ast }S^rT^{\ast }N\otimes
(p_{0}^r)^{\ast }V(p)\overset{\cong }{\longrightarrow }V(p_{r-1}^r),
\label{I^r}
\end{equation}
where the vector bundle $W^r$ is defined in \eqref{W^r}. Given an
arbitrary vector bundle $W\to N$, there exists an antiderivation
\begin{equation*}
d_{E\!/\!N}\colon \Gamma (E,\wedge ^rV^{\ast }(p)\otimes p^{\ast
}W)\to \Gamma (E,\wedge ^{r+1}V^{\ast }(p)\otimes p^{\ast }W)
\end{equation*}
of degree $+1$---called the fibre differential
(e.g., see \cite[(1.9)]{GoS})---such that,
$d_{E\!/\!N}(fp^{\ast }\xi )=\left. df\right\vert
_{V(p)}\otimes \xi $, for all $f\in C^{\infty }(E)$ and all
$\xi \in \Gamma (E,W)$. (In the previous paragraph, the relevant fact
is that the vector bundle $W\to N$ is defined over the base manifold $N$,
and not over the fibred manifold $E$.)

In what follows we are mainly concerned with the fibre derivative
$d_{J^1E\!/\!J^0E}$, which will simply be denoted by $d_{10}$
for the sake of simplicity.

A Lagrangian $L\in C^{\infty }(J^{2}E)$ is an affine function with respect
to the affine structure of $p_{1}^{2}\colon J^{2}E\to J^{1}E$ if
there exists a linear form $w_{L}\colon W^{2}\to \mathbb{R}$, which
is unique, such that, $L(\tau +j_{x}^{2}s)=w_{L}(\tau )+L(j_{x}^{2}s)$,
$\forall \tau \in S^{2}T_{x}^{\ast }N\otimes V_{s(x)}(p)$
and $\forall j_{x}^{2}s\in J^{2}E$.

By using the $(W^{2})^{\ast }
\cong (p^{1})^{\ast }S^{2}TN\otimes (p_{0}^{1})^{\ast }V^{\ast }(p)$,
the linear form $w_{L}$ defines a section of the vector bundle
$(p^{1})^{\ast }S^{2}TN\otimes (p_{0}^{1})^{\ast}V^{\ast }(p)\to J^{1}E$.
If $L$ is locally given by the formula \eqref{affine},
then $w_{L}=L_{\alpha }^{hi}\tfrac{\partial }{\partial x^{h}}
\odot \tfrac{\partial }{\partial x^{i}}\otimes dy^{\alpha }|_{V(p)}$,
where the symbol $\odot $ denotes symmetric product.

If $\iota ^2\colon (W^2)^\ast
\to(p^1)^\ast \otimes ^2 TN\otimes(p_0^1)^\ast V^\ast (p)$
is the natural embedding, then we consider the section
\begin{equation}  \label{w_prime}
w_L^\prime =\tfrac{1}{2} \bigl( \tilde{I}^1\circ \iota^2 \circ w_L \bigr)
\colon J^1E\to (p^1)^\ast T N \otimes V^\ast (p_0^1)
\end{equation}
obtained by composing the following mappings:
\begin{multline*}
J^1E\overset{w_L}{\longrightarrow }(p^1)^\ast S^2TN\otimes(p_0^1)^\ast
V^\ast (p) =(W^2)^\ast \overset{\iota ^2}{\longrightarrow }(p^1)^\ast
\otimes^2TN\otimes(p_0^1)^\ast V^\ast (p) \\
=(p^1)^\ast TN\otimes
\left[ (p^1)^\ast TN \otimes(p_0^1)^\ast V^\ast (p)\right]
\overset{\tilde{I}^1}{-\!\!\!\! \longrightarrow } (p^1)^\ast
TN\otimes V^\ast (p_0^1),
\end{multline*}
where $\tilde{I}^1=1_{(p^1)^\ast TN} \otimes((I^1)^\ast )^{-1}$ is the
isomorphism deduced from \eqref{I^r} for $r=1$.
As $I^1(dx^a\otimes \partial /\partial y^\alpha )
=\partial /\partial y_a^\alpha $, dually we obtain
$(I^1)^\ast (d_{10}y_a^\alpha )
=\partial /\partial x^a\otimes dy^\alpha |_{V(p)}$.

Hence $w_L^\prime =L_\alpha ^{hi}d_{10} \left( y_h^\alpha \right)
\otimes \tfrac{\partial }{\partial x^i}$.

\begin{remark}
\label{remark1} The equations \eqref{first_tris} simply means that for every
index $h$ the form $\eta ^{h}=L_{\alpha }^{hi}dy_{i}^{\alpha }$
is $d_{10}$-closed, namely $d_{10}\eta ^{h}=0$. Hence, there exist functions
$L^{i}\in C^{\infty }(J^{1}E)$ such that locally,
\begin{equation}
\mathrm{(i)}\;L_{\alpha }^{ih}
=\tfrac{\partial L^{i}}{\partial y_{h}^{\alpha }},
\quad
\mathrm{(ii)}\;\tfrac{\partial L^{h}}{\partial y_{i}^{\alpha }}
=\tfrac{\partial L^{i}}{\partial y_{h}^{\alpha }},
\qquad 1\leq \alpha \leq m,\;h,i=1,\dotsc ,n,\label{L^i}
\end{equation}
the equations (ii) above being a consequence of the symmetry
$L_{\alpha }^{hi}=L_{\alpha }^{ih}$.
\end{remark}

Letting $W=TN$ in the definition of the fibre differential above, recalling
that the Poincar\'{e} lemma also holds for fibre differentiation (e.g., see
\cite{MP}) and recalling that the fibres of $p_{0}^{1}\colon J^{1}E\to E$
are simply connected as they are diffeomorphic to $\mathbb{R}^{mn}$,
the following global characterization of second-order
variational problems with a P-C form projecting onto $J^{1}E$, is obtained:
\begin{proposition}
\emph{(see \cite[Proposition 3.1]{EJ})} \label{proposition2} The Poincar\'e
-Cartan form of a Lagrangian $L\in C^{\infty }(J^{2}E)$ projects onto $J^{1}E$
if and only if $L$ is an affine function with respect to the affine
structure of $p_{1}^{2}\colon J^{2}E\to J^{1}E$ and the $TN$-valued
$1$-form $w_{L}^{\prime }$ defined in the formula \emph{\eqref{w_prime}}
is $d_{10}$-closed. In this case, for every global (smooth) section
$\sigma \colon E\to J^{1}E$ of $p_{0}^{1}$, there exists a unique globally
defined section $w_{L}^{\sigma }\in \Gamma (J^{1}E,(p^{1})^{\ast }TN)$ such
that, $d_{10}\left( w_{L}^{\sigma }\right) =w_{L}^{\prime }$,
$w_{L}^{\sigma }\left( \sigma (e)\right) =0,\forall e\in E$.
\end{proposition}
\begin{remark}
\label{remark_sigma} A general procedure to obtain global sections
$\sigma \colon E\to J^1E$ of $p^1_0$ is to use Ehresmann (or non-linear)
connections, i.e., to use a differential $1$-form $\gamma$ on $E$ taking
values in the vertical sub-bundle $V(p)$ such that $\gamma(X)=X$,
$\forall X\in V(p)$; hence, locally (cf.\ \cite{MR2}), $\gamma
= (dy^\alpha +\gamma ^\alpha _jdx^j)
\otimes \tfrac{\partial }{\partial y^\alpha }$,
$\gamma ^\alpha _j\in C^\infty (E)$.
The vertical differential of a section $s\colon U\to E$ (defined
on a neigbourhood $U$ of $x\in N$) at $e=s(x)$ is defined
to be the linear mapping $(d^vs)_e\colon T_eE\to V_e(p)$,
$(d^vs)_eX=X-s_\ast p_\ast (X)$, $\forall X\in T_eE$. We claim that for every
$e\in E$, there exists a unique $j^1_xs\in J^1E$ such that, i) $s(x)=e$,
where $x=p(e)$, and ii) $(d^vs)_e=\gamma _e$. In fact, one has $(\partial
(y^\alpha \circ s)/\partial x^j)(x) =-\gamma ^\alpha _j(e)$, and the section
$\sigma ^\gamma $ attached to $\gamma $ is defined by, $\sigma ^\gamma (e)=j^1_xs$.
\end{remark}
\subsection{Summary of contents}
Bearing the previous definitions and notations in mind, the paper is
organized as follows: In section \ref{Hamilton_formalism} the Hamiltonian
function, the momenta, and the Hamilton-Cartan equations attached to each of
the aforementioned Lagrangians are introduced as a consequence of a normal
form for their P-C form. This section also deals with the notion of
regularity for the class of second-order variational problems with a P-C
form that projects to first-order jet bundle. Although the Hessian metric
vanishes identically for the Lagrangians of such class, a suitable notion of
regularity is introduced for them.

In \cite{EJ} the study of the formal integrability of the field equations of
second-order Lagrangians with projectable P-C form to first order in their
Hamiltonian form is devoted. In the real analytic case, this allows one to
solve the Cauchy initial value problem for this class of Lagrangians. The
previous sections are then applied to GR in section \ref{GR}, thus showing
how the theory developed fits very well to the standard Lagrangians in this
setting. Specifically, section \ref{EH} studies Einstein-Hilbert Lagrangian
from this point of view, proving its regularity and giving a new statement
for the initial problem. Similarly, section \ref{BF} provides a strong
generalization of the classical Lagrangians in BF-theory, again showing that
the results obtained above can naturally be applied to these new
Lagrangians. In section \ref{first_equiv}, the existence of first-order
Lagrangians variationally equivalent to a second-order Lagrangian admitting
a first-order Hamiltonian formalism is studied, both from local and global
point of view. This generalizes previous results obtained for the E-H
Lagrangian in \cite{MJE}. Section \ref{symmetriesNoether} introduces the
notions of symmetry and Noether invariant for the class of variational
problems dealt with throughout the paper and section \ref{symmetriesEH}
discusses in particular such concepts for the E-H Lagrangian. Finally, in
section \ref{ss4.3} the notion of a Jacobi field along an extremal is
introduced and the presymplectic structure attached to a variational problem
is defined. Several explicit examples are also developed in detail.

\section{Regularity and Hamiltonian formalism\label{Hamilton_formalism}}

In the usual (i.e., first-order) calculus of variations, a section $s$ is an
extremal of the Lagrangian density $\Lambda $ on $J^{1}E$ if and only if it
satisfies the so-called \textquotedblleft Hamilton-Cartan
equations\textquotedblright\ (or H-C for short; e.g., see \cite[(3.8)]{GoS},
\cite[(1)]{GM1}), namely, if and only if the following equation holds:
$(j^{1}s)^{\ast }(i_{X}d\Theta _{\Lambda })=0$ for every $p^{1}$-vertical
vector field $X$ on $J^{1}E$.

If $\Lambda =Lv$ is an arbitrary second-order Lagrangian density on $E$,
then the following formula holds (e.g., see \cite{Mu2}):
\begin{equation}
d\Theta _{\Lambda }
=\mathcal{E}_{\alpha }(L)\theta ^{\alpha }\wedge v+\eta _{n+1}(L),
\label{differential_P-C}
\end{equation}
where $\eta _{n+1}(L)=(-1)^{i}\eta _{2}^{i}(L)\wedge v_{i}$ and
$\eta _{2}^{i}(L)$\ is the $2$-contact $2$-form given by,
\begin{eqnarray*}
\eta _{2}^{i}(L)
&=&\tfrac{\partial L_{\alpha }^{i0}}{\partial y^{\beta }}
\theta ^{\alpha }\wedge \theta ^{\beta }
+\left( \tfrac{\partial L_{\alpha }^{i0}}{\partial y_{j}^{\beta }}
-\tfrac{\partial L_{\beta }^{ij}}{\partial y^{\alpha }}\right)
\theta ^{\alpha }\wedge \theta _{j}^{\beta } \\
&&+\sum_{j\leq k}
\tfrac{\partial L_{\alpha }^{i0}}{\partial y_{(jk)}^{\beta }}
\theta ^{\alpha }\wedge \theta _{(jk)}^{\beta }
+\sum_{i\leq k\leq l}
\tfrac{\partial L_{\alpha }^{i0}}{\partial y_{(jkl)}^{\beta }}
\theta ^{\alpha}\wedge \theta _{(jkl)}^{\beta } \\
&&+\tfrac{\partial L_{\alpha }^{ij}}{\partial y_{k}^{\beta }}
\theta _{j}^{\alpha }\wedge \theta _{k}^{\beta }
+\sum_{k\leq l}\tfrac{\partial L_{\alpha }^{ij}}{\partial y_{(kl)}^{\beta }}
\theta _{j}^{\alpha }\wedge \theta _{(kl)}^{\beta }.
\end{eqnarray*}
From the formula \eqref{differential_P-C} it follows that the H-C equations
also characterize critical sections for a second-order density $\Lambda $;
i.e., $s$ is an extremal for $\Lambda $ if and only if, $(j^{3}s)^{\ast
}(i_{X}d\Theta _{\Lambda })=0$ for every $p^{3}$-vertical vector field $X$
on $J^{3}E$.
\begin{remark}
\label{remark.4.1} If the P-C form of a second-order density $\Lambda $
projects onto $J^{1}E$, then its H-C equations have the same formal
expression of a first-order density (see the formula \eqref{HCequations}
below), although there is no first-order density having $\Theta _{\Lambda }$
as its P-C form. In fact, the P-C form of a first-order Lagrangian density
$\tilde{\Lambda}=\tilde{L}v$, $\tilde{L}\in C^{\infty }(J^{1}E)$, is given
by,
\begin{equation}
\Theta _{\tilde{\Lambda}}=(-1)^{i-1}\tfrac{\partial \tilde{L}}{\partial
y_{i}^{\alpha }}dy^{\alpha }\wedge v_{i}+\tilde{H}v,
\quad \tilde{H}=\tilde{L}
-\tfrac{\partial \tilde{L}}{\partial y_{i}^{\alpha }}y_{i}^{\alpha }.
\label{ThetaLambdaTilde1}
\end{equation}
If $\Theta _{\Lambda }=\Theta _{\tilde{\Lambda}}$, then
the following three equations are obtained:
\begin{equation*}
1)\text{ }L_{\alpha }^{ih}=0,\qquad 2)\text{ }L_{0}
-y_{i}^{\alpha }L_{\alpha }^{i0}=\tilde{L}
-\tfrac{\partial \tilde{L}}{\partial y_{i}^{\alpha }}
y_{i}^{\alpha },
\qquad 3)\text{ }L_{\alpha }^{i0}
=\frac{\partial \tilde{L}}{\partial y_{i}^{\alpha }}.
\end{equation*}
From \eqref{affine} and $1)$ it follows $L=L_{0}$; hence $L$
is of first order.

Moreover, taking \eqref{f12} into account, the formulas $2)$ and $3)$
above are respectively rewritten as $L_{0}-\tilde{L}=y_{i}^{\alpha }
\frac{\partial(L_{0}-\tilde{L})}{\partial y_{i}^{\alpha }}$,
$\frac{\partial (L_{0}-\tilde{L})}{\partial y_{i}^{\alpha }}=0$.
Hence $\tilde{L}=L$.
\end{remark}
\begin{theorem}
\emph{(see \cite[Theorem 4.1]{EJ})}\label{differentialPCform}
If $\Lambda =Lv $ is a second-order Lagrangian density on $E$
whose Poincar\'e-Cartan form projects onto $J^{1}E$, then letting
\begin{align}
p_{\alpha }^{i}& =L_{\alpha }^{i0}-\tfrac{\partial L^{i}}{\partial y^{\alpha
}},\quad 1\leq \alpha \leq m,\;1\leq i\leq n,  \label{p's} \\
H& =L_{0}-y_{i}^{\alpha }L_{\alpha }^{i0}-\tfrac{\partial L^{i}}{\partial
x^{i}},  \label{H}
\end{align}
where the functions $L^{i}$ are defined by the formulas
\eqref{L^i}-\emph{(i)}, the following formula holds:
\begin{equation}
d\Theta _{\Lambda }=(-1)^{i-1}dp_{\alpha }^{i}\wedge dy^{\alpha }
\wedge v_{i}+dH\wedge v.  \label{differential_P-C_bis}
\end{equation}
Furthermore, if the linear forms
$d_{10}(p_{\alpha }^{i})\colon V(p_{0}^{1})\to \mathbb{R}$,
$1\leq \alpha \leq m$, $1\leq i\leq n$,
are linearly independent, then a section $s\colon N\to E$
is an extremal for $\Lambda $ if and only if it satisfies
the following equations:
\begin{equation}
\left\{
\begin{array}{ll}
0=\dfrac{\partial (p_{\alpha }^{i}\circ j^{1}s)}{\partial x^{i}}
-\dfrac{\partial H}{\partial y^{\alpha }}\circ j^{1}s,
& 1\leq \alpha \leq m,\smallskip \\
0=\dfrac{\partial (y^{\alpha }\circ s)}{\partial x^{i}}
+\dfrac{\partial H}{\partial p_{\alpha }^{i}}\circ j^{1}s,
& 1\leq \alpha \leq m,\;1\leq i\leq n.
\end{array}
\right.  \label{HCequations}
\end{equation}
\end{theorem}

As is well known (e.g., see \cite{GoS}), if the Hessian metric
$\operatorname{Hess}(L)$ of a first-order density $\Lambda =Lv$
is non-singular, then every section $s^{1}\colon N\to J^{1}E$
of the projection $p^{1}\colon J^{1}E\to N$ that satisfies
the P-C equation for $\Lambda $ is holonomic; i.e., $s^{1}$
coincides with the $1$-jet extension of the section
$s=p_{0}^{1}\circ s^{1}$ of the projection $p$. Namely,
$(s^{1})^{\ast }(i_{X}d\Theta _{\Lambda })=0$ for every
$p^{1}$-vertical vector field $X$ on $J^{1}E$, implies $s^{1}=j^{1}s$.

In the case of a second-order density with a P-C form projecting
onto $J^1E$, the following result holds:
\begin{proposition}[\protect\cite{EJ}]
\label{regular1} If $\Lambda =Lv$ is a second-order Lagrangian on $E$
such that, \emph{(i) }its Poincar\'e-Cartan form $\Theta _{\Lambda }$
projects onto $J^{1}E$, \emph{(ii) }the linear forms
$d_{10}(p_{\alpha }^{i})\colon V(p_{0}^{1})\to \mathbb{R}$,
$1\leq \alpha \leq m$, $1\leq i\leq n$, where the functions
$p_{\alpha }^{i}$ are introduced in \eqref{p's}, are linearly
independent, then every solution to its H-C equations,
is holonomic.
\end{proposition}

As $p_{0}^{1}\colon J^{1}E\to E$ is an affine bundle modelled
over $W^{1}=p^{\ast }(T^{\ast }N)\otimes V(p)$ (cf.\ \eqref{W^r}),
there is a canonical isomorphism
$I\colon (p_{0}^{1})^{\ast }W^{1}
\overset{\cong }{\to }V(p_{0}^{1})$
locally given by,
$I(j_{x}^{1}s,(dx^{i})_{x}
\otimes (\partial /\partial y^{\alpha })_{s(x)})
=(\partial /\partial y_{i}^{\alpha })_{j_{x}^{1}s}$.

According to the previous lemma, we can define a bilinear form
\begin{equation}
\left\{
\begin{array}{l}
b_{\Lambda }\colon (p_{0}^{1})^{\ast }W^{1}\times _{J^{1}E}
(p_{0}^{1})^{\ast }W^{1}\to \mathbb{R},
\medskip \\
b_{\Lambda }
\left(
j_{x}^{1}s;w_{0}\otimes Y_{0},w_{1}\otimes Y_{1}
\right)
=\left\langle
w_{0},(\phi _{v}^{1})^{-1}
\left( i_{Y_{0}}i_{Y}(d\Theta _{\Lambda })
\right)
\right\rangle ,
\smallskip \\
w_{a}\in T_{x}^{\ast }N,Y_{a}\in V_{s(x)}(p),a=0,1;\;
Y=I(j_{x}^{1}s,w_{1}\otimes Y_{1}),
\end{array}
\right.  \label{B_Lambda}
\end{equation}
where $\phi _{v}^{k}$ is the isomorphism defined by
\begin{equation}
\phi _{v}^{k}\colon \wedge ^{k}T_{x}N\to \wedge ^{n-k}T_{x}^{\ast }N
\label{phi^k}
\end{equation}
for every $1\leq k\leq n-1$, obtained by contracting with $v$, namely
\begin{equation*}
\phi _{v}^{k}(X_{1}\wedge \cdots \wedge X_{k})=i_{X_{1}}\ldots
i_{X_{k}}v,\quad \forall X_{1},\dotsc ,X_{k}\in T_{x}N.
\end{equation*}
If $w_{0}=(dx^{i})_{x}$ and $Y_{0}=(\partial /\partial y^{\alpha })_{s(x)}$,
then one readily obtains,
\begin{eqnarray*}
i_{Y_{0}}i_{Y}(d\Theta _{\Lambda })
&=&(-1)^{i-1}\left(
\frac{\partial L_{\alpha }^{i0}}{\partial y_{j}^{\beta }}(j_{x}^{1}s)
-\frac{\partial L_{\beta }^{ij}}{\partial y^{\alpha }}(j_{x}^{1}s)
\right)
(v_{i})_{x}, \\
\left\langle
w_{0},(\phi _{v}^{1})^{-1}\left( i_{Y_{0}}i_{Y}(d\Theta _{\Lambda })
\right)
\right\rangle
&=&\frac{\partial L_{\alpha }^{i0}}{\partial y_{j}^{\beta }}
(j_{x}^{1}s)
-\frac{\partial L_{\beta }^{ij}}{\partial y^{\alpha }}(j_{x}^{1}s).
\end{eqnarray*}
In other words,
\begin{equation*}
b_{\Lambda }\left( j_{x}^{1}s;\left( dx^{i}\right) _{x}
\otimes \left( \frac{\partial }{\partial y^{\alpha }}\right) _{s(x)},
\left( dx^{j}\right) _{x}\otimes
\left( \frac{\partial }{\partial y^{\beta }}\right) _{s(x)}
\right) =\frac{\partial L_{\alpha }^{i0}}{\partial y_{j}^{\beta }}
(j_{x}^{1}s)-\frac{\partial L_{\beta }^{ij}}{\partial y^{\alpha }}
(j_{x}^{1}s).
\end{equation*}
Hence, the next result follows:
\begin{corollary}
\label{RegularityCondition} Let $\Lambda $ be a second-order density on $E$
whose P-C form projects onto $J^1E$. If the bilinear form defined in
\eqref{B_Lambda} is non-singular, then every solution to the H-C equations
for $\Lambda $ is holonomic.
\end{corollary}
\begin{proposition}
\emph{(see \cite[Proposition 5.4]{EJ})}The bilinear form $b_{\Lambda }$
defined in \emph{\eqref{B_Lambda}} is symmetric.
\end{proposition}
In fact, if $\bar{L}$ is the Lagrangian defined by
\begin{equation}
\bar{L}=L_{0}-\frac{\partial L^{i}}{\partial x^{i}}
-y_{i}^{\alpha }\frac{\partial L^{i}}{\partial y^{\alpha }},
\label{barL}
\end{equation}
then, as a calculation shows,
\begin{equation}
p_{\alpha }^{i}=\frac{\partial \bar{L}}{\partial y_{i}^{\alpha }}.
\label{momenta}
\end{equation}
\section{Applications to GR\label{GR}}
\subsection{Einstein-Hilbert Lagrangian\label{EH}}
Below, we follow \cite{EJ}. Let $p_{M}\colon M=M(N)\to N$ be the
bundle of pseudo-Riemannian metrics of a given signature $(n^{+},n^{-})$,
$n^{+}+n^{-}=n$. Every coordinate system $(x^{i})_{i=1}^n $ on an open
domain $U\subseteq N$ induces a coordinate system $(x^{i},y_{jk})$ on
$(p_{M})^{-1}(U)$, where the functions $y_{jk}=y_{kj}$ are defined by,
\begin{equation*}
\begin{array}{ll}
g_{x}\!\!\! & =y_{ij}(g_{x})(dx^{i})_{x}\otimes (dx^{j})_{x}
\smallskip \\
\!\!\! & =\sum\limits_{i\leq j}\frac{1}{1+\delta _{ij}}
y_{ij}(g_{x})(dx^{i})_{x}\odot (dx^{j})_{x},
\end{array}
\quad \forall g_{x}\in (p_{M})^{-1}(U).
\end{equation*}
Following the notations in \cite{KN}, the Ricci tensor field attached
to the symmetric connection $\Gamma $ is given by
$S^{\Gamma }(X,Y)=\operatorname{trace}(Z\mapsto R^{\Gamma }(Z,X)Y)$,
where $R^{\Gamma }$ denotes the curvature
tensor field of the covariant derivative $\nabla ^{\Gamma }$ associated to
$\Gamma $ on the tangent bundle; hence
$S^{\Gamma }=(R^{\Gamma })_{jl}dx^{l}\otimes dx^{j}$,
where $(R^{\Gamma })_{jl}=(R^{\Gamma })_{jkl}^{k}$,
and $(R^{\Gamma })_{jkl}^{i}=\partial \Gamma _{jl}^{i}/\partial x^{k}
-\partial \Gamma _{jk}^{i}/\partial x^{l}+\Gamma _{jl}^{m}\Gamma _{km}^{i}
-\Gamma _{jk}^{m}\Gamma _{lm}^{i}$.

The E-H Lagrangian density is given by
\begin{equation*}
(\Lambda
_{EH})_{j_{x}^{2}g}=g^{ij}(x)(R^{g})_{ihj}^{h}(x)v_{g}(x)=L_{EH}(j_{x}^{2}g)v_{x},
\end{equation*}
where $v$ is the standard volume form, $R^{g}$ is the curvature tensor
of the Levi-Civita connection $\Gamma ^{g}$ of the metric $g$, and $v_{g}$
denotes the Riemannian volume form attached to $g$; i.e., in coordinates,
$v_{g}=\sqrt{|\det ((g_{ab})_{a,b=1}^n )|}v$. Hence,
\begin{equation}
L_{EH}\circ j^{2}g=(\rho \circ g)(y^{ij}\circ g)(R^{g})_{ihj}^{h},
\quad \rho =\sqrt{|\det ((y_{ab})_{a,b=1}^n )|}. \label{rho}
\end{equation}
The local expression for $L_{EH}$ is readily seen to be

\begin{equation*}
L_{EH}=\rho \sum\nolimits_{a,b}\sum\nolimits_{c,d}\left(
y^{ac}y^{bd}-y^{ab}y^{cd}\right) y_{ab,cd}+\left( L_{EH}\right) _{0},
\end{equation*}
\begin{align}
\left( L_{EH}\right) _{0}& =\tfrac{\rho }{2}\sum_{r\leq s}\sum_{k\leq l}
\tfrac{1}{(1+\delta _{kl})(1+\delta _{rs})}\left( 2y^{rs}\left(
y^{ki}y^{jl}+y^{li}y^{jk}\right) -2y^{kl}y^{sr}y^{ji}\right. \label{L_EH_0}
\\
& +2y^{kl}\left( y^{jr}y^{si}+y^{js}y^{ri}\right) +3y^{ij}\left(
y^{kr}y^{ls}+y^{ks}y^{lr}\right)  \notag \\
& -y^{ir}\left( y^{ks}y^{jl}+y^{ls}y^{jk}\right) -y^{is}\left(
y^{kr}y^{jl}+y^{lr}y^{jk}\right)  \notag \\
& \left. -2y^{ki}\left( y^{sl}y^{jr}+y^{rl}y^{js}\right) -2y^{li}\left(
y^{sk}y^{jr}+y^{rk}y^{js}\right) \right) y_{kl,i}y_{rs,j}.  \notag
\end{align}
Hence $L_{EH}$ is an affine function and according
to Proposition \ref{proposition1} its P-C form projects onto $J^{1}M$
if and only if the following equations hold:
\begin{equation*}
0=2\dfrac{\partial (L_{EH})_{rs}^{hi}}{\partial y_{ht,a}}
-\dfrac{\partial (L_{EH})_{ht}^{ai}}{\partial y_{rs,h}}
-\dfrac{\partial (L_{EH})_{ht}^{ah}}{\partial y_{rs,i}},
\end{equation*}
where
\begin{align}
\left( L_{EH}\right) _{rs}^{ij}
& =\tfrac{1}{2-\delta _{ij}}\frac{\partial L_{EH}}{\partial y_{rs,ij}}
\label{L_EH^ij_rs} \\
& =\tfrac{1}{1+\delta _{rs}}\rho \left(
y^{ir}y^{js}+y^{jr}y^{is}-2y^{rs}y^{ij}\right) , \notag
\end{align}
and the result follows immediately as $(L_{EH})_{rs}^{ij}$
does not depend on the variables $y_{ij,k}$.

In the present case, one has
\begin{align}
p_{kl}^{i}& =\sum_{r\leq s}
\left(
\frac{\partial ^{2}L_{0}}{\partial y_{rs,j}\partial y_{kl,i}}
-\frac{\partial (L_{EH})_{kl}^{ij}}{\partial y_{rs}}
-\frac{\partial (L_{EH})_{rs}^{ij}}{\partial y_{kl}}
\right) y_{rs,j}
\label{p^j_mr} \\
& =\sum_{r\leq s}Y_{kl}^{i;rs,j}y_{rs,j}, \notag
\end{align}
\begin{align}
Y_{kl}^{i;rs,j}
& =\tfrac{\rho }{(1+\delta _{kl})(1+\delta _{rs})}\left[
2y^{rs}y^{kl}y^{ij}-\left( y^{rk}y^{sl}+y^{rl}y^{sk}\right) y^{ij}\right.
\label{Y's} \\
& +\left( y^{sk}y^{lj}+y^{sl}y^{kj}\right) y^{ri}+\left(
y^{rk}y^{lj}+y^{rl}y^{kj}\right) y^{si}  \notag \\
& \left. -\left( y^{ki}y^{lj}+y^{li}y^{kj}\right) y^{rs}-\left(
y^{ri}y^{sj}+y^{rj}y^{si}\right) y^{kl}\right] ,  \notag
\end{align}
\begin{align}
H& =\rho \sum_{k\leq l}\sum_{r\leq s}\tfrac{1}{(1+\delta _{rs})
(1+\delta _{kl})}\left( -y^{ij}y^{kl}y^{rs}\right.  \label{HforEH} \\
& +y^{kl}\left( y^{ir}y^{js}+y^{is}y^{jr}\right)
+\tfrac{1}{2}y^{ij}\left(
y^{ks}y^{lr}+y^{kr}y^{ls}\right)  \notag \\
& \left. -\tfrac{1}{2}y^{ir}\left( y^{jl}y^{ks}+y^{jk}y^{ls}\right)
-\tfrac{1}{2}y^{is}\left( y^{jl}y^{kr}+y^{jk}y^{lr}\right)
\right) y_{rs,j}y_{kl,i}.
\notag
\end{align}
\begin{remark}
As a calculation shows, from the expression in \eqref{HforEH} for the
Hamiltonian of the E-L Lagrangian, for every $j^1_xg\in J^1M$ the following
formula holds true: $H(j^1_xg)=\rho (x) g^{ij}(x) ((\Gamma ^g)^r_{ij}(x)
(\Gamma ^g)^h_{hr}(x)-(\Gamma ^g)^r_{hi}(x)(\Gamma ^g)^h_{jr}(x))$.
Hence the function $H$---considered as a first-order
Lagrangian---not only provides the H-C equations for $\Lambda _{EH}$ but
also its own E-L equations, e.g., see \cite[3.3.1]{Carmeli}.
\end{remark}
\begin{theorem}[cf. \protect\cite{MJE}, \protect\cite{FF}, \protect\cite{EJ}]
\label{thEH} We have
\begin{enumerate}
\item[\emph{(i)}] With the natural identification
$V(p_M)\cong p_M^\ast S^2T^\ast N$, the bilinear form $
b_{\Lambda _{EH}}$ is defined on $p_M^\ast (T^\ast N\otimes S^2T^\ast N)$.
\item[\emph{(ii)}] The Lagrangian function $\bar{L}_{EH}$ defined in
\eqref{barL} coincides with the opposite to the Hamiltonian function.
\item[\emph{(iii)}] The E-H Lagrangian satisfies the regularity condition
of \emph{Corollary \ref{RegularityCondition}}.
\end{enumerate}
\end{theorem}
\begin{proof}
(i) From the formula
\begin{equation*}
\frac{\partial p_{\alpha }^{i}}{\partial y_{h}^{\beta }}
=\frac{\partial L_{\alpha }^{i0}}{\partial y_{h}^{\beta }}
-\frac{\partial L_{\beta }^{ih}}{\partial y^{\alpha }},
\end{equation*}
and \eqref{p^j_mr}, \eqref{Y's} it follows that the matrix
of $b_{\Lambda _{EH}}$ in the basis $(dx^{i})_{x}\otimes
(\partial /\partial y_{jk})_{g_{x}} $, $g_{x}\in p^{-1}(x)$,
$1\leq i\leq n$, $1\leq j\leq k\leq n$, at a point $j_{x}^{1}g$ is
\begin{equation*}
\left( (\partial p_{mr}^{j}/\partial y_{cd,h})(j_{x}^{1}g)
\right) _{c\leq d,h}^{m\leq r,j}
=\left( Y_{mr}^{j;cd,h}(g_{x})\right) _{c\leq d,h}^{m\leq r,j},
\end{equation*}
and one can conclude.

\medskip

\noindent (ii) It follows from the formulas \eqref{momenta}, \eqref{p^j_mr},
\eqref{HforEH} by means of a simple calculation.

\medskip

\noindent (iii) The proof is similar to that of Proposition 5.1 in \cite{MJE}, as
\begin{equation*}
\frac{\partial p_{mr}^{j}}{\partial y_{cd,h}}
=\frac{\partial ^{2}\bar{L}_{EH}}{\partial y_{mr,j}\partial y_{cd,h}}
=\frac{\partial ^{2}L^{\nabla }}{\partial y_{mr,j}\partial y_{cd,h}},
\end{equation*}
where $L^{\nabla }$ is the first-order Lagrangian variationally equivalent
to $L_{EH}$ introduced in \cite{MJE}.

In the present case, the equations \eqref{HCequations} become
\begin{equation*}
\left\{
\begin{array}{ll}
0=\dfrac{\partial (p_{kl}^i\circ j^1s)}{\partial x^i}
-\dfrac{\partial H}{\partial y_{kl}} \circ j^1s,
\smallskip
& 1\leq k\leq l\leq n, \\
0=\dfrac{\partial (y_{kl}\circ s)}{\partial x^i}
+\dfrac{\partial H}{\partial p_{kl}^i} \circ j^1s,
& 1\leq i\leq n, \; 1\leq k\leq l\leq n.
\end{array}
\right.
\end{equation*}
\end{proof}
\begin{remark}
By using the previous theorem, in \cite[Theorem 6.2]{EJ} the following
result has been obtained:

\textquotedblleft Given symmetric scalars $\gamma _{jk}^{i}=\gamma _{kj}^{i}$,
$i,j,k=1,\dotsc ,n$, there exists a Ricci-flat (pseudo-)Riemannian metric
$g$ of signature $(n^{-},n^{+})$ defined on a neighbourhood of $x_{0}\in N$
such that, $g_{ij}(x_{0})=\delta _{ij}$, $(\Gamma
^{g})_{jk}^{i}(x_{0})=\gamma _{jk}^{i}$, for all $i,j,k$.\textquotedblright
\end{remark}
\subsection{BF field theory\label{BFth}}
\label{BF} In this section we consider a new approach to BF Lagrangians
(cf.\ \cite{BGM}, \cite{CF}, \cite{FS}, \cite{Krasnov0}, \cite{Krasnov1})
generalizing the E-H functional.

Let $\pi \colon F(N)\to N$ be the principal $Gl(n,\mathbb{R})$-bundle of
linear frames on $N$. Given a metric $g$ on $N$, let $\pi _g\colon
F_g(N)\subset F(N)\to N$ be the subbundle of orthonormal linear frames with
respect to $g$, i.e., $u=(X_1,\dotsc,X_n)$ belongs to $F_g(N)$ if and only
if, $g(X_i,X_j)=\varepsilon _i\delta _{ij}$, with $\varepsilon _i=+1$ for
$1\leq i\leq n^+$ and $\varepsilon _i=-1$ for $1+n^+\leq i\leq n$. This is a
principal bundle with structure group the orthogonal group $O(n^+,n^-)$,
$n^++n^-=n$, associated to the quadratic form $q(x)=\sum _{a=1}^{n^+}(x^a)^2
-\sum _{b=n^++1}^{n^++n^-}(x^b)^2$.

By virtue of the symmetries of the curvature tensor $R^g$ of the Levi-Civita
connection of a metric $g$, for every $X,Y\in T_xN$ the endomorphism
$R^g(X,Y)$ takes values in the vector subspace of skew-symmetric linear
operators (with respect to $g_x$) in
$\operatorname{End}(T_xN)=T^\ast _xN\otimes T_xN $. More generally,
let $p_M\colon M\to N$ be the bundle of pseudo-Riemannian metrics
of signature $(n^+,n^-)$, and let
\begin{equation*}
\mathcal{A}(TN) \subset (p_M)^\ast \operatorname{End}(TN)
=M\times _N\operatorname{End}(TN)
\end{equation*}
be the vector subbundle of the pairs $(g_x,A)$, $g_x\in (p_M)^{-1}(x)$
and $A\in \operatorname{End}(T_xN)$, such that $g_x(AX,Y)+g_x(X,AY)=0$,
$\forall X,Y\in T_xN$; i.e., $A$ is skew-symmetric with respect to $g_x$.
Pulling $\mathcal{A}(TN)$ back along a metric $g$, understood as a smooth
section of  $p_M\colon M\to N$, one obtains the adjoint bundle
of the bundle of orthonormal frames with respect to $g$, i.e.,
the bundle associated to $F_g(N)$ under the adjoint representation
of $O(n^+,n^-)$ on its Lie algebra $\mathfrak{o}(n^+,n^-)$, i.e.,
$g^\ast \mathcal{A}(TN)=\operatorname{ad}F_g(N)
=(F_g(N)\times \mathfrak{o}(n^+,n^-))/O(n^+,n^-)$.

If $\beta $ is an $\mathcal{A}(TN)$-valued $p_{M}$-horizontal $(n-2)$-form
on $M$, then a second-order Lagrangian density $\Lambda _{\beta }$ is
defined on $J^{2}M$ by setting,
\begin{equation}
\left( \Lambda _{\beta }\right) _{j_{x}^{2}g}
=L_{\beta }(j_{x}^{2}g)v(x)
=\operatorname{trace}\left( \beta (g_{x})\wedge R^{g}(x)\right) ,
\label{Lambda_beta}
\end{equation}
where $R^{g}$ is considered as a $\operatorname{ad}F_{g}(N)$-valued
$2$-form on $N$. Locally,
\begin{align}
R^{g}& =\sum_{k<l}(R^{g})_{jkl}^{i}dx^{k}
\wedge dx^{l}\otimes dx^{j}
\otimes \frac{\partial }{\partial x^{i}},
\notag \\
\beta & =\sum_{k<l}\beta _{kl,j}^{i}v_{kl}
\otimes dx^{j}
\otimes \frac{\partial }{\partial x^{i}},
\quad \beta _{kl,j}^{i}\in C^{\infty }(M),
\label{beta_1}
\end{align}
where
$v_{kl}=dx^{1}\wedge \cdots \wedge \widehat{dx^{k}}
\wedge \cdots \wedge \widehat{dx^{l}}
\wedge \cdots \wedge dx^n $.
Here and below, we identify the vector space
$\operatorname{End}(T_{x}N)$
to $T_{x}^{\ast }N\otimes T_{x}N$ by agreeing that $w\otimes X$
is identified to the endomorphism given by, $(w\otimes X)(Y)=w(Y)X$,
$\forall X,Y\in T_{x}N$, $w\in T_{x}^{\ast }N$. Hence
\begin{equation}
L_{\beta }(j_{x}^{2}g)=\sum_{k<l}(-1)^{k+l+1}\beta
_{kl,j}^{i}(g_{x})(R^{g})_{ikl}^{j}(x).  \label{L_beta}
\end{equation}
If we set $\beta _{kl,i}^{j}=-\beta _{lk,i}^{j}$ for $k\geq l$, then,
as a calculation shows, the following local expression holds:
\begin{equation}
L_{\beta }=(-1)^{k+l+1}\beta _{kl,i}^{j}y^{ih}y_{hl,jk}+L_{\beta }^{0},
\label{L_beta_bis}
\end{equation}
with
\begin{align*}
L_{\beta }^{0}&
=\sum_{k\leq l}\sum_{r\leq s}
\tfrac{-1}{4(1+\delta _{kl})(1+\delta _{rs})}
\left\{
\left[
(-1)^{s}\beta _{st}^{kl}y^{tr}+(-1)^r\beta _{rt}^{kl}y^{ts}
\right]
y^{ij}\right. \\
& +\left[
(-1)^{j}\beta _{jt}^{li}y^{tr}+(-1)^r\beta _{rt}^{li}y^{tj}
\right]
y^{ks}+\left[
(-1)^{j}\beta _{jt}^{ki}y^{tr}+(-1)^r\beta _{rt}^{ki}y^{tj}
\right]
y^{ls} \\
& +\left[
(-1)^{j}\beta _{jt}^{li}y^{ts}+(-1)^{s}\beta _{st}^{li}y^{tj}
\right] y^{kr}
+\left[
(-1)^{j}\beta _{jt}^{ki}y^{ts}+(-1)^{s}\beta
_{st}^{ki}y^{tj}\right] y^{lr} \\
& -\left[
(-1)^{s}\beta _{st}^{li}y^{tr}+(-1)^r\beta _{rt}^{li}y^{ts}
\right] y^{kj}
-\left[
(-1)^{s}\beta _{st}^{ki}y^{tr}+(-1)^r\beta
_{rt}^{ki}y^{ts}\right] y^{lj} \\
& \left.
-\left[ (-1)^{k}\beta _{kt}^{rj}y^{tl}+(-1)^{l}\beta
_{lt}^{rj}y^{tk}\right] y^{is}
-\left[
(-1)^{k}\beta _{kt}^{sj}y^{tl}+(-1)^{l}\beta _{lt}^{sj}y^{tk}
\right] y^{ir}\right\} \\
& \qquad \qquad \qquad \qquad \qquad \qquad \qquad
\qquad \qquad \qquad \qquad \qquad \qquad \quad
\cdot y_{kl,i}y_{rs,j},
\end{align*}
where $\beta _{lt}^{jk}=(-1)^{k}\beta _{kl,t}^{j}+(-1)^{j}\beta _{jl,t}^{k}$,
and the equations $\beta _{ac,i}^{d}y^{ib}+\beta _{ac,i}^{b}y^{id}=0$ have
been used, which hold because $\beta $ takes values in $\mathcal{A}(TN)$.
\begin{remark}
Attached to each $\mathcal{A}(TN)$-valued $p_{M}$-horizontal $(n-2)$-form
$\beta $ on $M$ there exists a section $\tilde{\beta}$ of the vector bundle
$(p_{M})^{\ast }(\wedge ^{2}TN)\otimes \mathcal{A}(TN)$, given by
\begin{equation*}
\tilde{\beta}(g_{x})
=\beta (g_{x})\circ \left( \phi _{v}^{2}
\otimes \operatorname{id}_{\mathcal{A}(TN)}\right) ^{-1},
\quad \forall g_{x}\in M,
\end{equation*}
where $\phi _{v}^{2}$ is the isomorphism defined in \eqref{phi^k}.
If $\beta $ is locally given as in \eqref{beta_1}, then
\begin{equation*}
\begin{array}{l}
\tilde{\beta}(g_{x})=\sum_{k<l}(-1)^{k+l}\beta _{kl,j}^{i}(g_{x})
\left( \!
\dfrac{\partial }{\partial x^{k}}\!\right) _{x}\wedge
\left( \!\dfrac{\partial }{\partial x^{l}}\! \right) _{x}
\otimes (dx^{j})_{x}
\otimes \left( \!\dfrac{\partial }{\partial x^{i}}\! \right) _{x},
\smallskip \\
\forall g_{x}\in M.
\end{array}
\end{equation*}
If $\operatorname{sym}_{14}\colon \otimes ^{4}T_{x}N\to \otimes ^{4}T_{x}N$
is the symmetrization operator of the arguments $1$ and $4$, i.e.,
$\operatorname{sym}_{14}(X_{1}\otimes X_{2}\otimes X_{3}\otimes X_{4})
=X_{1}\otimes X_{2}\otimes X_{3}\otimes X_{4}
+X_{4}\otimes X_{2}\otimes X_{3}\otimes X_{1}$, for all $X_{i}\in T_{x}N$,
$1\leq i\leq 4$, and for every $p\geq 0$, $q\geq 1$, the symbol
$\sharp $ denotes the isomorphism
$\otimes ^{p+1}T_{x}^{\ast }N\otimes ^{q-1}T_{x}N
\to \otimes ^{p}T_{x}^{\ast }N\otimes ^{q}T_{x}N$
induced by the metric $g_{x}$, then
\begin{equation*}
\operatorname{sym}\nolimits_{14}\left( \tilde{\beta}^{\sharp }(g_{x})\right)
=(-1)^{l}\beta _{lj}^{ik}(g_{x})g^{jt}(x)\left( \frac{\partial }{\partial
x^{k}}\right) _{x}\otimes \left( \frac{\partial }{\partial x^{l}}\right)
_{x}\otimes \left( \frac{\partial }{\partial x^{t}}\right) _{x}\otimes
\left( \frac{\partial }{\partial x^{i}}\right) _{x},
\end{equation*}
and the formula \eqref{L_beta_bis} can be rewritten as,
$L_{\beta }=(-1)^{c+1}\beta _{ci}^{ab}y^{id}y_{ab,cd}+L_{\beta }^{0}$.
\end{remark}
\begin{theorem}
\label{theorem_beta} Let $\Lambda _{\beta }$ be the Lagrangian density
attached to a $\mathcal{A}(TN)$-valued $p_{M}$-horizontal $(n-2)$-form
$\beta $ as defined in \emph{\eqref{Lambda_beta}}. Then
\begin{enumerate}
\item[\emph{(i)}] The Lagrangian function \emph{\eqref{L_beta}} coincides
with the E-H Lagrangian (i.e., $L_{\beta }=L_{EH}$) if and only if the form
$\beta $ is given by,
\begin{equation}
\left( \beta _{EH}\right) _{kl,i}^{j}
=(-1)^{k+l+1}\rho \left( \delta ^{ik}y^{jl}-\delta ^{il}y^{jk}\right) ,
\label{betaEH}
\end{equation}
where the function $\rho $ is defined in \eqref{rho}.
\item[\emph{(ii)}] With the natural identification
$V(p_M)\cong p_M^\ast S^2T^\ast N$,
the bilinear form $b_{\Lambda _\beta }$ is defined on
$p_M^\ast (T^\ast N\otimes S^2T^\ast N)$.
\item[\emph{(iii)}] The E-L equations for the Lagrangian density
$\Lambda_\beta $ are the following:
\begin{equation}  \label{EL_eqs_beta}
g^\ast (d_{M/N}\beta )\barwedge R^g
+\operatorname{sym}_{12}\circ d^{\nabla ^g}
\left( \omega _{n-1}(g,\beta ) \right) =0,
\end{equation}
where,
\begin{itemize}
\item $\nabla ^g$ is the covariant differentiation with respect to the
Levi-Civita connection of a section $g$ of the bundle $p_M\colon M\to N$.
\item The fibre differential $d_{M/N}\beta $ is understood to be a section
of the vector bundle $(p_M)^\ast ((S^2TN) \otimes \wedge ^{n-2}T^\ast N
\otimes \operatorname{End}(TN))$, taking the isomorphism
$V^\ast (p_M)\cong (p_{M})^\ast (S^2TN)$ into account.
\item $g^\ast (d_{M/N}\beta )\barwedge R^g$ is the $S^2TN$-valued $n$-form
on $N$ defined by,
{\small
\begin{equation*}
\begin{array}{l}
\left( g^\ast (d_{M/N}\beta )\barwedge R^g \right) \left(
w_1,w_2,X_1,\dotsc,X_n \right) \medskip \\
=\sum _{k<l}(-1)^{k+l+1}\cdot \smallskip \\
\quad \operatorname{trace} \left\{ g^{\ast}(d_{M/N}\beta) \left( w_1,w_2,X_1,
\dotsc,\widehat{X_k}, \dotsc,\widehat{X_l}, \dotsc,X_n \right) \circ R^g(X_k,X_l)
\right\} , \\
\forall X_1,\dotsc,X_n\in T_xN, \forall w_1,w_2\in T_x^\ast N.
\end{array}
\end{equation*}
}
\item $\omega _{n-1}(g,\beta )$ is the $(TN\otimes TN)$-valued $(n-1)$-form
on $N$ given by,
\begin{equation*}
\omega _{n-1}(g,\beta ) =\left( (\phi _v^1)^{-1} \otimes \operatorname{id}_{TN}
\otimes \phi_v^1 \right) \left( d^{\nabla ^g}(g^\ast \beta)^\sharp \right) ,
\end{equation*}
$\phi _v^1$ being defined in the formula \eqref{phi^k}.
\item $\operatorname{sym}_{12}\colon \otimes ^2TN\to S^2TN$ denotes the
symmetrization operator.
\end{itemize}
\end{enumerate}
\end{theorem}
\begin{proof}
(i) By comparing the formula \eqref{L_beta} with the following:
\begin{equation*}
L_{EH}(j_x^2g)=\sum _{k<l}\rho (x) \left( \delta ^{ik}g^{jl}(x) -\delta
^{il}g^{jk}(x) \right) (R^{g})_{jkl}^i(x),
\end{equation*}
we obtain \eqref{betaEH} directly. \medskip

\noindent (ii) As a calculation shows, the matrix of $b_{\Lambda _\beta }$
is given as follows:
\begin{equation*}
\begin{array}{l}
\left( \digamma _\beta \right) _{r\leq s;i,a\leq b,j} \medskip \\
=\dfrac{\partial ^2L_\beta ^0}{\partial y_{rs,i}\partial y_{ab,j}}
-\dfrac{\partial L_{ab}^{ij}}{\partial y_{rs}}
-\dfrac{\partial L_{rs}^{ij}}{\partial y_{ab}}
\medskip \\
=\tfrac{1}{2}\tfrac{1}{1+\delta_{ab}}\tfrac{1}{1+\delta_{rs}}
\left\{ -\left[
(-1)^a\beta_{at}^{rs}y^{tb} +(-1)^b\beta_{bt}^{rs}y^{ta}
\right]
y^{ij}\right.
\smallskip \\
+\left[
(-1)^j\beta_{jt}^{rs}y^{tb} +(-1)^b\beta_{bt}^{rs}y^{tj}
\right] y^{ia}
+\left[
(-1)^a\beta_{at}^{rs}y^{tj} +(-1)^j\beta_{jt}^{rs}y^{ta}
\right] y^{ib}
\smallskip \\
+\left[ (-1)^i\beta_{it}^{ab}y^{ts} +(-1)^s\beta_{st}^{ab}y^{ti}
\right] y^{rj}
+\left[ (-1)^i\beta_{it}^{ab}y^{tr} +(-1)^r\beta_{rt}^{ab}y^{ti}
\right] y^{sj}
\smallskip \\
-\left[ (-1)^b\beta_{bt}^{is}y^{tj} +(-1)^j\beta_{jt}^{is}y^{tb}
\right] y^{ra}
-\left[ (-1)^b\beta_{bt}^{ir}y^{tj} +(-1)^j\beta_{jt}^{ir}y^{tb}
\right] y^{sa}
\smallskip \\
-\left[ (-1)^a\beta_{at}^{is}y^{tj} +(-1)^j\beta_{jt}^{is}y^{ta}
\right] y^{rb}
-\left[ (-1)^a\beta_{at}^{ir}y^{tj} +(-1)^j\beta_{jt}^{ir}y^{ta}
\right] y^{sb}
\smallskip \\
-(-1)^a\beta_{at}^{ij} \left( y^{tr}y^{bs}+y^{ts}y^{br}
\right)
-(-1)^b\beta_{bt}^{ij} \left( y^{tr}y^{as}+y^{ts}y^{ar} \right)
\smallskip \\
-(-1)^r\beta_{rt}^{ij} \left( y^{ta}y^{bs}+y^{tb}y^{as} \right)
-(-1)^s\beta_{st}^{ij} \left( y^{ta}y^{br}+y^{tb}y^{ar} \right)
\smallskip \\
+\left( 1+\delta_{rs} \right)
\left(
(-1)^a\dfrac{\partial \beta_{at}^{ij}}{\partial y_{rs}} y^{tb}
+(-1)^b\dfrac{\partial\beta_{bt}^{ij}}{\partial y_{rs}} y^{ta}
\right)
\smallskip \\
\left. +\left( 1+\delta_{ab} \right)
\left( (-1)^r\dfrac{\partial \beta_{rt}^{ij}} {\partial y_{ab}}y^{ts}
+(-1)^s\dfrac{\partial \beta_{st}^{ij}} {\partial y_{ab}}y^{tr}
\right)
\right\} ,
\end{array}
\end{equation*}
thus proving the statement.

\noindent (iii) The E-L equations for the Lagrangian density
$\Lambda _{\beta }=L_{\beta }v$ are straightforwardly computed,
thus obtaining,
\begin{multline*}
\mathcal{E}^{ab}(L_{\beta })\circ j^{2}g
=\tfrac{1}{2}(-1)^{k+l+1}\left( \!\frac{\partial \beta _{kl,i}^{j}}
{\partial y_{ab}}\!\circ \!g\!\right) (R^{g})_{jkl}^{i} \\
-\tfrac{1}{1+\delta _{ab}}\left\{ \frac{\partial }{\partial x^r}
\left[ (-1)^{a}\Phi _{a}^{rb}+(-1)^{b}\Phi _{b}^{ra}\right]
+(-1)^{l}\left[ \Phi _{l}^{rb}(\Gamma ^{g})_{rl}^{a}
+\Phi _{l}^{ra}(\Gamma ^{g})_{rl}^{b}\right] \right\} ,
\end{multline*}
for $1\leq a\leq b\leq n$, where
\begin{equation*}
\Phi _{a}^{rb}=\sum_{k}(-1)^{k}
\left( -\frac{\partial (\beta \circ g)_{ka,i}^{b}}{\partial x^{k}}
+\left( \beta \circ g\right) _{ka,m}^{b}(\Gamma ^{g})_{ki}^{m}
-\left( \beta \circ g\right) _{ka,i}^{m}(\Gamma ^{g})_{km}^{b}
\right) g^{ri}.
\end{equation*}
Moreover, the following local expressions are deduced:
\begin{eqnarray*}
g^{\ast }(d_{M/N}\beta )\barwedge R^{g}
&=&\tfrac{1}{2}(-1)^{k+l+1}\sum_{a\leq b}
\left(
\!\frac{\partial \beta _{kl,i}^{j}}{\partial y_{ab}}\!\circ\!g\!
\right)
(R^{g})_{jkl}^{i}\frac{\partial }{\partial x^{a}}\odot
\frac{\partial }{\partial x^{b}}\otimes v, \\
\left( d^{\nabla ^{g}}(g^{\ast }\beta )\right) ^{\sharp }
&=&\sum_{l}\Phi _{l}^{ab}v_{l}\otimes
\frac{\partial }{\partial x^{a}}
\otimes \frac{\partial }{\partial x^{b}},
\end{eqnarray*}
from which the result follows.
\end{proof}

\begin{corollary}
A flat metric $g$ is a solution to the equations \eqref{EL_eqs_beta} if and
only if the form $\beta $ in \emph{\eqref{beta_1}} satisfies the following
equation:
\begin{equation*}
c_{12}^{23} \left[ \left( \nabla ^g \right) ^2 \left\{ \operatorname{sym}_{14}
\left( \tilde{\beta }^{\sharp }\circ g \right) \right\} \right] =0,
\end{equation*}
where $c_{12}^{23}\colon \otimes ^2T^\ast M\otimes ^4TM \to \otimes ^2TM$
denotes the contraction operator of the first covariant index with the
second contravariant one, and the second covariant index with the third
contravariant one.
\end{corollary}
\begin{remark}
The geometric construction of the form \eqref{betaEH} is as follows: Given
an arbitrary system $X_1, \dotsc,X_{n-2}\in T_xN$, we must define a
skew-symmetric (with respect to $g_x$) endomorphism $\beta (g_x)(X_1,
\dotsc,X_{n-2})\colon T_xN\to T_xN$.

If the given system is linearly dependent, then $\beta
(g_x)(X_1,\dotsc,X_{n-2})=0$. We assume: i) The system $(X_1,
\dotsc,X_{n-2}) $ is linearly independent. Hence its orthogonal $\Pi
=\left\langle X_1,\dotsc,X_{n-2} \right\rangle ^\bot $ is a subspace of
dimension $2$ in $T_xN$, ii) the subspace $\left\langle X_1,\dotsc,X_{n-2}
\right\rangle $ is not singular with respect to $g_x$. Hence
\begin{equation*}
T_xN=\Pi \oplus \left\langle X_1,\dotsc,X_{n-2} \right\rangle ,
\end{equation*}
and $\Pi $ is also non-singular. Let
$(n^+(\Pi ),n^-(\Pi )) \in \{ (2,0),(1,1),(0,2)\} $ be its signature
and let
\begin{equation*}
\left(
\begin{array}{cc}
\varepsilon _1(\Pi ) & 0 \\
0 & \varepsilon _2(\Pi )
\end{array}
\right) ,
\quad
\left( \varepsilon _1(\Pi ),\varepsilon _2(\Pi )\right)
\in
\{ (1,1),(1,-1),(-1,-1)\} ,
\end{equation*}
be the matrix of $g_x$ in an orthonormal basis
$(Y_1,Y_2)$ of $\Pi $, which, in addition, is assumed to satisfy
the following: $v(X_1,\dotsc,X_{n-2},Y_1,Y_2)>0$.
If $Z_j=b_j^iY_i$, $i,j=1,2$, is another
orthonormal basis with $v(X_1,\dotsc,X_{n-2},Z_1,Z_2)>0$, then
$\det (b_j^i)=1$. Hence $(b_j^i)$ belongs to $SO(n^+(\Pi ),n^-(\Pi ))$,
and the endomorphism $J_\Pi ^{g_x}\colon \Pi \to \Pi $ given by
$J_\Pi ^{g_x}(Y_1)=\varepsilon _1(\Pi )Y_2$,
$J_\Pi ^{g_x}(Y_2)=-\varepsilon _{2}(\Pi )Y_1$, is independent
of the basis chosen (as $SO(n^+(\Pi ),n^-(\Pi ))$ is commutative)
and skew-symmetric. We define
$\tilde{J}_\Pi ^{g_x}\colon T_{x}N\to T_xN$ by setting,
$\tilde{J}_\Pi ^{g_x}|_\Pi =J_\Pi ^{g_x}$,
$\tilde{J}_\Pi ^{g_x}|_{\left\langle X_1,\dotsc,X_{n-2}\right\rangle } =0$.
Finally,
\begin{equation*}
\left( \beta _{EH} \right) (g_x)(X_1,\dotsc,X_{n-2})
=\det \left( g(Y_a,Y_b)
\right) _{a,b=1}^2v_{g_x} \left(X_1,\dotsc,X_{n-2},Y_1,Y_2 \right)
\tilde{J}_\Pi ^{g_x}.
\end{equation*}
\end{remark}
\begin{remark}
The bilinear form $b_{\Lambda _\beta }$ is identified to a section of the
vector bundle $p_M^\ast \left( (TN\otimes S^2TN)\otimes (TN\otimes S^2TN)
\right) $, and the following formula holds:
\begin{align*}
b_{\Lambda _\beta } & =\tfrac{1}{2}\operatorname{sym}_{45}
\left( \operatorname{alt}_{46}\left[ \operatorname{sym}_{12}(\hat{\beta })
+\operatorname{alt}_{13}(\hat{\beta })
-\hat{\beta } \right] \right)
-\tfrac{1}{2}\left[
\operatorname{sym}_{12}(\hat{\beta }) +\operatorname{alt}_{13}(\hat{\beta })
-\hat{\beta} \right] \\
& +\tfrac{1}{2}\operatorname{sym}_{(12),(45)} d_{M\!/\!N}
\left( \operatorname{sym}_{23}
\left( \operatorname{sym}_{14}(\tilde{\beta }^\sharp )
\right)
\right) ,
\end{align*}
where the operators $\operatorname{alt}_{ij}, \operatorname{sym}_{ij},
\operatorname{sym}_{(1,2)(4,5)} \colon \otimes ^6T_xN\to \otimes ^6T_xN$,
$1\leq i<j\leq 6$, are defined as follows:
\begin{equation*}
\begin{array}{l}
\operatorname{alt}_{ij}(X_1 \otimes \cdots \otimes X_i\otimes \cdots
\otimes X_j\otimes \cdots \otimes X_6)= \\
X_1\otimes \cdots \otimes X_i \otimes \cdots \otimes X_j\otimes \cdots
\otimes X_6-X_1 \otimes \cdots \otimes X_j \otimes \cdots
\otimes X_i \otimes \cdots \otimes X_6, \medskip \\
\operatorname{sym}_{ij}(X_1 \otimes \cdots \otimes X_i
\otimes \cdots \otimes X_j \otimes\cdots \otimes X_6)= \\
X_1 \otimes \cdots \otimes X_i \otimes \cdots \otimes X_j
\otimes \cdots \otimes X_6+X_1 \otimes \cdots \otimes X_j
\otimes \cdots \otimes X_i \otimes \cdots \otimes X_6,
\medskip \\
\operatorname{sym}_{(1,2)(4,5)}(X_1 \otimes \cdots \otimes X_6)= \\
X_1\otimes \cdots \otimes X_6+X_4\otimes X_5 \otimes X_3
\otimes X_1\otimes X_2\otimes X_6, \medskip \\
X_1,\dotsc,X_6\in T_xN,
\end{array}
\end{equation*}
and the contravariant $6$-tensor $\hat{\beta }$ is given by,
\begin{equation*}
\hat{\beta }=\operatorname{sym}_{15}
\left[
\operatorname{sym}_{23}
\left(
\operatorname{sym}_{14} (\tilde{\beta }^\sharp )
\right)
\otimes(g^\sharp )^\sharp \right] -\operatorname{sym}_{23}
\left(
\operatorname{sym}_{14} (\tilde{\beta }^\sharp )
\right)
\otimes (g^\sharp )^\sharp .
\end{equation*}
\end{remark}
\begin{remark}
If $\beta =\beta _{EH}$ in Theorem \ref{theorem_beta}-(iii), then the
functions $\Phi _a^{rb} $ (appearing in the proof) vanish, and the equations
\eqref{EL_eqs_beta} reduce to Einstein's vacuum equations for arbitrary
signature.
\end{remark}
\section{First-order equivalent Lagrangians\label{first_equiv}}
\begin{theorem}
\label{first-order-equivalent} Let $\Lambda =Lv$ be a second-order
Lagrangian density on $p\colon E\to N$ whose Poincar\'e-Cartan form projects
onto $J^1E$. We have
\begin{enumerate}
\item[\emph{(i)}] The H-C equations of the first-order Lagrangian $\bar{L}v$
given in \emph{\eqref{barL}} coincide locally with the H-C equations
of $\Lambda $. Furthermore, if $\bar{L}^{\prime }$ is another first-order
Lagrangian fulfilling this property, then $\bar{L}^{\prime }v
-\bar{L}v=D\alpha _{n-1}$, where $D$ denotes the horizontal exterior
derivative and $\alpha _{n-1}$ is a $p$-horizontal $(n-1)$-form on $E$.

\item[\emph{(ii)}] The E-L equations of $\Lambda $, considered as a
second-order partial differential system, satisfy the Helmholtz conditions.

\item[\emph{(iii)}] The E-L equations of the first-order Lagrangian
$\bar{L}v $ above coincide with E-L equations of $\Lambda $.

\item[\emph{(iv)}] Let $\phi _v^1$ be the isomorphism defined in
\emph{\eqref{phi^k}} for $k=1$ and let $w^{0,\sigma }_L$ be the $TN$-valued
section on $J^1E$ defined as in \emph{Proposition \ref{proposition2}}. The
composite mapping $\phi ^1_v\circ w^{0,\sigma }_L$ can be viewed as a
$p^1$-horizontal $(n-1)$-form on $J^1E$ and the difference $\bar{L}_\sigma v
=Lv-D(\phi ^1_v\circ w^{0,\sigma }_L)$ determines a globally defined
first-order Lagrangian which is variationally equivalent to $Lv$, but this
is not canonically attached to $Lv$ as it depends on the section $\sigma $.
\end{enumerate}
\end{theorem}
\begin{proof}
(i) Locally, the Hamiltonian and the momenta associated to $\bar{L}$ are
given respectively by (cf.\ formula \eqref{ThetaLambdaTilde1} in
Remark \ref{remark.4.1}),
\begin{equation*}
\bar{H}=\bar{L}-y_i^\alpha \frac{\partial \bar{L}}{\partial y_i^\alpha },
\quad \bar{p}_\alpha ^i =\frac{\partial \bar{L}}{\partial y_i^\alpha }.
\end{equation*}
By comparing the H-C equations for $\bar{L}$ with the H-C equations for $L$
given in \eqref{HCequations}, one obtains, $H=\bar{H}$ and
$p_\alpha ^i=\bar{p}_\alpha ^i$. Hence
\begin{eqnarray}
L_0-y_i^\alpha L_\alpha ^{i0} -\frac{\partial L^i}{\partial x^i}
& = & \bar{L}-y_i^\alpha \frac{\partial \bar{L}}{\partial y_i^\alpha },
\label{1st} \\
L_\alpha ^{i0} -\frac{\partial L^i}{\partial y^\alpha }
& = & \frac{\partial \bar{L}}{\partial y_i^\alpha }.
\label{2nd}
\end{eqnarray}
Replacing \eqref{2nd} into \eqref{1st}, one concludes that $\bar{L}$ is
given as in the formula \eqref{barL}. Moreover, if $\bar{L}^\prime $ is the
first-order Lagrangian associated to other primitive functions $L^{\prime
i}=L^i+A^i$, $A^i\in C^\infty (E)$, according to Proposition \ref{proposition2},
then $\bar{L}^\prime =\bar{L}-D_iA^i$.

\medskip \noindent (ii) As a simple---although rather long---computation
shows, the second-order differential operator $\mathcal{E}_\alpha
(L)dy^\alpha \wedge v$ satisfies the equations (1.5a), (1.5b), and (1.5c) in
\cite{AD}. In fact, by using the formulas \eqref{f11}, \eqref{f12}, and
\eqref{L^i}, the following equations are checked:
\begin{equation*}
\begin{array}{ll}
\text{(1.5a)} & 0=
\dfrac{\partial \mathcal{E}_\alpha (L) }{\partial y_{(ij)}^\sigma }
-\dfrac{\partial \mathcal{E}_\sigma (L)} {\partial y_{(ij)}^\alpha },
\medskip \\
\text{(1.5b)} & 0=
\dfrac{\partial \mathcal{E}_\alpha (L)}{\partial y_i^\sigma }
+\dfrac{\partial \mathcal{E}_\sigma (L)} {\partial y_i^\alpha }
-(1+\delta _{ij})D_j \left(
\dfrac{\partial \mathcal{E}_\sigma (L)}{\partial y_{(ij)}^\alpha }
\right) ,
\medskip \\
\text{(1.5c)} & 0=
\dfrac{\partial \mathcal{E}_\alpha (L)} {\partial y^\sigma }
-\dfrac{\partial \mathcal{E}_{\sigma }(L)} {\partial y^\alpha }
+D_{i}\left(
\dfrac{\partial \mathcal{E}_\sigma (L)} {\partial y_i^\alpha }
\right) -\sum_{i\leq j}D_iD_j
\left(
\dfrac{\partial \mathcal{E}_\sigma (L)} {\partial y_{(ij)}^\alpha }
\right) .
\end{array}
\end{equation*}
(iii) From the formula \eqref{barL}, it follows that the Lagrangian $\bar{L}$
can also be written as $\bar{L}=L-D_iL^i$, thus proving that $L$ and $\bar{L}
$ differ on a total divergence and, hence
$\mathcal{E}_\alpha (L)=\mathcal{E}_\alpha (\bar{L})$.

\medskip \noindent (iv) Locally,
$w^{0,\sigma }_L =L_\sigma ^i\partial /\partial x^i$;
hence $\phi ^1_v\circ w^{0,\sigma }_L =(-1)^{i-1}L_\sigma ^iv_i$,
and consequently, $D(\phi ^1_v\circ w^{0,\sigma }_L)=(D_iL_\sigma
^i)v$. The result thus follows from $\bar{L}_\sigma =L-D_iL_\sigma ^i$
in item (iii).
\end{proof}
\begin{remark}
As is known (e.g., see \cite[(2.21)--(2.25)]{Grigore}), the Vainberg-Tonti
Lagrangian $L_{VT}$ attached to a second-order affine Lagrangian as
in \eqref{affine} is also affine, say $L_{VT}=(L_{VT})_0+(L_{VT})_1$, with
$(L_{VT})_1=(L_{VT})_\alpha ^{ij}y^\alpha _{(ij)}$. Then, as a computation
shows, one has
\begin{equation*}
L_{VT}-\bar{L}=-D_h \left( \int _0^1y^\alpha
\left(
\frac{\partial \bar{L}}{\partial y^\alpha _h} \circ \chi _\lambda
\right) d\lambda \right) ,
\end{equation*}
where $\chi _\lambda(x^i,y^\alpha,y^\alpha_i)
=(x^i,\lambda y^\alpha,\lambda y^\alpha_i)$, but it should be noted
that the Vainberg-Tonti Lagrangian is of second order in the general
case; e.g., if $L(x,y,\dot{y},\ddot{y})
=L_1(x,y,\dot{y})\ddot{y}+L_0(x,y,\dot{y})$, then
$L_{VT}=(L_{VT})_0+(L_{VT})_1\ddot{y}$, with
\begin{align*}
\left( L_{VT} \right) _1 &
=y\int _0^1
\left\{ 2\lambda \frac{\partial L_1}{\partial y}(x,\lambda y,\lambda \dot{y})
+\lambda
\frac{\partial ^2L_1}{\partial x\partial \dot{y}}(x,\lambda y,\lambda \dot{y})
\right. \\
& \qquad \qquad
\left.
+\lambda ^2\dot{y}
\frac{\partial ^2L_1}{\partial y\partial \dot{y}}(x,\lambda y,\lambda \dot{y})
-\lambda\frac{\partial ^2L_0}{\partial \dot{y}^2} (x,\lambda y,\lambda \dot{y})
\right\} d\lambda ,
\end{align*}
\begin{align*}
\left( L_{VT} \right) _0
& =y\int _0^1 \left\{
\frac{\partial L_0}{\partial y} (x,\lambda y,\lambda \dot{y})
+\frac{\partial ^2L_1}{\partial x^2}
(x,\lambda y,\lambda \dot{y})
\right. \\
& \qquad \qquad
+\lambda ^2(\dot{y})^2 \frac{\partial ^2L_1}{\partial y^2}
(x,\lambda y,\lambda \dot{y}) +2\lambda \dot{y}
\frac{\partial ^2L_1}{\partial x\partial y}(x,\lambda y,\lambda \dot{y}) \\
& \qquad \qquad
\left. -\frac{\partial ^2L_0}{\partial x\partial \dot{y}}
(x,\lambda y,\lambda \dot{y})
-\lambda \dot{y}
\frac{\partial ^2L_0}{\partial y\partial \dot{y}}(x,\lambda y,\lambda \dot{y})
\right\} d\lambda.
\end{align*}
Therefore $L_{VT}$ is of second order, except when $(L_{VT})_1=0$, and this
latter condition is seen to be equivalent to the following:
\begin{equation*}
0=2\frac{\partial L_{1}}{\partial y}
+\frac{\partial ^2L_1}{\partial x\partial \dot{y}}
+\dot{y} \frac{\partial ^2L_1}{\partial y\partial \dot{y}}
-\frac{\partial ^2L_0}{\partial \dot{y}^2}.
\end{equation*}
\end{remark}
In the particular case of the bundle of metrics, there exists a more
specific way to obtain a section $\sigma $ of $p^1_0\colon M\to J^1M$ than
the procedure suggested in Remark \ref{remark_sigma}, which depends on a
linear connection only rather than a non-linear connection; namely,
\begin{lemma}
\label{sigma_nabla} Let $p_M\colon M\to N$ be the bundle of
pseudo-Riemannian metrics of a given signature $(n^+,n^-)$, $n^++n^-=n$, and
let $\nabla $ be a symmetric linear connection on $N$. For every
$g_x\in (p_M)^{-1}(x)$, there exists a unique $1$-jet of metric
$j^1_x\tilde{g}\in J^1_xM$ such that, \emph{1)} $\tilde{g}_x=g_x$, and
\emph{2)} $(\nabla \tilde{g})_x=0$. The mapping
$\sigma ^\nabla \colon M\to J^1M$ given by $\sigma ^\nabla (g_x) =j^1_x\tilde{g}$
is a section of $p^1_0\colon J^1M\to M$.
\end{lemma}
\begin{proof}
If $\Gamma ^i_{jk}$ are the local symbols of $\nabla $ in a coordinate
system, then as a calculation shows, the condition 2)---assuming 1)---of the
statement is equivalent to,
\begin{equation*}
\frac{\partial \tilde{g}_{ij}}{\partial x^k}(x)= \Gamma _{ik}^h(x)g_{hj}(x)
+\Gamma _{jk}^h(x)g_{hi}(x),
\end{equation*}
thus proving that $\sigma ^\nabla $ makes sense.
\end{proof}
\begin{proposition}[{cf. \protect\cite[II]{MJE}}]
Let $p_{M}\colon M\to N$ be as in \emph{Lemma \ref{sigma_nabla}}.
For the E-H Lagrangian, the density $(\bar{L}_{EH})_{\sigma ^{\nabla }}v$
introduced in \emph{Theorem \ref{first-order-equivalent}-(iv)} is given by,
$(\bar{L}_{EH})_{\sigma ^{\nabla }}(j_{x}^{2}g)v_{x}
=c\left( \left( \operatorname{alt}_{23}
\left( \nabla ^{g}T^{g}\right) _{x}
\right) ^{\sharp }
\right)
\left(
v_{g}\right) _{x}$, for all $j_{x}^{2}g\in J^{2}M$, where
\begin{equation*}
\operatorname{alt}_{23}
\colon \otimes ^{3}T^{\ast }M\otimes TM\to
\otimes ^{3}T^{\ast }M\otimes TM
\end{equation*}
denotes the alternating operator of the second and third covariant indices,
and
\begin{equation*}
{}^{\sharp }\colon \otimes ^{3}T^{\ast }M\otimes TM
\to \otimes ^{2}T^{\ast }M\otimes ^{2}TM
\end{equation*}
is the isomorphism induced by $g$, i.e.,
\begin{equation*}
w_{1}\otimes w_{2}\otimes w_{3}\otimes X
\mapsto w_{1}\otimes w_{2}\otimes (w_{3})^{\sharp }\otimes X,
\end{equation*}
and $c\colon \otimes ^{2}T^{\ast }M\otimes ^{2}TM\to \mathbb{R}$ is
the total contraction of the first (resp.\ second) covariant index with the
first (resp.\ second) contravariant one.
\end{proposition}
\section{Symmetries and Noether invariants\label{symmetriesNoether}}
Given fibred manifolds $p\colon E\to N$,
$p^{\prime }\colon E^{\prime }\to N^{\prime }$, every morphism
$\Phi \colon E\to E^{\prime }$ for which the associated map
on the base manifolds $\phi \colon N\to N^{\prime }$ is a diffeomorphism,
induces a map
\begin{equation*}
\begin{array}{l}
\Phi ^{(r)}\colon J^rE\to J^rE^{\prime }, \\
\Phi ^{(r)}(j_{x}^rs)=j_{\phi (x)}^r(\Phi \circ s\circ \phi ^{-1}).
\end{array}
\end{equation*}
If $\Phi _{t}$ is the flow of a $p$-projectable vector field $X$, then
$\Phi _{t}^{(r)}$ is the flow of a vector field $X^{(r)}\in \mathfrak{X}(J^rE)$,
called the infinitesimal contact transformation of order $r$ associated to
the vector field $X$. The mapping $X\mapsto X^{(r)}$ is an injection of Lie
algebras. For $r=1,2$, the general prolongation formulas read as follows:
\begin{equation*}
\begin{array}{l}
X=u^{i}\frac{\partial }{\partial x^{i}}+v^{\alpha }
\frac{\partial }{\partial y^{\alpha }}, \\
u^{i}\in C^{\infty }(N),v^{\alpha }\in C^{\infty }(E),
\end{array}
\end{equation*}
\begin{equation*}
\begin{array}{l}
X^{(1)}=u^{i}\frac{\partial }{\partial x^{i}}+v^{\alpha }
\frac{\partial }{\partial y^{\alpha }}+v_{i}^{\alpha }
\frac{\partial }{\partial y_{i}^{\alpha }}, \\
v_{i}^{\alpha }=D_{i}\left( v^{\alpha }-u^{h}y_{h}^{\alpha }\right)
+u^{h}y_{(hi)}^{\alpha },
\end{array}
\end{equation*}
\begin{equation*}
\begin{array}{l}
X^{(2)}=u^{i}\frac{\partial }{\partial x^{i}}
+v^{\alpha }\frac{\partial }{\partial y^{\alpha }}
+v_{i}^{\alpha }\frac{\partial }{\partial y_{i}^{\alpha }}
+\sum\nolimits_{i\leq j}v_{ij}^{\alpha }
\frac{\partial }{\partial y_{(ij)}^{\alpha }}, \\
v_{ij}^{\alpha }
=D_{i}D_{j}\left( v^{\alpha }-u^{h}y_{h}^{\alpha }\right)
+u^{h}y_{(hij)}^{\alpha }.
\end{array}
\end{equation*}
\begin{theorem}
\label{symmetry1}Let $\Lambda =Lv$ be a second-order Lagrangian density
on $p\colon E\to N$ with P-C form projectable onto $J^{1}E$. If $X$ is a
$p$-projectable vector field on $E$, then the P-C form of the second-order
Lagrangian density $\Lambda ^{\prime }=L^{\prime }v=L_{X^{(2)}}\Lambda $
also projects onto $J^{1}E$ and the following formula holds:
\begin{equation*}
\Theta _{L_{X^{(2)}}\Lambda }=L_{X^{(1)}}\Theta _{\Lambda }.
\end{equation*}
Therefore, if $s\colon N\to E$ is an extremal for $\Lambda $ and $X$
is an infinitesimal symmetry (i.e., $L_{X^{(2)}}\Lambda =0$), then the
$(n-1) $-form $(j^{1}s)^{\ast }i_{X^{(1)}}\Theta $ is closed.
(The $(n-1)$-form $i_{X^{(1)}}\Theta $ is called the Noether invariant
associated to $X$.)
\end{theorem}
\begin{proof}
We have $L^{\prime }=X^{(2)}(L)+\operatorname{div}(X^{\prime })L$, $X^{\prime }$
being the projection of $X$ onto $N$ and $\operatorname{div}(X^{\prime })$ the
divergence of $X^{\prime }$ with respect to $v$. According to Proposition
\ref{proposition1}, we must prove the existence of functions $L_{0}^{\prime
} $, $L_{\alpha }^{\prime ji}=L_{\alpha }^{\prime ij}$ on $J^{1}E$ such that,
\begin{equation*}
\begin{array}{rl}
L^{\prime }=L_{\alpha }^{\prime ij}y_{(ij)}^{\alpha }+L_{0}^{\prime },
& \smallskip \\
\tfrac{\partial L_{\beta }^{\prime ih}}{\partial y_{a}^{\alpha }}
=\tfrac{\partial L_{\alpha }^{\prime ia}}{\partial y_{h}^{\beta }},
& a,h,i=1,\dotsc,n,\;\alpha ,\beta =1,\dotsc ,m.
\end{array}
\end{equation*}
As $L$ satisfies such formulas by virtue of the hypothesis, and $X^{(2)}$
projects onto $X^{(1)}$, we have
\begin{equation}
\begin{array}{rl}
L^{\prime }
= & \left[ X^{(1)}\left( L_{\alpha }^{ab}\right)
+\tfrac{\partial v^{\beta }}{\partial y^{\alpha }}L_{\beta }^{ab}
-2\tfrac{\partial u^{a}}{\partial x^r}L_{\alpha }^{br}
+\operatorname{div}(X^{\prime })L_{\alpha }^{ab}\right] y_{(ab)}^{\alpha } \\
& \multicolumn{1}{r}{+L_{\alpha }^{ij}
\left(
\tfrac{\partial ^{2}v^{\alpha }}{\partial x^{i}\partial x^{j}}
+y_{j}^{\beta }
\tfrac{\partial ^{2}v^{\alpha }}{\partial x^{i}\partial y^{\beta }}
+y_{j}^{\beta }y_{i}^{\gamma }\tfrac{\partial ^{2}v^{\alpha }}
{\partial y^{\beta }\partial y^{\gamma }}
-\tfrac{\partial ^{2}u^{h}}{\partial x^{i}\partial x^{j}}y_{h}^{\alpha }
\right)} \\
& \multicolumn{1}{r}{+X^{(1)}\left( L_{0}\right)
+\operatorname{div}(X^{\prime })L_{0}.}
\end{array}
\label{L^prime}
\end{equation}
Hence
\begin{equation}
L_{\alpha }^{\prime ab}=X^{(1)}\left( L_{\alpha }^{ab}\right)
+\tfrac{\partial v^{\beta }}{\partial y^{\alpha }}L_{\beta }^{ab}
-2\tfrac{\partial u^{a}}{\partial x^r}L_{\alpha }^{br}
+\operatorname{div}(X^{\prime })L_{\alpha }^{ab},
\label{L^prime^ab_alpha}
\end{equation}
\begin{eqnarray}
L_{0}^{\prime }\! &=&\!L_{\alpha }^{ij}
\left( \!\tfrac{\partial ^{2}v^{\alpha }}{\partial x^{i}\partial x^{j}}
\!+\!\tfrac{\partial ^{2}v^{\alpha }}{\partial x^{i}\partial y^{\beta }}
y_{j}^{\beta }
\!+\!\tfrac{\partial ^{2}v^{\alpha }}{\partial x^{j}\partial y^{\beta }}
y_{i}^{\beta }
\!+\!\tfrac{\partial ^{2}v^{\alpha }}{\partial y^{\beta }\partial y^{\gamma}}
y_{j}^{\beta }y_{i}^{\gamma }
\!-\!\tfrac{\partial ^{2}u^{h}}{\partial x^{i}\partial x^{j}}y_{h}^{\alpha }\!
\right)  \label{L^prime_0} \\
&&\!+X^{(1)}\left( L_{0}\right)
+\operatorname{div}(X^{\prime })L_{0}.  \notag
\end{eqnarray}
From the formula (\ref{L^prime_0}) it also follows:
\begin{equation*}
L_{0}^{\prime }=L_{\alpha }^{ij}v_{ij}^{\alpha }-\left( L_{\beta }^{ij}
\tfrac{\partial v^{\beta }}{\partial y^{\alpha }}
-2\tfrac{\partial u^{i}}{\partial x^r}
L_{\alpha }^{jr}\right) y_{(ij)}^{\alpha }
+X^{(1)}\left(
L_{0}\right)
+\operatorname{div}(X^{\prime })L_{0}.
\end{equation*}
Replacing
$\frac{\partial (X^{(1)}(L_{\alpha }^{ij}))}{\partial y_{h}^{\beta }}
=X^{(1)}
\left(
\frac{\partial L_{\alpha }^{ij}}{\partial y_{h}^{\beta }}
\right)
+\frac{\partial v^{\gamma }}{\partial y^{\beta }}
\frac{\partial L_{\alpha }^{ij}}{\partial y_{h}^{\gamma }}
-\frac{\partial u^{h}}{\partial x^{a}}
\frac{\partial L_{\alpha }^{ij}}{\partial y_{a}^{\beta }}$
into the formula for
$\frac{\partial L_{\alpha }^{\prime ij}}{\partial y_{h}^{\beta }}$,
we obtain
\begin{eqnarray*}
\tfrac{\partial L_{\alpha }^{\prime ij}}{\partial y_{h}^{\beta }}
&=&X^{(1)}\left(
\tfrac{\partial L_{\alpha }^{ij}}{\partial y_{h}^{\beta }}
\right) +\tfrac{\partial v^{\gamma }}{\partial y^{\beta }}
\tfrac{\partial L_{\alpha }^{ij}}{\partial y_{h}^{\gamma }}
-\tfrac{\partial u^{h}}{\partial x^{a}}
\tfrac{\partial L_{\alpha }^{ij}}{\partial y_{a}^{\beta }}
+\tfrac{\partial v^{\gamma }}{\partial y^{\alpha }}
\tfrac{\partial L_{\gamma }^{ij}}{\partial y_{h}^{\beta }} \\
&&-\tfrac{\partial u^{j}}{\partial x^{a}}
\tfrac{\partial L_{\alpha }^{ia}}{\partial y_{h}^{\beta }}
-\tfrac{\partial u^{i}}{\partial x^{a}}
\tfrac{\partial L_{\alpha }^{aj}}{\partial y_{h}^{\beta }}
+\operatorname{div}(X^{\prime })
\tfrac{\partial L_{\alpha }^{ij}}{\partial y_{h}^{\beta }},
\end{eqnarray*}
and similarly,
\begin{eqnarray*}
\tfrac{\partial L_{\beta }^{\prime ih}}{\partial y_{j}^{\alpha }}
&=&X^{(1)}\left(
\tfrac{\partial L_{\beta }^{ih}}{\partial y_{j}^{\alpha }}
\right)
+\tfrac{\partial v^{\gamma }}{\partial y^{\alpha }}
\tfrac{\partial L_{\beta }^{ih}}{\partial y_{j}^{\gamma }}
-\tfrac{\partial u^{j}}{\partial x^{a}}
\tfrac{\partial L_{\beta }^{ih}}{\partial y_{a}^{\alpha }}
+\tfrac{\partial v^{\gamma }}{\partial y^{\beta }}
\tfrac{\partial L_{\gamma }^{ih}}{\partial y_{j}^{\alpha }} \\
&&-\tfrac{\partial u^{h}}{\partial x^{a}}
\tfrac{\partial L_{\beta }^{ia}}{\partial y_{j}^{\alpha }}
-\tfrac{\partial u^{i}}{\partial x^{a}}
\tfrac{\partial L_{\beta }^{ah}}{\partial y_{j}^{\alpha }}
+\operatorname{div}(X^{\prime })
\tfrac{\partial L_{\beta }^{ih}}{\partial y_{j}^{\alpha }},
\end{eqnarray*}
and taking the formulas (\ref{first_tris}) into account, we can conclude
that $\frac{\partial L_{\alpha }^{\prime ij}}{\partial y_{h}^{\beta }}
=\frac{\partial L_{\beta }^{\prime ih}}{\partial y_{j}^{\alpha }}$.
Moreover, from the formula \cite[(8)]{EJ} we know
\begin{equation}
\Theta _{\Lambda }=(-1)^{i-1}\left( L_{\alpha }^{i0}dy^{\alpha }
+L_{\alpha }^{ih}dy_{h}^{\alpha }\right) \wedge v_{i}
+\left( L-y_{i}^{\alpha }L_{\alpha }^{i0}
-y_{(hi)}^{\alpha }L_{\alpha }^{ih}\right) v, \label{localPCform}
\end{equation}
where
\begin{equation*}
\begin{array}{rl}
L-y_{i}^{\alpha }L_{\alpha }^{i0}-y_{(hi)}^{\alpha }L_{\alpha }^{ih}
= & L_{0}-y_{i}^{\alpha }L_{\alpha }^{i0}, \\
L_{\alpha }^{ij}
= &
\tfrac{1}{2-\delta _{ij}}\tfrac{\partial L}{\partial y_{(ij)}^{\alpha }}, \\
L_{\beta }^{h0}= & \tfrac{\partial L_{0}}{\partial y_{h}^{\beta }}
-\tfrac{\partial L_{\beta }^{hk}}{\partial x^{k}}
-y_{k}^{\gamma }\tfrac{\partial L_{\beta }^{hk}}{\partial y^{\gamma }},
\end{array}
\end{equation*}
the third equation again being a consequence of (\ref{first_tris}). Hence
\begin{eqnarray*}
L_{X^{(1)}}\Theta _{\Lambda }
&=&(-1)^{i-1}\left( X^{(1)}\left( L_{\alpha }^{i0}\right) dy^{\alpha }
+X^{(1)}\left( L_{\alpha }^{ih}\right)
dy_{h}^{\alpha }\right) \wedge v_{i} \\
&&+\left( L_{\alpha }^{i0}\tfrac{\partial v^{\alpha }}{\partial x^{i}}
+L_{\alpha }^{ih}\tfrac{\partial v_{h}^{\alpha }}{\partial x^{i}}\right)
v+(-1)^{i-1}
\tfrac{\partial v^{\alpha }}{\partial y^{\beta }}
L_{\alpha }^{i0}dy^{\beta }\wedge v_{i} \\
&&+(-1)^{i-1}L_{\alpha }^{ih}
\left( \tfrac{\partial v_{h}^{\alpha }}{\partial y^{\beta }}dy^{\beta }
+\tfrac{\partial v_{h}^{\alpha }}{\partial y_{j}^{\beta }}dy_{j}^{\beta }
\right) \wedge v_{i} \\
&&+(-1)^{i-1}\left(
L_{\alpha }^{i0}dy^{\alpha }+L_{\alpha }^{ih}dy_{h}^{\alpha }
\right)
\wedge L_{X^{\prime }}(v_{i}) \\
&&+\left[ X^{(1)
}\left(
L_{0}-y_{i}^{\alpha }L_{\alpha }^{i0}
\right)
\right]
v+\operatorname{div}(X^{\prime })
\left( L_{0}-y_{i}^{\alpha }L_{\alpha }^{i0}\right) v.
\end{eqnarray*}
Expanding the right-hand side above we obtain
\begin{eqnarray*}
L_{X^{(1)}}\Theta _{\Lambda } &=&(-1)^{i-1}
\left( X^{(1)}\left(
L_{\alpha }^{i0}\right)
+\tfrac{\partial v^{\beta }}{\partial y^{\alpha }}
L_{\beta }^{i0}
+\tfrac{\partial v_{h}^{\beta }}{\partial y^{\alpha }}
L_{\beta }^{ih}\right)
dy^{\alpha }\wedge v_{i} \\
&&+(-1)^{i-1}\left( X^{(1)}\left( L_{\alpha }^{ih}\right)
+\tfrac{\partial v_{j}^{\beta }}{\partial y_{h}^{\alpha }}
L_{\beta }^{ij}\right)
dy_{h}^{\alpha }\wedge v_{i} \\
&&+(-1)^{i-1}\left( L_{\alpha }^{i0}dy^{\alpha }
+L_{\alpha }^{ih}dy_{h}^{\alpha }\right)
\wedge L_{X^{\prime }}(v_{i}) \\
&&+\left( X^{(1)}\left( L_{0}-y_{i}^{\alpha }L_{\alpha }^{i0}
\right)
+L_{\alpha }^{i0}\tfrac{\partial v^{\alpha }}{\partial x^{i}}
+L_{\alpha }^{ih}\tfrac{\partial v_{h}^{\alpha }}{\partial x^{i}}
\right) v \\
&&+\operatorname{div}(X^{\prime })\left( L_{0}-y_{i}^{\alpha }
L_{\alpha }^{i0}\right) v.
\end{eqnarray*}
Moreover, by applying the formula (\ref{localPCform}) to the density
$\Lambda ^{\prime }$ we have
\begin{equation*}
\Theta _{\Lambda ^{\prime }}
=(-1)^{i-1}\left( L_{\alpha }^{\prime i0}dy^{\alpha }
+L_{\alpha }^{\prime ih}dy_{h}^{\alpha }\right)
\wedge v_{i}+\left( L_{0}^{\prime }
-y_{i}^{\alpha }L_{\alpha }^{\prime i0}\right) v.
\end{equation*}
We first compute $L_{\alpha }^{\prime i0}$. From (\ref{f12}),
(\ref{L^prime}), (\ref{L^prime^ab_alpha}), and (\ref{L^prime_0})
we deduce
\begin{equation*}
L_{\alpha }^{\prime i0}=X^{(1)}\left( L_{\alpha }^{i0}\right)
+\operatorname{div}(X^{\prime })L_{\alpha }^{i0}
+\tfrac{\partial v^{\beta }}{\partial y^{\alpha }}L_{\beta }^{i0}
-\frac{\partial u^{i}}{\partial x^r}L_{\alpha }^{r0}
+\tfrac{\partial v_{h}^{\beta }}{\partial y^{\alpha }}L_{\beta }^{hi}.
\end{equation*}
Furthermore,
\begin{align*}
L_{\alpha }^{\prime ij}& =X^{(1)}(L_{\alpha }^{ij})
+\operatorname{div}(X^{\prime })L_{\alpha }^{ij}+A_{\alpha }^{ij}, \\
L_{0}^{\prime }& =X^{(1)}(L_{0})+\operatorname{div}(X^{\prime })L_{0}
+T_{hk}^{\beta }L_{\beta }^{hk},
\end{align*}
with
\begin{align*}
A_{\alpha }^{ij}&
=\tfrac{\partial v^{\beta }}{\partial y^{\alpha }}L_{\beta }^{ij}
-\frac{\partial u^{i}}{\partial x^r}L_{\alpha }^{rj}
-\tfrac{\partial u^{j}}{\partial x^r}L_{\alpha }^{ri}, \\
T_{hk}^{\beta }&
=\tfrac{\partial ^{2}v^{\beta }}{\partial x^{h}\partial x^{k}}
+\tfrac{\partial ^{2}v^{\beta }}{\partial y^{\gamma }\partial x^{k}}
y_{h}^{\gamma }
+\tfrac{\partial ^{2}v^{\beta }}{\partial y^{\gamma }\partial x^{h}}
y_{k}^{\gamma }
+\tfrac{\partial ^{2}v^{\beta }}{\partial y^{\gamma }\partial y^{\sigma }}
y_{h}^{\gamma }y_{k}^{\sigma }
-\tfrac{\partial ^{2}u^r}{\partial x^{h}\partial x^{k}}y_{r}^{\beta }.
\end{align*}
Hence
\begin{equation*}
L_{0}^{\prime }-y_{i}^{\alpha }L_{\alpha }^{\prime i0}
=X^{(1)}(L_{0}-y_{i}^{\alpha }L_{\alpha }^{i0})
+\operatorname{div}(X^{\prime })
\left(
L_{0}-y_{i}^{\alpha }L_{\alpha }^{i0}
\right)
+\tfrac{\partial v^{\alpha }}{\partial x^{i}}L_{\alpha }^{i0}
+\tfrac{\partial v_{h}^{\beta }}{\partial x^{k}}L_{\beta }^{hk},
\end{equation*}
and we obtain
\begin{align*}
\Theta _{\Lambda ^{\prime }}& =(-1)^{i-1}
\left( X^{(1)}
\left( L_{\alpha }^{i0}
\right)
+\tfrac{\partial v^{\beta }}{\partial y^{\alpha }}L_{\beta }^{i0}
+\tfrac{\partial v_{h}^{\beta }}{\partial y^{\alpha }}L_{\beta }^{hi}
\right)
dy^{\alpha }\wedge v_{i} \\
& +(-1)^{i-1}
\left( \operatorname{div}(X^{\prime })L_{\alpha }^{i0}
-\tfrac{\partial u^{i}}{\partial x^r}L_{\alpha }^{r0}
\right) dy^{\alpha }\wedge v_{i} \\
& +(-1)^{i-1}
\left( X^{(1)}(L_{\alpha }^{ih})
+\tfrac{\partial v^{\beta }}{\partial y^{\alpha }}L_{\beta }^{ih}
-\tfrac{\partial u^{h}}{\partial x^r}L_{\alpha }^{ri}
\right)
dy_{h}^{\alpha }\wedge v_{i} \\
& -(-1)^{i-1}
\tfrac{\partial u^{i}}{\partial x^r}L_{\alpha }^{rh}
dy_{h}^{\alpha }\wedge v_{i}
+(-1)^{i-1}\operatorname{div}(X^{\prime })
L_{\alpha }^{ih}dy_{h}^{\alpha }\wedge v_{i} \\
& +\left( X^{(1)}(L_{0}-y_{i}^{\alpha }L_{\alpha }^{i0})
+\operatorname{div}(X^{\prime })
\left( L_{0}-y_{i}^{\alpha }L_{\alpha }^{i0}
\right) +\tfrac{\partial v^{\alpha }}{\partial x^{i}}
L_{\alpha }^{i0}+\tfrac{\partial v_{h}^{\beta }}{\partial x^{k}}
L_{\beta }^{hk}\right) v.
\end{align*}
By using the formula $L_{X^{\prime }}(v_{i})=\operatorname{div}(X^{\prime
})v_{i}+\sum\nolimits_{h=1}^n (-1)^{h-i-1}\tfrac{\partial u^{h}}{\partial
x^{i}}v_{h}$, we can thus conclude that $\Theta _{\Lambda ^{\prime
}}=L_{X^{(1)}}\Theta _{\Lambda }$. Finally, if $L_{X^{(2)}}\Lambda =0$, then
$\Theta _{L_{X^{(2)}}\Lambda }=0$ and by virtue of the formula in the first
part of the statement we deduce $L_{X^{(1)}}\Theta _{\Lambda }=0$. Hence
$(j^{1}s)^{\ast }(di_{X^{(1)}}\Theta )+(j^{1}s)^{\ast }(i_{X^{(1)}}d\Theta )=0 $,
and we can conclude recalling that the second term in the left-hand
side vanishes, as follows from the H-C equations in Theorem \ref{differentialPCform}.
\end{proof}
\section{Symmetries of the E-H Lagrangian density\label{symmetriesEH}}
\begin{example}
\label{example0} In the particular case of the bundle of pseudo-Riemannian
metrics of a given signature $p_{M}\colon M\to N$ (cf.\ section \ref{EH}),
the natural lift of a vector field $X^{\prime }
=u^{i}\tfrac{\partial}{\partial x^{i}}$
in $\mathfrak{X}(N)$ is given as follows (cf. \cite[section 2.2]{MR2}):
\begin{equation*}
X_{M}^{\prime }=u^{i}\tfrac{\partial }{\partial x^{i}}-\sum_{i\leq j}
\left(\tfrac{\partial u^{h}}{\partial x^{i}}y_{hj}
+\tfrac{\partial u^{h}}{\partial x^{j}}y_{ih}\right)
\tfrac{\partial }{\partial y_{ij}}\in \mathfrak{X}(M),
\end{equation*}
and from the geometric properties of the scalar curvature the E-H Lagrangian
density $\Lambda _{EH}$ admits $X_{M}^{\prime }$ as an infinitesimal
symmetry for every $X^{\prime }\in \mathfrak{X}(N)$. Let us compute its
Noether invariant $(j^{1}g)^{\ast }i_{\left( X_{M}^{\prime }\right)
^{(1)}}\Theta _{EH}$ along an Einstein metric $g$. From the formulas
\begin{multline*}
\left(
X_{M}^{\prime }\right) ^{(1)}
=u^{i}\tfrac{\partial }{\partial x^{i}}
-\sum_{i\leq j}\left(
\tfrac{\partial u^{h}}{\partial x^{i}}y_{hj}
+\tfrac{\partial u^{h}}{\partial x^{j}}y_{hi}
\right)
\tfrac{\partial }{\partial
y_{ij}} \\
-\sum_{i\leq j}
\left(
\tfrac{\partial ^{2}u^{h}}{\partial x^{i}\partial x^{k}}y_{hj}
+\tfrac{\partial ^{2}u^{h}}{\partial x^{j}\partial x^{k}}y_{hi}
+\tfrac{\partial u^{h}}{\partial x^{i}}y_{hj,k}
+\tfrac{\partial u^{h}}{\partial x^{j}}y_{hi,k}
+\tfrac{\partial u^{h}}{\partial x^{k}}y_{ij,h}
\right)
\tfrac{\partial }{\partial y_{ij,k}},
\end{multline*}
\begin{eqnarray*}
\Theta _{EH} &=&(-1)^{i-1}\sum\nolimits_{k\leq l}\left( \left( L_{EH}\right)
_{kl}^{i0}dy_{kl}+\left( L_{EH}\right) _{kl}^{ih}dy_{kl,h}\right) \wedge v_{i} \\
&&+\left( \left( L_{EH}\right) _{0}-\sum\nolimits_{k\leq l}y_{kl,i}\left(
L_{EH}\right) _{kl}^{i0}\right) v,
\end{eqnarray*}
where $\left( L_{EH}\right) _{kl}^{i0}$, $\left( L_{EH}\right) _{kl}^{ih}$,
and $\left( L_{EH}\right) _{0}$ are given in (\ref{f12}), (\ref{L_EH^ij_rs}),
and\ (\ref{L_EH_0}), respectively, by using a normal coordinate system
$(x^{i})_{i=1}^n $ centred at $x\in N$ we eventually obtain
\begin{eqnarray*}
\left(
\left( j^{1}g\right) ^{\ast }\left( i_{\left( X_{M}^{\prime }
\right)
^{(1)}}\Theta _{EH}\right) \right) _{x}
\! &=&\!(-1)^{i}
\left\{
\left(
\varepsilon _{h}
\tfrac{\partial ^{2}u^{i}}{\partial x^{h}\partial x^{h}}
-\varepsilon _{i}
\tfrac{\partial ^{2}u^{h}}{\partial x^{i}\partial x^{h}}
\right) (x)\right. \\
&&\!+\left(
\varepsilon _{i}\varepsilon _{k}
\tfrac{\partial ^{2}g_{ik}}{\partial x^{k}\partial x^{j}}
-\varepsilon _{i}\varepsilon _{k}
\tfrac{\partial ^{2}g_{kk}}{\partial x^{i}\partial x^{j}}
\right)
\!(x)u^{j}(x) \\
&&\!\left.
-\left( \varepsilon _{j}\varepsilon _{k}
\tfrac{\partial ^{2}g_{kj}}{\partial x^{j}\partial x^{k}}
-\varepsilon _{h}\varepsilon _{k}
\tfrac{\partial ^{2}g_{kk}}{\partial x^{h}\partial x^{h}}
\right)
\!(x)u^{i}(x)\right\} (v_{i})_{x}.
\end{eqnarray*}
By composing the tensor
$\left( \left( \nabla ^{g}\right) ^{2}X^{\prime }\right) _{x}$
and the isomorphism induced by the metric,
$g_{x}^{\sharp }\otimes \mathrm{id}
\colon T_{x}^{\ast }N\otimes T_{x}^{\ast }N\otimes T_{x}N
\to T_{x}N\otimes T_{x}^{\ast }N\otimes T_{x}N$, we have
\begin{equation*}
\left(
\left(
\nabla ^{g}
\right) ^{2}X^{\prime }
\right) _{x}^{\sharp }
=\varepsilon _{a}
\left(
\tfrac{\partial }{\partial x^{a}}
\right) _{x}
\otimes \left( dx^{h}\right) _{x}\otimes
\left(
\tfrac{\partial ^{2}u^{c}}{\partial x^{a}\partial x^{h}}
+u^{b}\tfrac{\partial \Gamma _{hb}^{c}}{\partial x^{a}}
\right) (x)
\left(
\tfrac{\partial }{\partial x^{c}}
\right) _{x}.
\end{equation*}
Contracting the first contravariant index and the first covariant one,
it follows:
\begin{equation*}
c_{1}^{1}\left(
\left( \nabla ^{g}\right) ^{2}X^{\prime }
\right)
_{x}^{\sharp }=\varepsilon _{h}
\left(
\tfrac{\partial ^{2}u^{i}}{\partial x^{h}\partial x^{h}}
+u^{b}\tfrac{\partial \Gamma _{hb}^{i}}{\partial x^{h}}
\right) (x)
\left(
\tfrac{\partial }{\partial x^{i}}
\right) _{x},
\end{equation*}
and contracting
$c_{1}^{1}\left( \left( \nabla ^{g}\right) ^{2}X^{\prime }
\right) _{x}^{\sharp }$ and the volume form,
\begin{equation*}
i_{c_{1}^{1}
\left(
\left( \nabla ^{g}
\right) ^{2}X^{\prime }\right) _{x}^{\sharp }}v_{x}
=(-1)^{i-1}\varepsilon _{h}\left(
\tfrac{\partial ^{2}u^{i}}{\partial x^{h}\partial x^{h}}
+u^{b}\tfrac{\partial \Gamma
_{hb}^{i}}{\partial x^{h}}\right) (x)(v_{i})_{x}.
\end{equation*}
Similarly, contracting the second contravariant index in
$\left( \left(\nabla ^{g}\right) ^{2}X^{\prime }\right) _{x}^{\sharp }$
and the first covariant one, it follows:
\begin{equation*}
c_{1}^{2}
\left(
\left( \nabla ^{g}
\right) ^{2}X^{\prime }
\right) _{x}^{\sharp }
=\varepsilon _{i}
\left( \tfrac{\partial ^{2}u^{h}}{\partial x^{i}\partial x^{h}}
+u^{b}\tfrac{\partial \Gamma _{hb}^{h}}{\partial x^{i}}
\right) (x)\left( \tfrac{\partial }{\partial x^{i}}\right) _{x},
\end{equation*}
and also,
\begin{equation*}
i_{c_{1}^{2}\left(
\left( \nabla ^{g}\right) ^{2}X^{\prime }
\right)
_{x}^{\sharp }}v_{x}=(-1)^{i-1}\varepsilon _{i}
\left( \tfrac{\partial ^{2}u^{h}}{\partial x^{i}\partial x^{h}}
+u^{b}\tfrac{\partial \Gamma _{hb}^{h}}{\partial x^{i}}
\right) (x)(v_{i})_{x}.
\end{equation*}
Finally,
\begin{equation*}
\left( j^{1}g\right) ^{\ast }
\left( i_{\left( X_{M}^{\prime }\right) ^{(1)}}
\Theta _{EH}\right) =i_{c_{1}^{2}
\left( \left( \nabla ^{g}\right)
X^{\prime }\right) ^{\sharp }}(v)
-i_{c_{1}^{1}\left( \left( \nabla ^{g}\right)
X^{\prime }\right) ^{\sharp }}(v).
\end{equation*}
\end{example}
\begin{theorem}
For $n=\dim N\geq 3$ the vector fields of the form $X_{M}^{\prime }$,
$X^{\prime }\in \mathfrak{X}(N)$, are the only infinitesimal symmetries
of the Lagrangian density $\Lambda _{EH}$.
\end{theorem}
\begin{proof}
Let $p_{M}\colon M\to N$ be the bundle of pseudo-Riemannian metrics
of a given signature. If $X$ is an infinitesimal symmetry of $\Lambda _{EH}$
and $X^{\prime }$ is its $p_{M}$-projection onto $N$, then $X-X_{M}^{\prime
} $ is a $p_{M}$-vertical symmetry of $\Lambda _{EH}$. Hence, the statement
is equivalent to saying that the only $p_{M}$-vertical symmetry $X$ of the
E-H Lagrangian is the null vector field.

Let $p\colon E\to N$ be a submersion.
If $X=V^{\alpha }\frac{\partial }{\partial y^{\alpha }}$,
$V^{\alpha }\in C^{\infty }(E)$ is an
infinitesimal symmetry of a second-order Lagrangian $L$ with P-C form
projectable onto $J^{1}E$, then $X^{(2)}\left( L\right) =0$, where
\begin{equation*}
X^{(2)}=V^{\alpha }\tfrac{\partial }{\partial y^{\alpha }}
+D_{i}(V^{\alpha })\tfrac{\partial }{\partial y_{i}^{\alpha }}
+\sum\nolimits_{h\leq i}D_{h}D_{i}(V^{\alpha })
\tfrac{\partial }{\partial y_{hi}^{\alpha }}.
\end{equation*}
As
\begin{eqnarray*}
D_{i}(V^{\alpha })
&=&\tfrac{\partial V^{\alpha }}{\partial x^{i}}
+y_{i}^{\rho }\tfrac{\partial V^{\alpha }}{\partial y^{\rho }}, \\
D_{h}D_{i}(V^{\alpha })
&=&\tfrac{\partial ^{2}V^{\alpha }}{\partial x^{h}\partial x^{i}}
+y_{h}^{\beta }
\tfrac{\partial ^{2}V^{\alpha }}{\partial x^{i}\partial y^{\beta }}
+y_{hi}^{\beta }\tfrac{\partial V^{\alpha }}{\partial y^{\beta }}
+y_{i}^{\beta }\left(
\tfrac{\partial ^{2}V^{\alpha }}{\partial x^{h}\partial y^{\beta }}
+y_{h}^{\gamma }
\tfrac{\partial ^{2}V^{\alpha }}{\partial y^{\beta }\partial y^{\gamma }}
\right) ,
\end{eqnarray*}
it follows:
\begin{eqnarray*}
X^{(2)}\left( L\right)
&=&V^{\alpha }\tfrac{\partial L_{0}}{\partial y^{\alpha }}
+\left( \
tfrac{\partial V^{\alpha }}{\partial x^{i}}
+y_{i}^{\rho}\tfrac{\partial V^{\alpha }}{\partial y^{\rho }}
\right)
\tfrac{\partial L_{0}}{\partial y_{i}^{\alpha }} \\
&&
+\left( V^{\alpha }
\tfrac{\partial L_{\beta }^{jk}}{\partial y^{\alpha }}
+\left(
\tfrac{\partial V^{\alpha }}{\partial x^{i}}+y_{i}^{\rho }
\tfrac{\partial V^{\alpha }}{\partial y^{\rho }}
\right)
\tfrac{\partial L_{\beta}^{jk}}{\partial y_{i}^{\alpha }}
+L_{\alpha }^{jk}\tfrac{\partial V^{\alpha }}{\partial y^{\beta }}
\right) y_{jk}^{\beta } \\
&&+\sum_{h\leq i}L_{\alpha }^{hi}
\left(
\tfrac{\partial ^{2}V^{\alpha }}{\partial x^{h}\partial x^{i}}
+y_{h}^{\beta }
\tfrac{\partial ^{2}V^{\alpha }}{\partial x^{i}\partial y^{\beta }}
+y_{i}^{\beta }
\left(
\tfrac{\partial ^{2}V^{\alpha }}{\partial x^{h}\partial y^{\beta }}
+y_{h}^{\gamma }
\tfrac{\partial ^{2}V^{\alpha }}{\partial y^{\beta }\partial y^{\gamma }}
\right)
\right) .
\end{eqnarray*}
Hence, the coefficient of $y_{jk}^{\beta }$ must vanish and we obtain the
following system of partial differential equations:
\begin{eqnarray*}
0 &=&V^{\alpha }\tfrac{\partial L_{\beta }^{jk}}{\partial y^{\alpha }}
+\left( \tfrac{\partial V^{\alpha }}{\partial x^{i}}+y_{i}^{\rho }
\tfrac{\partial V^{\alpha }}{\partial y^{\rho }}\right)
\tfrac{\partial L_{\beta}^{jk}}{\partial y_{i}^{\alpha }}
+L_{\alpha }^{jk}\tfrac{\partial V^{\alpha }}{\partial y^{\beta }}, \\
0 &=&V^{\alpha }\tfrac{\partial L_{0}}{\partial y^{\alpha }}
+\left(
\tfrac{\partial V^{\alpha }}{\partial x^{i}}+y_{i}^{\rho }
\tfrac{\partial V^{\alpha}}{\partial y^{\rho }}
\right)
\tfrac{\partial L_{0}}{\partial y_{i}^{\alpha }} \\
&&+\sum\nolimits_{h\leq i}L_{\alpha }^{hi}
\left(
\tfrac{\partial ^{2}V^{\alpha }}{\partial x^{h}\partial x^{i}}
+y_{h}^{\beta }
\tfrac{\partial ^{2}V^{\alpha }}{\partial x^{i}\partial y^{\beta }}
+y_{i}^{\beta }
\left(
\tfrac{\partial ^{2}V^{\alpha }}{\partial x^{h}\partial y^{\beta }}
+y_{h}^{\gamma }
\tfrac{\partial ^{2}V^{\alpha }}{\partial y^{\beta }\partial
y^{\gamma }}
\right)
\right) .
\end{eqnarray*}
In the case of the E-H Lagrangian, we obtain
\begin{equation}
\begin{array}{rl}
\text{\lbrack i]}
& 0=\sum\limits_{a\leq b}
\left[ \tfrac{\partial
\left(
L_{EH}\right) _{st}^{jk}}{\partial y_{ab}}V^{ab}
+\left( L_{EH}\right)
_{ab}^{jk}\tfrac{\partial V^{ab}}{\partial y_{st}}
\right] ,
\quad j\leq
k,s\leq t, \\
& \multicolumn{1}{r}{\medskip} \\
\text{\lbrack ii]}
& 0=\sum\limits_{a\leq b}\left\{ \tfrac{\partial
\left(
L_{EH}\right) _{0}}{\partial y_{ab}}V^{ab}
+\left(
\tfrac{\partial V^{ab}}{\partial x^{i}}
+\sum\limits_{u\leq v}y_{uv,i}
\tfrac{\partial V^{ab}}{\partial y_{uv}}\right)
\tfrac{\partial \left( L_{EH}\right) _{0}}
{\partial y_{ab,i}}\right.
\smallskip \\
& \multicolumn{1}{r}{+\sum\limits_{h\leq i}
\left( L_{EH}\right) _{ab}^{hi}
\left[
\tfrac{\partial ^{2}V^{ab}}{\partial x^{h}\partial x^{i}}
+\sum\limits_{s\leq t}y_{st,h}
\tfrac{\partial ^{2}V^{ab}}{\partial x^{i}\partial y_{st}}
\right.
\smallskip} \\
& \multicolumn{1}{r}{\left. \left. +\sum\limits_{s\leq t}y_{st,i}
\left(
\tfrac{\partial ^{2}V^{ab}}{\partial x^{h}\partial y_{st}}
+\sum\limits_{u\leq v}y_{uv,h}
\tfrac{\partial ^{2}V^{ab}}{\partial y_{st}\partial y_{uv}}
\right)
\right]
\right\} ,}
\end{array}
\label{SymVert_EH}
\end{equation}
as $\tfrac{\partial \left( L_{EH}\right) _{st}^{jk}}{\partial y_{ab,i}}=0$,
by virtue of (\ref{L_EH^ij_rs}), with
$V=\sum_{a\leq b}V^{ab}\frac{\partial }{\partial y_{ab}}$,
$V^{ab}\in C^{\infty }(M)$. Collecting the terms of degrees
$2$, $1$, and $0$ in the variables $y_{ab,c}$, $a\leq b$, on the
right-hand side of [ii]-(\ref{SymVert_EH}), it breaks into the following
equations:
\begin{equation*}
\begin{array}{rc}
0= & \sum\limits_{a\leq b}
\left\{ \tfrac{\rho }{2}
\left[
\tfrac{\partial
A_{0}^{kl,i;rs,j}}{\partial y_{ab}}V^{ab}
+\left(
A_{0}^{ab,i;rs,j}+A_{0}^{rs,i;ab,j}
\right) \tfrac{\partial V^{ab}}{\partial y_{kl}}
\right]
\right. ,
\smallskip \\
& \multicolumn{1}{r}{\left. +\tfrac{1}{2-\delta _{ij}}\left( L_{EH}\right)
_{ab}^{ij}\tfrac{\partial ^{2}V^{ab}}{\partial y_{rs}\partial y_{kl}}
\right\}
\medskip} \\
& \multicolumn{1}{r}{r\leq s,k\leq l;i,j,k,l,r,s=1,\ldots ,n,}
\end{array}
\end{equation*}
\begin{equation*}
\begin{array}{rc}
0= & \tfrac{2}{2-\delta _{ij}}\left( L_{EH}\right) _{ab}^{ij}
\tfrac{\partial ^{2}V^{ab}}{\partial x^{i}\partial y_{rs}}
+\tfrac{\rho }{2}\left(
A_{0}^{ab,i;rs,j}+A_{0}^{rs,j;ab,i}\right)
\tfrac{\partial V^{ab}}{\partial x^{i}}, \\
& \multicolumn{1}{r}{r\leq s;j,r,s=1,\ldots ,n,}
\end{array}
\end{equation*}
\begin{equation*}
0=\sum\limits_{h\leq i}\sum\limits_{a\leq b}
\left( L_{EH}\right) _{ab}^{hi}
\tfrac{\partial ^{2}V^{ab}}{\partial x^{h}\partial x^{i}},
\end{equation*}
where we have used the notations below,
\begin{align*}
\left( L_{EH}\right) _{0}
& =\tfrac{\rho }{2}\sum\nolimits_{r\leq s}
\sum\nolimits_{k\leq l}A_{0}^{kl,i;rs,j}y_{kl,i}y_{rs,j}, \\
A_{0}^{kl,i;rs,j}& =\sum_{r\leq s}\sum_{k\leq l}
\tfrac{1}{(1+\delta _{kl})(1+\delta _{rs})}
\left( 2y^{rs}\left( y^{ki}y^{jl}+y^{li}y^{jk}
\right)
-2y^{kl}y^{sr}y^{ji}
\right. \\
& +2y^{kl}
\left( y^{jr}y^{si}+y^{js}y^{ri}\right) +3y^{ij}
\left(
y^{kr}y^{ls}+y^{ks}y^{lr}\right) \\
& -y^{ir}\left( y^{ks}y^{jl}+y^{ls}y^{jk}\right)
-y^{is}\left(
y^{kr}y^{jl}+y^{lr}y^{jk}\right) \\
& \left.
-2y^{ki}
\left( y^{sl}y^{jr}+y^{rl}y^{js}
\right)
-2y^{li}\left(
y^{sk}y^{jr}+y^{rk}y^{js}
\right)
\right) .
\end{align*}
Moreover, as a calculation shows, we have
\begin{equation*}
\det
\left(
\left( L_{EH}\right) _{rs}^{ij}
\right) _{1\leq r\leq s\leq n}^{1\leq i\leq j\leq n}
=-(n-1)\rho ^{\frac{1}{2}(n+1)(n+4)},
\end{equation*}
where $\rho $ is defined in (\ref{rho}).
If $\Lambda
=\left( \Lambda _{ab}^{jk}
\right) _{1\leq a\leq b\leq n}^{1\leq j\leq k\leq n}$
is the inverse matrix of
$\left(
\left( L_{EH}\right) _{ab}^{jk}
\right) _{1\leq a\leq b\leq n}^{1\leq j\leq k\leq n}$,
then from (\ref{SymVert_EH})-[i] for $h\leq i$, it follows:
\begin{equation}
\tfrac{\partial V^{ab}}{\partial y_{st}}
=-\sum\nolimits_{c\leq d}\sum\nolimits_{p\leq q}
\Lambda _{pq}^{ab}\tfrac{\partial \left(
L_{EH}\right) _{st}^{pq}}{\partial y_{cd}}V^{cd},
\quad a\leq b,s\leq t,
\label{partialVab}
\end{equation}
and by imposing the integrability conditions to these equations we obtain
\begin{equation*}
\begin{array}{rc}
0= & \sum\limits_{a\leq b}\sum\limits_{j\leq k}
\left[ \left(
\tfrac{\partial \Lambda _{jk}^{hi}}{\partial y_{uv}}
\tfrac{\partial \left( L_{EH}\right) _{st}^{jk}}{\partial y_{ab}}
-\tfrac{\partial \Lambda _{jk}^{hi}}{\partial y_{st}}
\tfrac{\partial \left( L_{EH}\right) _{uv}^{jk}}{\partial y_{ab}}
\right. \right.
\smallskip \\
& \multicolumn{1}{r}{\left. +\Lambda _{jk}^{hi}
\left\{
\tfrac{\partial ^{2}\left( L_{EH}\right) _{st}^{jk}}{\partial y_{ab}\partial y_{uv}}
-\tfrac{\partial ^{2}\left( L_{EH}\right) _{uv}^{jk}}{\partial y_{ab}\partial y_{st}}
\right\}
\right) V^{ab}
\smallskip} \\
& \multicolumn{1}{r}{\left. +\Lambda _{jk}^{hi}\left\{ \tfrac{\partial
\left( L_{EH}\right) _{st}^{jk}}{\partial y_{ab}}
\tfrac{\partial V^{ab}}{\partial y_{uv}}
-\tfrac{\partial \left( L_{EH}\right) _{uv}^{jk}}{\partial y_{ab}}
\tfrac{\partial V^{ab}}{\partial y_{st}}\right\} \right]}
\end{array}
\end{equation*}
and substituting (\ref{partialVab}) in the previous equation, we eventually
have
\begin{eqnarray}
0\text{\negthinspace \negthinspace }
&=&\!\!\!\sum\limits_{c\leq d}\sum\limits_{j\leq k}
\left[
\!\tfrac{\partial \Lambda _{jk}^{hi}}{\partial y_{uv}}
\tfrac{\partial \left( L_{EH}\right) _{st}^{jk}}{\partial y_{cd}}
-\tfrac{\partial \Lambda _{jk}^{hi}}{\partial y_{st}}
\tfrac{\partial \left( L_{EH}\right) _{uv}^{jk}}{\partial y_{cd}}
\right.
\label{integrability_conditions} \\
&&\!\!\!+\Lambda _{jk}^{hi}
\left( \!\tfrac{\partial ^{2}\left( L_{EH}\right) _{st}^{jk}}
{\partial y_{cd}\partial y_{uv}}
-\tfrac{\partial ^{2}\left( L_{EH}\right) _{uv}^{jk}}
{\partial y_{cd}\partial y_{st}}\!\right) \notag \\
&&\!\!\!\left. +\Lambda _{jk}^{hi}
\sum\limits_{a\leq b}\sum_{p\leq q}\Lambda _{pq}^{ab}
\left( \tfrac{\partial \left( L_{EH}\right) _{st}^{pq}}{\partial y_{cd}}
\tfrac{\partial \left( L_{EH}\right) _{uv}^{jk}}{\partial y_{ab}}
-\tfrac{\partial \left( L_{EH}\right) _{st}^{jk}}{\partial y_{ab}}
\tfrac{\partial \left( L_{EH}\right) _{uv}^{pq}}{\partial y_{cd}}
\right)
\!\right]
V^{cd}.  \notag
\end{eqnarray}
Furthermore, from $\frac{\partial \Lambda }{\partial y_{pq}}
=-\Lambda \cdot \frac{\partial L}{\partial y_{pq}}\cdot \Lambda $,
it follows:
\begin{equation*}
\frac{\partial \Lambda _{jk}^{hi}}{\partial y_{uv}}
=-\sum_{\zeta \leq \eta }
\sum_{\rho \leq \sigma }\Lambda _{\zeta \eta }^{hi}
\frac{\partial \left( L_{EH}\right) _{\rho \sigma }^{\zeta \eta }}
{\partial y_{uv}}\Lambda _{jk}^{\rho \sigma }.
\end{equation*}
Hence
\begin{multline*}
\tfrac{\partial \Lambda _{jk}^{hi}}{\partial y_{uv}}
\tfrac{\partial \left( L_{EH}\right) _{st}^{jk}}{\partial y_{cd}}
-\tfrac{\partial \Lambda _{jk}^{hi}}{\partial y_{st}}
\tfrac{\partial \left( L_{EH}\right) _{uv}^{jk}}{\partial y_{cd}} \\
=\sum_{\zeta \leq \eta }\sum_{\rho \leq \sigma }
\Lambda _{\zeta \eta }^{hi}
\left(
-\tfrac{\partial \left( L_{EH}\right) _{\rho \sigma }^{\zeta \eta }}
{\partial y_{uv}}
\Lambda _{jk}^{\rho \sigma }
\tfrac{\partial \left( L_{EH}\right) _{st}^{jk}}{\partial y_{cd}}
+\tfrac{\partial \left( L_{EH}\right) _{\rho \sigma }^{\zeta \eta }}
{\partial y_{st}}\Lambda _{jk}^{\rho \sigma }
\tfrac{\partial \left( L_{EH}\right) _{uv}^{jk}}{\partial y_{cd}}
\right) ,
\end{multline*}
and letting
\begin{eqnarray*}
\Phi _{st,uv,cd}^{jk}
&=&\tfrac{\partial ^{2}\left( L_{EH}\right) _{st}^{jk}}
{\partial y_{cd}\partial y_{uv}}
-\tfrac{\partial ^{2}\left( L_{EH}\right)
_{uv}^{jk}}{\partial y_{cd}\partial y_{st}} \\
&&+\sum_{a\leq b}\sum_{p\leq q}\Lambda _{pq}^{ab}
\left(
\left( \tfrac{\partial \left( L_{EH}\right) _{ab}^{jk}}{\partial y_{st}}
-\tfrac{\partial \left( L_{EH}\right) _{st}^{jk}}{\partial y_{ab}}
\right)
\tfrac{\partial \left( L_{EH}\right) _{uv}^{pq}}{\partial y_{cd}}
\right. \\
&&\left.
+\left(
\tfrac{\partial \left( L_{EH}\right) _{uv}^{jk}}{\partial y_{ab}}
-\tfrac{\partial \left( L_{EH}\right) _{ab}^{jk}}{\partial y_{uv}}
\right)
\tfrac{\partial \left( L_{EH}\right) _{st}^{pq}}{\partial y_{cd}}
\right)
\end{eqnarray*}
the equations (\ref{integrability_conditions}) transforms into the
following:
\begin{equation*}
0=\sum\limits_{c\leq d}\left( \Lambda \cdot \Phi _{st,uv}\right)
_{cd}^{hi}V^{cd},\quad h\leq i,s\leq t,u\leq v,
\end{equation*}
where $\Phi _{st,uv}$ is the matrix $\left( \Phi _{st,uv}\right)
_{cd}^{jk}=\Phi _{st,uv,cd}^{jk}$, for every $s\leq t$, $u\leq v$.
As $\dim \Phi _{11,23}\neq 0$ for $n=3$ and $\det \Phi _{12,34}\neq 0$
for $n\geq 4$, it follows $V^{cd}=0$.
\end{proof}
\begin{remark}
For $n=2$ the E-H Lagrangian density is known to be a conformally invariant
$2$-form; hence $\Lambda _{EH}$ admits---in this dimension---the Liouville
vector field as a vertical infinitesimal symmetry.
\end{remark}
\section{Jacobi fields and presymplectic structure\label{ss4.3}}
Let $V(p)\subset TE$ be the sub-bundle of $p$-vertical tangent vectors for
the submersion $p\colon E\to N$. The infinitesimal variation of a
one-parameter variation $S_{t}$ of a section $s\colon N\to E$ is the
$p$-vertical vector field along $s$, $X\in \Gamma (N,s^{\ast }V(p))$,
defined by the formula, $X_{x}=$ tangent vector at $t=0$ to the curve
$t\mapsto S_{t}(x)$, $\forall x\in N$. On a fibred coordinate system
$(x^{i},y^{\alpha })$ we have
\begin{equation}
X_{x}
=\tfrac{\partial (y^{\alpha }\circ S)}{\partial t}(0,x)
\Bigl(\tfrac{\partial }{\partial y^{\alpha }}\Bigr)_{s(x)},
\quad \forall x\in N.
\label{f4}
\end{equation}
Let $\mathcal{S}$ be the sheaf of extremals of a second-order Lagrangian
density $\Lambda =Lv$ on $p\colon E\to N$ whose Poincar\'e-Cartan
form projects onto $J^{1}E$: For every open subset $U\subseteq N$ we denote
by $\mathcal{S}(U)$ the set of solutions to the Euler-Lagrange equations
of $\Lambda $, which are defined on $U$. As is well known (\cite{Garcia},
\cite{GoS}, \cite{Saunders}) in the Hamiltonian formalism extremals can be
characterized as the solutions to the Hamilton-Cartan equation; that is, $s$
is an extremal if and only if $(j^{1}s)^{\ast }(i_{Y}d\Theta _{\Lambda })=0$
for all $Y\in \mathfrak{X}(J^{1}E)$. Jacobi fields are the solutions to the
linearized Hamilton-Cartan equation. Precisely, a Jacobi field along an
extremal $s\in \mathcal{S}(U)$ is a $p$-vertical vector field defined along
$s$, $X\in \Gamma (U,s^{\ast }V(p))$, satisfying the Jacobi equation
$(j^{1}s)^{\ast }(i_{Y}L_{X^{(1)}}d\Theta _{\Lambda })=0$,
$\forall Y\in \mathfrak{X}(J^{1}(p^{-1}U))$, where $X^{(1)}$ is the first-order
infinitesimal contact transformation on $J^{1}E$ associated to $X$ (e.g.,
see \cite{Mu1}, \cite{MP}, \cite{Saunders}). If
$X_{x}=V^{\alpha }(x)\left( \frac{\partial }{\partial y^{\alpha }}\right) _{s(x)}$,
then
\begin{equation*}
\left( X^{(1)}\right) _{j_{x}^{1}s}
=V^{\alpha }(x)
\left( \tfrac{\partial }{\partial y^{\alpha }}
\right) _{j_{x}^{1}s}+\tfrac{\partial V^{\alpha }}{\partial x^{j}}(x)
\left( \tfrac{\partial }{\partial y_{j}^{\alpha }}
\right) _{j_{x}^{1}s}.
\end{equation*}
In fact, it is readily checked that if $S_{t}$ is a one-parameter variation
of $s$ and $S_{t}$ is an extremal for every $t$, then the infinitesimal
variation $X$ of $S_{t}$ (see \eqref{f4}) satisfies the Jacobi equation.
Hence we think of the Jacobi fields along $s$ as being the tangent space at
$s$ to the \textquotedblleft manifold\textquotedblright\ $\mathcal{S}(U)$ of
extremals and accordingly we denote it by $T_{s}\mathcal{S}(U)$. Let
$s\colon N\to E$ be an extremal of a Lagrangian density $\Lambda $
defined on $J^{1}E$.

In a fibred coordinate system $(x^{i},y^{\alpha })$ a vector field
$X\in \Gamma (U,s^{\ast }V(p))$ along an extremal $s$ is a Jacobi field if and
only if
$(j^{1}s)^{\ast }(i_{\partial /\partial y^{\alpha }}L_{X^{(1)}}
d\Theta _{\Lambda })=0$, for $1\leq \alpha \leq m$ (see \cite[section 3.5]{MR0}).
By using the formulas (\ref{p's}), (\ref{H}), and (\ref{differential_P-C_bis}),
we obtain
\begin{eqnarray*}
L_{X^{(1)}}d\Theta _{\Lambda }
&=&L_{X^{(1)}}
\left\{
(-1)^{i-1}dp_{\alpha }^{i}\wedge dy^{\alpha }\wedge v_{i}+dH\wedge v
\right\} \\
&=&(-1)^{i-1}d\left(
X^{(1)}p_{\alpha }^{i}
\right)
\wedge dy^{\alpha }\wedge v_{i} \\
&&+(-1)^{i-1}dp_{\alpha }^{i}\wedge dV^{\alpha }\wedge v_{i}
+d\left( X^{(1)}H\right) \wedge v.
\end{eqnarray*}
Hence
\begin{eqnarray*}
i_{\partial /\partial y^{\alpha }}L_{X^{(1)}}d\Theta _{\Lambda }
&=&(-1)^{i-1}\tfrac{\partial X^{(1)}(p_{\beta }^{i})}{\partial y^{\alpha }}
dy^{\beta }\wedge v_{i}
-\tfrac{\partial \left( X^{(1)}(p_{\alpha }^{i})\right) }{\partial x^{i}}v \\
&&-(-1)^{i-1}
\tfrac{\partial \left( X^{(1)}(p_{\alpha }^{i})\right) }{\partial y^{\beta }}
dy^{\beta }\wedge v_{i}-(-1)^{i-1}
\tfrac{\partial \left( X^{(1)}(p_{\alpha }^{i})\right) }
{\partial y_{j}^{\beta }}
dy_{j}^{\beta }\wedge v_{i} \\
&&+\tfrac{\partial p_{\beta }^{i}}{\partial y^{\alpha }}
\tfrac{\partial V^{\beta }}{\partial x^{i}}v
+\tfrac{\partial X^{(1)}(H)}{\partial y^{\alpha }}v,
\end{eqnarray*}
and finally,
\begin{eqnarray*}
0 &=&
\tfrac{\partial s^{\beta }}{\partial x^{i}}
\left\{ \left(
\tfrac{\partial X^{(1)}(p_{\beta }^{i})}{\partial y^{\alpha }}\circ j^{1}s
\right)
-\left(
\tfrac{\partial \left( X^{(1)}(p_{\alpha }^{i})\right) }
{\partial y^{\beta }}
\circ j^{1}
s\right)
\right\} \\
&&+\left(
\tfrac{\partial X^{(1)}(H)}{\partial y^{\alpha }}\circ j^{1}s
\right)
-\left( \tfrac{\partial \left( X^{(1)}(p_{\alpha }^{i})\right)
}{\partial x^{i}}\circ j^{1}s
\right) \\
&&-\left( \tfrac{\partial \left( X^{(1)}(p_{\alpha }^{i})\right) }
{\partial y_{j}^{\beta }}\circ j^{1}s\right)
\tfrac{\partial ^{2}s^{\beta }}{\partial x^{i}\partial x^{j}}
+\tfrac{\partial V^{\beta }}{\partial x^{i}}
\left(
\tfrac{\partial p_{\beta }^{i}}{\partial y^{\alpha }}\circ j^{1}s
\right) .
\end{eqnarray*}
Expanding,
\begin{equation}
\begin{array}{rl}
\tfrac{\partial ^{2}V^{\gamma }}{\partial x^{i}\partial x^{j}}
\left(
\tfrac{\partial p_{\alpha }^{i}}{\partial y_{j}^{\gamma }}\circ j^{1}s
\right)
= &
V^{\gamma }
\left\{ \tfrac{\partial s^{\beta }}{\partial x^{i}}
\left(
\tfrac{\partial ^{2}p_{\beta }^{i}}
{\partial y^{\alpha }\partial y^{\gamma }}
\circ j^{1}s-\tfrac{\partial ^{2}p_{\alpha }^{i}}
{\partial y^{\beta }\partial y^{\gamma }}
\circ j^{1}s
\right)
\right.
\smallskip \\
& \multicolumn{1}{r}{+\tfrac{\partial ^{2}H}{\partial y^{\alpha }
\partial y^{\gamma }}\circ j^{1}s
-\tfrac{\partial ^{2}s^{\beta }}{\partial x^{i}\partial x^{j}}
\left( \tfrac{\partial ^{2}p_{\alpha }^{i}}
{\partial y^{\gamma }\partial y_{j}^{\beta }}
\circ j^{1}s
\right)
\smallskip} \\
& \multicolumn{1}{r}{\left. -\tfrac{\partial ^{2}p_{\alpha }^{i}}
{\partial x^{i}\partial y^{\gamma }}\circ j^{1}s
\right\}
\smallskip} \\
& \multicolumn{1}{r}{+\tfrac{\partial V^{\gamma }}{\partial x^{h}}
\left\{
\tfrac{\partial s^{\beta }}{\partial x^{i}}
\left( \tfrac{\partial ^{2}p_{\beta }^{i}}
{\partial y^{\alpha }\partial y_{h}^{\gamma }}
\circ j^{1}s
-\tfrac{\partial ^{2}p_{\alpha }^{i}}
{\partial y^{\beta }\partial y_{h}^{\gamma }}
\circ j^{1}s
\right)
\right. \smallskip} \\
& \multicolumn{1}{r}{-\tfrac{\partial ^{2}s^{\beta }}
{\partial x^{i}\partial x^{j}}
\left( \tfrac{\partial ^{2}p_{\alpha }^{i}}
{\partial y_{j}^{\beta }\partial y_{h}^{\gamma }}
\circ j^{1}s\right)
+\left(
\tfrac{\partial p_{\gamma }^{h}}
{\partial y^{\alpha }}
-\tfrac{\partial p_{\alpha }^{h}}{\partial y^{\gamma }}
\right)
\circ j^{1}s\smallskip} \\
& \multicolumn{1}{r}{\left.
-\tfrac{\partial ^{2}p_{\alpha }^{i}}
{\partial x^{i}\partial y_{h}^{\gamma }}
\circ j^{1}s+\tfrac{\partial ^{2}H}
{\partial y^{\alpha }\partial y_{h}^{\gamma }}
\circ j^{1}s\right\} ,
\medskip} \\
& \multicolumn{1}{r}{1\leq \alpha \leq m.}
\end{array}
\label{Jacobi1}
\end{equation}
\begin{remark}
In the case of the E-H Lagrangian density, Greek indices of the general case
transform into a pair of non-decreasing Latin indices: $\alpha =(a,b)$,
$1\leq a\leq b\leq n$, and a Jacobi vector field along $g$ can locally be
written as follows:
\begin{eqnarray*}
X_{x} &=&\sum\nolimits_{a\leq b}V^{ab}(x)\left(
\tfrac{\partial }{\partial y_{ab}}
\right) _{g(x)} \\
&=&\sum\nolimits_{a\leq b}V^{ab}(x)(dx^{a})_{x}\odot (dx^{b})_{x},
\quad \forall x\in N,
\end{eqnarray*}
with $V^{ab}=V^{ba}$ for $a>b$. Moreover, in this case, the general
equations (\ref{Jacobi1}) for Jacobi fields can also be written as follows:
{\footnotesize
\begin{equation}
\begin{array}{ll}
{\normalsize 0\!=\!} & {\normalsize \!\!\!}\tfrac{1}{2}
\left[ \left(
\delta _{a\nu }\delta _{j\mu }+\delta _{a\mu }\delta _{\nu j}
\right)
g^{ib}-g^{ij}\delta _{a\nu }\delta _{b\mu }-g^{ab}\delta _{i\nu }
\delta _{j\mu }
\right] \frac{\partial ^{2}V^{ab}}{\partial x^{i}\partial x^{j}}
{\normalsize \smallskip } \\
\multicolumn{1}{r}{} & {\normalsize \!\!\!+}\left\{
\!\tfrac{1}{2}g^{ab}
\left( \Gamma ^{g}\right) _{\mu \nu }^{i}\!-\!g^{ib}
\left(
\Gamma ^{g}\right) _{\mu \nu }^{a}
\!+\!\tfrac{\delta _{a\nu }\delta _{i\mu }
+\delta _{a\mu }\delta _{i\nu }}{2}
\!\left[
g^{\sigma b}
\left(
\Gamma ^{g}\right) _{\lambda \sigma }^{\lambda }
\!-\!g^{\lambda \sigma }
\frac{\partial g_{\sigma \beta }}
{\partial x^{\lambda }}g^{b\beta }\!
\right] \right.
{\normalsize \smallskip } \\
\multicolumn{1}{r}{} & {\normalsize \!\!\!-}\tfrac{\delta _{a\mu }
\delta _{b\nu }}{2}
\left[
g^{\sigma i}
\left( \Gamma ^{g}
\right) _{\lambda \sigma }^{\lambda }
-g^{\lambda \sigma }\frac{\partial g_{\sigma \beta }}
{\partial x^{\lambda }}g^{i\beta }
\right]
{\normalsize +}\tfrac{\delta _{i\nu }}{2}
{\normalsize g}^{\lambda a}
\left( \Gamma ^{g}\right) _{\mu \lambda }^{b}
{\normalsize +}\tfrac{\delta _{i\mu }}{2}{\normalsize g}^{\lambda a}
\left(
\Gamma ^{g}\right) _{\lambda \nu }^{b}{\normalsize \smallskip } \\
\multicolumn{1}{r}{}
& {\normalsize \!\!\!}\left. +\tfrac{\delta _{b\nu }}{2}
\left[ g^{\lambda i}\left( \Gamma ^{g}\right) _{\mu \lambda }^{a}
-g^{\lambda
a}\left( \Gamma ^{g}\right) _{\mu \lambda }^{i}\right]
+\tfrac{\delta _{b\mu }}{2}
\left[
g^{\lambda i}\left( \Gamma ^{g}\right) _{\nu \lambda }^{a}
-g^{\lambda a}\left( \Gamma ^{g}\right) _{\nu \lambda }^{i}
\right]
\right\} \frac{\partial V^{ab}}{\partial x^{i}}{\normalsize
\smallskip } \\
\multicolumn{1}{r}{} & {\normalsize \!\!\!+g}^{\lambda b}
\left\{ (R^{g})_{\mu \nu \lambda }^{a}+g^{ar}g_{t\sigma }
\left(
\left( \Gamma ^{g}\right) _{r\nu }^{t}
\left( \Gamma ^{g_{0}}\right) _{\mu \lambda
}^{\sigma }-\left( \Gamma ^{g}\right) _{r\lambda }^{t}
\left( \Gamma ^{g}\right) _{\mu \nu }^{\sigma }
\right)
\right.
{\normalsize \smallskip }
\\
\multicolumn{1}{r}{}
& {\normalsize \!\!\!}\left.
+\left( \Gamma ^{g}\right) _{\nu \sigma }^{a}
\left( \Gamma ^{g}\right) _{\mu \lambda }^{\sigma }
\!-\!\left( \Gamma ^{g}\right) _{\lambda \sigma }^{a}
\left( \Gamma ^{g}\right) _{\mu \nu }^{\sigma }
\!-\!\left( \Gamma ^{g}\right) _{\sigma \lambda }^{\sigma }
\left( \Gamma ^{g}\right) _{\mu \nu }^{a}
\!+\!\left( \Gamma ^{g}\right) _{\mu \sigma }^{a}
\left( \Gamma ^{g}\right) _{\nu \lambda }^{\sigma }\!\right\}
{\normalsize V}^{ab}{\normalsize ,} \\
& {\normalsize 1\leq \mu \leq \nu \leq n,}
\end{array}
\label{JacobiEH}
\end{equation}
}where{\footnotesize }$\Gamma ^{g}$ denotes the Levi-Civita connection
of $g$, and $R^{g}$ its curvature tensor.
\end{remark}
\begin{example}
\label{example1}If $(N,g)$ is a flat $4$-dimensional Lorentzian manifold,
then locally, $g=\varepsilon _{i}(dx^{i})^{2}$, with $\varepsilon _{1}=-1$,
$\varepsilon _{2}=\varepsilon _{3}=\varepsilon _{4}=+1$, and the equations
(\ref{JacobiEH}) of the Jacobi fields along $g$ are as follows:
\begin{equation}
\sum\nolimits_{B=1}^{10}P_{B}^{A}(D)U^{B}=0,\quad 1\leq A\leq 10,
\label{JacobiEHflat}
\end{equation}
where
\begin{equation*}
\begin{array}{ccccc}
U^{1}=V^{11}, & U^{2}=V^{12}, & U^{3}=V^{13}, & U^{4}=V^{14}, & U^{5}=V^{22},
\\
U^{6}=V^{23}, & U^{7}=V^{24}, & U^{8}=V^{33}, & U^{9}=V^{34}, &
U^{10}=V^{44},
\end{array}
\end{equation*}
\begin{equation*}
\begin{array}{lll}
{\footnotesize P}_{1}^{1}{\footnotesize =}\frac{-1}{2}\sum_{i=2}^{4}
{\footnotesize (D}^{i}{\footnotesize )}^{2}, & {\footnotesize P}_{2}^{1}
{\footnotesize =D}^{1}{\footnotesize D}^{2}{\footnotesize ,} &
{\footnotesize P}_{3}^{1}{\footnotesize =D}^{1}{\footnotesize D}^{3}
{\footnotesize ,\smallskip } \\
{\footnotesize P}_{4}^{1}{\footnotesize =D}^{1}{\footnotesize D}^{4}
{\footnotesize ,} & {\footnotesize P}_{5}^{1}{\footnotesize =}\frac{-1}{2}
{\footnotesize (D}^{1}{\footnotesize )}^{2}{\footnotesize ,} &
{\footnotesize P}_{6}^{1}{\footnotesize =0,\smallskip } \\
{\footnotesize P}_{7}^{1}{\footnotesize =0,} & {\footnotesize P}_{8}^{1}
{\footnotesize =}\frac{-1}{2}{\footnotesize (D}^{1}{\footnotesize )}^{2}
{\footnotesize ,} & {\footnotesize P}_{9}^{1}{\footnotesize =0,\smallskip }
\\
{\footnotesize P}_{10}^{1}{\footnotesize =}\frac{-1}{2}{\footnotesize (D}^{1}
{\footnotesize )}^{2}{\footnotesize ,} & {\footnotesize P}_{1}^{2}
{\footnotesize =0,} & {\footnotesize P}_{2}^{2}{\footnotesize =}\frac{-1}{2}
{\footnotesize (D}^{3}{\footnotesize )}^{2}{\footnotesize -}\frac{1}{2}
{\footnotesize (D}^{4}{\footnotesize )}^{2}{\footnotesize ,\smallskip } \\
{\footnotesize P}_{3}^{2}{\footnotesize =}\frac{1}{2}{\footnotesize D}^{2}
{\footnotesize D}^{3}{\footnotesize ,} & {\footnotesize P}_{4}^{2}
{\footnotesize =}\frac{1}{2}{\footnotesize D}^{2}{\footnotesize D}^{4}
{\footnotesize ,} & {\footnotesize P}_{5}^{2}{\footnotesize =0,\smallskip }
\\
{\footnotesize P}_{6}^{2}{\footnotesize =}\frac{1}{2}{\footnotesize D}^{1}
{\footnotesize D}^{3}{\footnotesize ,} & {\footnotesize P}_{7}^{2}
{\footnotesize =}\frac{1}{2}{\footnotesize D}^{1}{\footnotesize D}^{4}
{\footnotesize ,} & {\footnotesize P}_{8}^{2}{\footnotesize =}\frac{-1}{2}
{\footnotesize D}^{1}{\footnotesize D}^{2}{\footnotesize ,\smallskip } \\
{\footnotesize P}_{9}^{2}{\footnotesize =0,} & {\footnotesize P}_{10}^{2}
{\footnotesize =}\frac{-1}{2}{\footnotesize D}^{1}{\footnotesize D}^{2}
{\footnotesize ,} & {\footnotesize P}_{1}^{3}{\footnotesize =0,\smallskip }
\\
{\footnotesize P}_{2}^{3}{\footnotesize =}\frac{1}{2}{\footnotesize D}^{2}
{\footnotesize D}^{3}{\footnotesize ,} & {\footnotesize P}_{3}^{3}
{\footnotesize =}\frac{-1}{2}{\footnotesize (D}^{2}{\footnotesize )}^{2}
{\footnotesize -}\frac{1}{2}{\footnotesize (D}^{4}{\footnotesize )}^{2}
{\footnotesize ,} & {\footnotesize P}_{4}^{3}{\footnotesize =}\frac{1}{2}
{\footnotesize D}^{3}{\footnotesize D}^{4}{\footnotesize ,\smallskip } \\
{\footnotesize P}_{5}^{3}{\footnotesize =}\frac{-1}{2}{\footnotesize D}^{1}
{\footnotesize D}^{3}{\footnotesize ,} & {\footnotesize P}_{6}^{3}
{\footnotesize =}\frac{1}{2}{\footnotesize D}^{1}{\footnotesize D}^{2}
{\footnotesize ,} & {\footnotesize P}_{7}^{3}{\footnotesize =0,\smallskip }
\\
{\footnotesize P}_{8}^{3}{\footnotesize =0,} & {\footnotesize P}_{9}^{3}
{\footnotesize =}\frac{1}{2}{\footnotesize D}^{1}{\footnotesize D}^{4}
{\footnotesize ,} & {\footnotesize P}_{10}^{3}{\footnotesize =}\frac{-1}{2}
{\footnotesize D}^{1}{\footnotesize D}^{3}{\footnotesize ,\smallskip } \\
{\footnotesize P}_{1}^{4}{\footnotesize =0,} & {\footnotesize P}_{2}^{4}
{\footnotesize =}\frac{1}{2}{\footnotesize D}^{2}{\footnotesize D}^{4}
{\footnotesize ,} & {\footnotesize P}_{3}^{4}{\footnotesize =}\frac{1}{2}
{\footnotesize D}^{3}{\footnotesize D}^{4}{\footnotesize ,\smallskip } \\
{\footnotesize P}_{4}^{4}{\footnotesize =}\frac{-1}{2}{\footnotesize (D}^{2}
{\footnotesize )}^{2}{\footnotesize -}\frac{1}{2}{\footnotesize (D}^{3}
{\footnotesize )}^{2}{\footnotesize ,} & {\footnotesize P}_{5}^{4}
{\footnotesize =}\frac{-1}{2}{\footnotesize D}^{1}{\footnotesize D}^{4}
{\footnotesize ,} & {\footnotesize P}_{6}^{4}{\footnotesize =0,\smallskip }
\\
{\footnotesize P}_{7}^{4}{\footnotesize =}\frac{1}{2}{\footnotesize D}^{1}
{\footnotesize D}^{2}{\footnotesize ,} & {\footnotesize P}_{8}^{4}
{\footnotesize =}\frac{-1}{2}{\footnotesize D}^{1}{\footnotesize D}^{4}
{\footnotesize ,} & {\footnotesize P}_{9}^{4}{\footnotesize =}\frac{1}{2}
{\footnotesize D}^{1}{\footnotesize D}^{3}{\footnotesize ,\smallskip } \\
{\footnotesize P}_{10}^{4}{\footnotesize =0,} & {\footnotesize P}_{1}^{5}
{\footnotesize =}\frac{1}{2}{\footnotesize D}^{2}{\footnotesize D}^{2} &
{\footnotesize P}_{2}^{5}{\footnotesize =-D}^{1}{\footnotesize D}^{2}
{\footnotesize \smallskip } \\
{\footnotesize P}_{3}^{5}{\footnotesize =0,} & {\footnotesize P}_{4}^{5}
{\footnotesize =0,} & {\footnotesize P}_{5}^{5}{\footnotesize =}\frac{1}{2}
{\footnotesize (D}^{1}{\footnotesize )}^{2}{\footnotesize -}\frac{1}{2}
\sum_{i=3}^{4}{\footnotesize (D}^{i}{\footnotesize )}^{2}{\footnotesize
,\smallskip } \\
{\footnotesize P}_{6}^{5}{\footnotesize =D}^{2}{\footnotesize D}^{3}
{\footnotesize ,} & {\footnotesize P}_{7}^{5}{\footnotesize =D}^{2}
{\footnotesize D}^{4}{\footnotesize ,} & {\footnotesize P}_{8}^{5}
{\footnotesize =}\frac{-1}{2}{\footnotesize D}^{2}{\footnotesize D}^{2}
{\footnotesize ,\smallskip } \\
{\footnotesize P}_{9}^{5}{\footnotesize =0,} & {\footnotesize P}_{10}^{5}
{\footnotesize =}\frac{-1}{2}{\footnotesize D}^{2}{\footnotesize D}^{2}
{\footnotesize ,} & {\footnotesize P}_{1}^{6}{\footnotesize =}\frac{1}{2}
{\footnotesize D}^{2}{\footnotesize D}^{3}{\footnotesize ,}
\end{array}
\end{equation*}
\begin{equation*}
\begin{array}{lll}
{\footnotesize P}_{2}^{6}{\footnotesize =}\frac{-1}{2}{\footnotesize D}^{1}
{\footnotesize D}^{3}{\footnotesize ,} & {\footnotesize P}_{3}^{6}
{\footnotesize =}\frac{-1}{2}{\footnotesize D}^{1}{\footnotesize D}^{2}
{\footnotesize ,} & {\footnotesize P}_{4}^{6}{\footnotesize =0,\smallskip }
\\
{\footnotesize P}_{5}^{6}{\footnotesize =0,} & {\footnotesize P}_{6}^{6}
{\footnotesize =}\frac{1}{2}{\footnotesize (D}^{1}{\footnotesize )}^{2}
{\footnotesize -}\frac{1}{2}{\footnotesize (D}^{4}{\footnotesize )}^{2}
{\footnotesize ,} & {\footnotesize P}_{7}^{6}{\footnotesize =}\frac{1}{2}
{\footnotesize D}^{3}{\footnotesize D}^{4}{\footnotesize ,\smallskip } \\
{\footnotesize P}_{8}^{6}{\footnotesize =0,} & {\footnotesize P}_{9}^{6}
{\footnotesize =}\frac{1}{2}{\footnotesize D}^{2}{\footnotesize D}^{4}
{\footnotesize ,} & {\footnotesize P}_{10}^{6}{\footnotesize =}\frac{-1}{2}
{\footnotesize D}^{2}{\footnotesize D}^{3}{\footnotesize ,\smallskip } \\
{\footnotesize P}_{1}^{7}{\footnotesize =}\frac{1}{2}{\footnotesize D}^{2}
{\footnotesize D}^{4}{\footnotesize ,} & {\footnotesize P}_{2}^{7}
{\footnotesize =}\frac{-1}{2}{\footnotesize D}^{1}{\footnotesize D}^{4}
{\footnotesize ,} & {\footnotesize P}_{3}^{7}{\footnotesize =0,\smallskip }
\\
{\footnotesize P}_{4}^{7}{\footnotesize =}\frac{-1}{2}{\footnotesize D}^{1}
{\footnotesize D}^{2}{\footnotesize ,} & {\footnotesize P}_{5}^{7}
{\footnotesize =0,} & {\footnotesize P}_{6}^{7}{\footnotesize =}\frac{1}{2}
{\footnotesize D}^{3}{\footnotesize D}^{4}{\footnotesize ,\smallskip } \\
{\footnotesize P}_{7}^{7}{\footnotesize =}\frac{1}{2}{\footnotesize (D}^{1}
{\footnotesize )}^{2}{\footnotesize -}\frac{1}{2}{\footnotesize (D}^{3}
{\footnotesize )}^{2}{\footnotesize ,} & {\footnotesize P}_{8}^{7}
{\footnotesize =}\frac{-1}{2}{\footnotesize D}^{2}{\footnotesize D}^{4}
{\footnotesize ,} & {\footnotesize P}_{9}^{7}{\footnotesize =}\frac{1}{2}
{\footnotesize D}^{2}{\footnotesize D}^{3}{\footnotesize ,\smallskip } \\
{\footnotesize P}_{10}^{7}{\footnotesize =0,} & {\footnotesize P}_{1}^{8}
{\footnotesize =}\frac{1}{2}{\footnotesize (D}^{3}{\footnotesize )}^{2}
{\footnotesize ,} & {\footnotesize P}_{2}^{8}{\footnotesize =0,\smallskip }
\\
{\footnotesize P}_{3}^{8}{\footnotesize =-D^{1}D^{3},} & {\footnotesize P}
_{4}^{8}{\footnotesize =0,} & {\footnotesize P}_{5}^{8}{\footnotesize
=-\frac{1}{2}(D^{3})^{2},\smallskip } \\
{\footnotesize P}_{6}^{8}{\footnotesize =D}^{2}{\footnotesize D}^{3}
{\footnotesize ,} & {\footnotesize P}_{7}^{8}{\footnotesize =0,} &
{\footnotesize P}_{8}^{8}{\footnotesize =}\frac{1}{2}{\footnotesize (D}^{1}
{\footnotesize )}^{2}{\footnotesize -}\frac{1}{2}{\footnotesize (D}^{2}
{\footnotesize )}^{2}{\footnotesize ,\smallskip } \\
{\footnotesize P}_{9}^{8}{\footnotesize =D}^{3}{\footnotesize D}^{4}
{\footnotesize ,} & {\footnotesize P}_{10}^{8}{\footnotesize =}\frac{-1}{2}
{\footnotesize (D}^{3}{\footnotesize )}^{2}{\footnotesize ,} &
{\footnotesize P}_{1}^{9}{\footnotesize =}\frac{1}{2}{\footnotesize D}^{3}
{\footnotesize D}^{4}{\footnotesize ,\smallskip } \\
{\footnotesize P}_{2}^{9}{\footnotesize =0,} & {\footnotesize P}_{3}^{9}
{\footnotesize =}\frac{-1}{2}{\footnotesize D}^{1}{\footnotesize D}^{4}
{\footnotesize ,} & {\footnotesize P}_{4}^{9}{\footnotesize =}\frac{-1}{2}
{\footnotesize D}^{1}{\footnotesize D}^{3}{\footnotesize ,\smallskip } \\
{\footnotesize P}_{5}^{9}{\footnotesize =}\frac{-1}{2}{\footnotesize D}^{3}
{\footnotesize D}^{4}{\footnotesize ,} & {\footnotesize P}_{6}^{9}
{\footnotesize =}\frac{1}{2}{\footnotesize D}^{2}{\footnotesize D}^{4}
{\footnotesize ,} & {\footnotesize P}_{7}^{9}{\footnotesize =}\frac{1}{2}
{\footnotesize D}^{2}{\footnotesize D}^{3}{\footnotesize ,\smallskip } \\
{\footnotesize P}_{8}^{9}{\footnotesize =0,} & {\footnotesize P}_{9}^{9}
{\footnotesize =}\frac{1}{2}{\footnotesize (D}^{1}{\footnotesize )}^{2}
{\footnotesize -}\frac{1}{2}{\footnotesize (D}^{2}{\footnotesize )}^{2}
{\footnotesize ,} & {\footnotesize P}_{10}^{9}{\footnotesize =0,\smallskip }
\\
{\footnotesize P}_{1}^{10}{\footnotesize =}\frac{1}{2}{\footnotesize (D}^{4}
{\footnotesize )}^{2}{\footnotesize ,} & {\footnotesize P}_{2}^{10}
{\footnotesize =0,} & {\footnotesize P}_{3}^{10}{\footnotesize =0,\smallskip}
\\
{\footnotesize P}_{4}^{10}{\footnotesize =-D}^{1}{\footnotesize D}^{4}
{\footnotesize ,} & {\footnotesize P}_{5}^{10}{\footnotesize =}\frac{-1}{2}
{\footnotesize (D}^{4}{\footnotesize )}^{2}{\footnotesize ,} &
{\footnotesize P}_{6}^{10}{\footnotesize =0,\smallskip } \\
{\footnotesize P}_{7}^{10}{\footnotesize =D}^{2}{\footnotesize D}^{4}
{\footnotesize ,} & {\footnotesize P}_{8}^{10}{\footnotesize =-}\frac{1}{2}
{\footnotesize (D}^{4}{\footnotesize )}^{2}{\footnotesize ,} &
{\footnotesize P}_{9}^{10}{\footnotesize =D}^{3}{\footnotesize D}^{4}
{\footnotesize ,\smallskip } \\
{\footnotesize P}_{10}^{10}{\footnotesize =}\frac{1}{2}{\footnotesize (D}^{1}
{\footnotesize )}^{2}{\footnotesize -}\frac{1}{2}\sum_{i=2}^{3}
{\footnotesize (D}^{i}{\footnotesize )}^{2}{\footnotesize ,} & D^{i}
=\frac{\partial }{\partial x^{i}}, & 1\leq i\leq 4.
\end{array}
\end{equation*}
If a solution $(U^{A})_{A=1}^{10}$ to (\ref{JacobiEHflat}) is expanded in
power series up to second order, i.e.,
$U^{A}=\lambda ^{A}+\sum\nolimits_{1\leq j\leq 4}
\lambda _{j}^{A}x^{j}
+\sum\nolimits_{1\leq j\leq k\leq 4}\lambda _{jk}^{A}x^{j}x^{k}+$
terms of order $\geq 3$, then
evaluating it at $x^{1}=\ldots =x^{4}=0$, we obtain
\begin{equation}
\left\{
\begin{array}{l}
\lambda _{22}^{1}=\lambda _{12}^{2}-\lambda _{11}^{5}+\lambda _{33}^{5}
+\lambda _{44}^{5}-\lambda _{23}^{6}-\lambda _{24}^{7}+\lambda _{22}^{8}
+\lambda _{22}^{10}, \\
\lambda _{23}^{1}=\lambda _{13}^{2}+\lambda _{12}^{3}-2\lambda _{11}^{6}
+2\lambda _{44}^{6}-\lambda _{34}^{7}-\lambda _{24}^{9}+\lambda _{23}^{10}, \\
\lambda _{24}^{1}=\lambda _{14}^{2}+\lambda _{12}^{4}-\lambda _{34}^{6}
-2\lambda _{11}^{7}+2\lambda _{33}^{7}+\lambda _{24}^{8}-\lambda _{23}^{9}, \\
\lambda _{33}^{1}=\lambda _{13}^{3}+\lambda _{33}^{5}-\lambda _{23}^{6}
-\lambda _{11}^{8}+\lambda _{22}^{8}-\lambda _{34}^{9}+\lambda _{33}^{10}, \\
\lambda _{34}^{1}=\lambda _{14}^{3}+\lambda _{13}^{4}+\lambda _{34}^{5}
-\lambda _{24}^{6}-\lambda _{23}^{7}-2\lambda _{11}^{9}+2\lambda _{22}^{9}, \\
\lambda _{44}^{1}=\lambda _{14}^{4}-\lambda _{44}^{5}+\lambda _{24}^{7}
-2\lambda _{33}^{5}+2\lambda _{23}^{6}-2\lambda _{22}^{8}-\lambda
_{22}^{10}+\lambda _{34}^{9}-\lambda _{33}^{10}-\lambda _{11}^{10}, \\
\lambda _{23}^{2}=2\lambda _{22}^{3}+2\lambda _{44}^{3}-\lambda
_{34}^{4}+\lambda _{13}^{5}-\lambda _{12}^{6}-\lambda _{14}^{9}
+\lambda _{13}^{10}, \\
\lambda _{24}^{2}=-\lambda _{34}^{3}+2\lambda _{22}^{4}
+2\lambda _{33}^{4}+\lambda _{14}^{5}-\lambda _{12}^{7}+\lambda _{14}^{8}
-\lambda _{13}^{9}, \\
\lambda _{33}^{2}=-\lambda _{44}^{2}+\lambda _{23}^{3}
+\lambda _{24}^{4}+\lambda _{13}^{6}+\lambda _{14}^{7}
-\lambda _{12}^{8}-\lambda _{12}^{10}, \\
\lambda _{44}^{8}=-2\lambda _{33}^{5}-2\lambda _{44}^{5}
+2\lambda _{23}^{6}+2\lambda _{24}^{7}-2\lambda _{22}^{8}
-2\lambda _{22}^{10}+2\lambda _{34}^{9}-2\lambda _{33}^{10}.
\end{array}
\right.  \label{EqLambdas}
\end{equation}
Hence the space of quadratic Jacobi fields along $g$ is a vector space of
dimension $90$, with basis
\begin{equation*}
\begin{array}{cccccc}
(x^{1})^{2}E_{1}, & x^{1}x^{2}E_{1}, & x^{1}x^{3}E_{1}, & x^{1}x^{4}E_{1}, &
(x^{1})^{2}E_{2}, & (x^{2})^{2}E_{2}, \\
x^{3}x^{4}E_{2}, & ((x^{4})^{2}-(x^{3})^{2})E_{2}, & (x^{1})^{2}E_{3}, &
x^{2}x^{4}E_{3}, & (x^{3})^{2}E_{3}, & (x^{1})^{2}E_{4}, \\
x^{2}x^{3}E_{4}, & (x^{4})^{2}E_{4}, & x^{1}x^{2}E_{5}, & (x^{2})^{2}E_{5},
& x^{2}x^{3}E_{5}, & x^{2}x^{4}E_{5}, \\
x^{1}x^{4}E_{6}, & (x^{2})^{2}E_{6}, & (x^{3})^{2}E_{6}, & x^{1}x^{3}E_{7},
& (x^{2})^{2}E_{7}, & (x^{4})^{2}E_{7}, \\
x^{1}x^{3}E_{8}, & x^{2}x^{3}E_{8}, & (x^{3})^{2}E_{8}, & x^{3}x^{4}E_{8}, &
x^{1}x^{2}E_{9}, & \left( x^{3}\right) ^{2}E_{9}, \\
\left( x^{4}\right) ^{2}E_{9}, & x^{1}x^{4}E_{10}, & x^{2}x^{4}E_{10}, &
x^{3}x^{4}E_{10}, & (x^{4})^{2}E_{10}, &
\end{array}
\end{equation*}
\begin{equation*}
\begin{array}{ccc}
(x^{2})^{2}E_{1}+x^{1}x^{2}E_{2}, & x^{2}x^{4}E_{1}+x^{1}x^{4}E_{2}, &
x^{2}x^{3}E_{1}+x^{1}x^{2}E_{3}, \\
x^{2}x^{3}E_{1}+x^{1}x^{3}E_{2}, & x^{2}x^{4}E_{1}+x^{1}x^{2}E_{4}, &
(x^{3})^{2}E_{1}+x^{1}x^{3}E_{3}, \\
x^{3}x^{4}E_{1}+x^{1}x^{4}E_{3}, & x^{3}x^{4}E_{1}+x^{1}x^{3}E, &
(x^{4})^{2}E_{1}+x^{1}x^{4}E_{4}, \\
-(x^{2})^{2}E_{1}+(x^{1})^{2}E_{5}, & x^{3}x^{4}E_{1}+x^{3}x^{4}E_{5}, &
((x^{4})^{2}-(x^{3})^{2})\left( E_{1}+E_{5}\right) , \\
2x^{2}x^{3}E_{1}+(x^{4})^{2}E_{6}, & -2x^{2}x^{3}E_{1}+(x^{1})^{2}E_{6}, &
-x^{2}x^{4}E_{1}+x^{3}x^{4}E_{6}, \\
-x^{3}x^{4}E_{1}+x^{2}x^{4}E_{6}, & -x^{2}x^{3}E_{1}+x^{3}x^{4}E_{7}, &
-2x^{2}x^{4}E_{1}+(x^{1})^{2}E_{7}, \\
2x^{2}x^{4}E_{1}+\left( x^{3}\right) ^{2}E_{7}, &
-x^{3}x^{4}E_{1}+x^{2}x^{3}E_{7}, & x^{2}x^{4}E_{1}+x^{2}x^{4}E_{8}, \\
-\left( x^{3}\right) ^{2}E_{1}+(x^{1})^{2}E_{8}, &
-x^{2}x^{3}E_{1}+x^{2}x^{4}E_{9}, & -x^{2}x^{4}E_{1}+x^{2}x^{3}E_{9}, \\
-2x^{3}x^{4}E_{1}+(x^{1})^{2}E_{9}, & 2x^{3}x^{4}E_{1}+(x^{2})^{2}E_{9}, &
x^{2}x^{3}\left( E_{1}+E_{10}\right) , \\
-(x^{4})^{2}E+(x^{1})^{2}E_{10}, & 2x^{2}x^{3}E_{2}+(x^{2})^{2}E_{3}, &
2x^{2}x^{3}E_{2}+(x^{4})^{2}E_{3}, \\
(x^{3})^{2}E_{2}+x^{2}x^{3}E_{3}, & -x^{2}x^{4}E_{2}+x^{3}x^{4}E_{3}, &
-x^{2}x^{3}E_{2}+x^{3}x^{4}E_{4}, \\
2x^{2}x^{4}E_{2}+(x^{2})^{2}E_{4}, & 2x^{2}x^{4}E_{2}+(x^{3})^{2}E_{4}, &
(x^{3})^{2}E_{2}+x^{2}x^{4}E_{4}, \\
x^{2}x^{3}E_{2}+x^{1}xE_{5}, & x^{2}x^{4}E_{2}+x^{1}x^{4}E_{5}, &
-x^{2}x^{3}E_{2}+x^{1}x^{2}E_{6}, \\
(x^{3})^{2}E_{2}+x^{1}x^{3}E_{6}, & -x^{2}x^{4}E_{2}+x^{1}x^{2}E_{7}, &
(x^{3})^{2}E_{2}+x^{1}x^{4}E_{7}, \\
x^{2}x^{4}E_{2}+x^{1}x^{4}E_{8}, & -(x^{3})^{2}E_{2}+x^{1}x^{2}E_{8}, &
-x^{2}x^{3}E_{2}+x^{1}x^{4}E_{9}, \\
-x^{2}x^{4}E_{2}+x^{1}x^{3}E_{9}, & x^{2}x^{3}E_{2}+x^{1}x^{3}E_{10}, &
-(x^{3})^{2}E_{2}+x^{1}x^{2}E_{10}, \\
(x^{3})^{2}E+x^{2}x^{3}E_{6}, & -(x^{3})^{2}E_{5}+(x^{2})^{2}E_{8}, &
\end{array}
\end{equation*}
\begin{equation*}
\begin{array}{c}
((x^{3})^{2}-(x^{4})^{2})E_{1}+(x^{3})^{2}E_{5}+x^{2}x^{4}E_{7}, \\
((x^{4})^{2}-(x^{2})^{2})E_{1}-(x^{3})^{2}E_{5}+(x^{4})^{2}E_{8}, \\
((x^{2})^{2}-(x^{4})^{2})E_{1}+(x^{3})^{2}E_{5}+x^{3}x^{4}E_{9}, \\
((x^{4})^{2}-(x^{2})^{2})E_{1}-(x^{3})^{2}E_{5}+(x^{3})^{2}E_{10}, \\
((x^{4})^{2}-(x^{3})^{2})E_{1}-(x^{3})^{2}E_{5}+(x^{2})^{2}E_{10},
\end{array}
\end{equation*}
where
\begin{equation}
\begin{array}{ccccc}
E_{1}=\frac{\partial }{\partial y_{11}}, 
& E_{2}=\frac{\partial }{\partial y_{12}}, 
& E_{3}=\frac{\partial }{\partial y_{13}}, 
& E_{4}=\frac{\partial }{\partial y_{14}}, 
& E_{5}=\frac{\partial }{\partial y_{22}},
\smallskip \\
E_{6}=\frac{\partial }{\partial y_{23}}, 
& E_{7}=\frac{\partial }{\partial y_{24}}, 
& E_{8}=\frac{\partial }{\partial y_{33}}, 
& E_{9}=\frac{\partial }{\partial y_{34}}, 
& E_{10}=\frac{\partial }{\partial y_{44}}.
\end{array}
\label{E's}
\end{equation}
More generally, if 
$U^{A}=\sum_{|J|=r}\lambda _{J}^{A}
\left( x^{1}\right)^{j_{1}}\cdots \left( x^{4}\right) ^{j_{4}}$, 
where $1\leq A\leq 10$, and $J=(j_{1},\ldots ,j_{4})$, 
is a homogeneous Jacobi field along $g$ of order 
$r\geq 3$, then for every multi-index $(i_{1},\ldots ,i_{4})$ 
of order $i_{1}+\ldots +i_{4}=r-2$, the functions 
$\left[ \left( D^{1}\right) ^{i_{1}}\circ 
\ldots \circ \left( D^{4}\right) ^{i_{4}}\right] 
\left(U^{A}\right) $ 
are a quadratic Jacobi field along $g$, as the operators
$P_{B}^{A}(D)$ are linear and of constant coefficients.

If $(j_{1},\ldots ,j_{4})=(i_{1},\ldots ,i_{4})+(ab)$, then for $a<b$ we
have $j_{a}=i_{a}+1$, $j_{b}=i_{b}+1$, $j_{h}=i_{h}$ for every $h\neq a,b$,
and for $a=b$ we have $j_{a}=i_{a}+2$, $j_{h}=i_{h}$ for every $h\neq a$.
Hence
\begin{eqnarray*}
\left[ \!\left( D^{1}\right) ^{i_{1}}\!\circ \ldots \circ 
\left(
D^{4}\right) ^{i_{4}}\!\right] \!\left( U^{A}
\right) \!\!
&=&\!\!\sum\limits_{a<b}\lambda _{I+(ab)}^{A}i_{1}!\cdots (i_{a}+1)!
\cdots (i_{b}+1)!\cdots i_{4}!x^{a}x^{b} \\
&&\!\!+\sum\limits_{a}\tfrac{1}{2}\lambda _{I+(aa)}^{A}i_{1}!
\cdots (i_{a}+2)!\cdots i_{4}!\left( x^{a}\right) ^{2},
\end{eqnarray*}
and consequently, the functions
\begin{eqnarray*}
\lambda _{ab}^{A} &=&\lambda _{I+(ab)}^{A}i_{1}!\cdots (i_{a}+1)!
\cdots (i_{b}+1)!\cdots i_{4}!, \\
\lambda _{aa}^{A} &=&\tfrac{1}{2}\lambda _{I+(aa)}^{A}i_{1}!
\cdots (i_{a}+2)!\cdots i_{4}!
\end{eqnarray*}
must satisfy the equations (\ref{EqLambdas}) for $1\leq a<b\leq 4$, and
every multi-index $(i_{1},\ldots ,i_{4})$ of order $r-2$.
\end{example}

\begin{example}
\label{example2}If $N=(\mathbb{R}/2\pi \mathbb{Z})^{4}$ is a $4$-dimensional
torus with Lorentzian metric $g=\varepsilon _{i}(dx^{i})^{2}$, $\varepsilon
_{1}=-1$, $\varepsilon _{2}=\varepsilon _{3}=\varepsilon _{4}=+1$, as in the
Example \ref{example1}, then we can obtain the global solutions to Jacobi
equations (\ref{JacobiEHflat}) by expanding in Fourier series; namely,
\begin{equation*}
U^{A}=\sum\nolimits_{(k_{1},\ldots ,k_{4})\in \mathbb{Z}^{4}}U_{k_{1},
\ldots,k_{4}}^{A}\exp (ik_{j}x^{j}),\quad U_{k_{1},\ldots ,k_{4}}^{A
}\in \mathbb{C},
\end{equation*}
so that $\tfrac{\partial ^{2}U^{A}}{\partial x^r\partial x^{s}}
=-\sum\nolimits_{(k_{1},\ldots ,k_{4})\in \mathbb{Z}^{4}}k_{r}k_{s}U_{k_{1},
\ldots ,k_{4}}^{A}\exp (ik_{j}x^{j})$ and the equations (\ref{JacobiEHflat})
transform into the following:

\begin{equation*}
\begin{array}{rl}
0= & \tfrac{1}{2}\left( (k_{2})^{2}+(k_{3})^{2}+(k_{4})^{2}\right)
U_{k_{1},\ldots ,k_{4}}^{1}-k_{1}k_{2}U_{k_{1},\ldots
,k_{4}}^{2}-k_{1}k_{3}U_{k_{1},\ldots ,k_{4}}^{3} \\
& \multicolumn{1}{r}{-k_{1}k_{4}U_{k_{1},\ldots ,k_{4}}^{4}+\tfrac{1}{2}
(k_{1})^{2}\left( U_{k_{1},\ldots ,k_{4}}^{5}+U_{k_{1},\ldots
,k_{4}}^{8}+U_{k_{1},\ldots ,k_{4}}^{9}\right) ,} \\
0= & -\left( (k_{3})^{2}+(k_{4})^{2}\right) U_{k_{1},\ldots
,k_{4}}^{2}+k_{2}k_{3}U_{k_{1},\ldots ,k_{4}}^{3}+k_{2}k_{4}U_{k_{1},
\ldots ,k_{4}}^{4} \\
& \multicolumn{1}{r}{+k_{1}k_{3}U_{k_{1},\ldots
,k_{4}}^{6}+k_{1}k_{4}U_{k_{1},\ldots ,k_{4}}^{7}-k_{1}k_{2}
\left(
U_{k_{1},\ldots ,k_{4}}^{8}+U_{k_{1},\ldots ,k_{4}}^{10}
\right) ,} \\
0= & k_{2}k_{3}U_{k_{1},\ldots ,k_{4}}^{2}
-\left(
(k_{2})^{2}+(k_{4})^{2}
\right) U_{k_{1},\ldots
,k_{4}}^{3}+k_{3}k_{4}U_{k_{1},\ldots ,k_{4}}^{4} \\
& \multicolumn{1}{r}{+k_{1}k_{2}U_{k_{1},\ldots ,k_{4}}^{6}-k_{1}k_{3}
\left(
U_{k_{1},\ldots ,k_{4}}^{5}+U_{k_{1},\ldots ,k_{4}}^{10}
\right)
+k_{1}k_{4}U_{k_{1},\ldots ,k_{4}}^{9},} \\
0= & k_{2}k_{4}U_{k_{1},\ldots ,k_{4}}^{2}+k_{3}k_{4}U_{k_{1},\ldots
,k_{4}}^{3}
-\left((k_{2})^{2}+(k_{3})^{2}\right) U_{k_{1},\ldots ,k_{4}}^{4}
\\
& \multicolumn{1}{r}{-k_{1}k_{4}\left( U_{k_{1},\ldots
,k_{4}}^{5}+U_{k_{1},\ldots ,k_{4}}^{8}\right) +k_{1}k_{2}U_{k_{1},\ldots
,k_{4}}^{7}+k_{1}k_{3}U_{k_{1},\ldots ,k_{4}}^{9},} \\
0= & \tfrac{1}{2}(k_{2})^{2}U_{k_{1},\ldots
,k_{4}}^{1}-k_{1}k_{2}U_{k_{1},\ldots ,k_{4}}^{2}+\tfrac{1}{2}\left(
(k_{1})^{2}-(k_{3})^{2}-(k_{4})^{2}\right) U_{k_{1},\ldots ,k_{4}}^{5} \\
& +k_{2}k_{3}U_{k_{1},\ldots ,k_{4}}^{6}+k_{2}k_{4}U_{k_{1},\ldots
,k_{4}}^{7}-\tfrac{1}{2}(k_{2})^{2}\left( U_{k_{1},\ldots
,k_{4}}^{8}+U_{k_{1},\ldots ,k_{4}}^{10}\right) ,
\end{array}
\end{equation*}
\begin{equation*}
\begin{array}{rl}
0= & k_{2}k_{3}U_{k_{1},\ldots ,k_{4}}^{1}-k_{1}k_{3}U_{k_{1},\ldots
,k_{4}}^{2}-k_{1}k_{2}U_{k_{1},\ldots ,k_{4}}^{3} \\
& \multicolumn{1}{r}{+\left( (k_{1})^{2}-(k_{4})^{2}\right) U_{k_{1},\ldots
,k_{4}}^{6}+k_{3}k_{4}U_{k_{1},\ldots ,k_{4}}^{7}+k_{2}k_{4}U_{k_{1},\ldots
,k_{4}}^{9}-k_{2}k_{3}U_{k_{1},\ldots ,k_{4}}^{10},} \\
0= & k_{2}k_{4}U_{k_{1},\ldots ,k_{4}}^{1}-k_{1}k_{4}U_{k_{1},\ldots
,k_{4}}^{2}-k_{1}k_{2}U_{k_{1},\ldots ,k_{4}}^{4} \\
& \multicolumn{1}{r}{+k_{3}k_{4}U_{k_{1},\ldots ,k_{4}}^{6}+\left(
(k_{1})^{2}-(k_{3})^{2}\right) U_{k_{1},\ldots
,k_{4}}^{7}-k_{2}k_{4}U_{k_{1},\ldots ,k_{4}}^{8}+k_{2}k_{3}U_{k_{1},\ldots
,k_{4}}^{9},} \\
0= & \tfrac{1}{2}k_{3}^{2}\left( U_{k_{1},\ldots ,k_{4}}^{1}-U_{k_{1},\ldots
,k_{4}}^{5}\right) -k_{1}k_{3}U_{k_{1},\ldots
,k_{4}}^{3}+k_{2}k_{3}U_{k_{1},\ldots ,k_{4}}^{6} \\
& \multicolumn{1}{r}{+\tfrac{1}{2}\left( (k_{1})^{2}-(k_{2})^{2}\right)
U_{k_{1},\ldots ,k_{4}}^{8}+k_{3}k_{4}U_{k_{1},\ldots ,k_{4}}^{9}
-\tfrac{1}{2}(k_{3})^{2}U_{k_{1},\ldots ,k_{4}}^{10},} \\
0= & k_{3}k_{4}U_{k_{1},\ldots ,k_{4}}^{1}-k_{1}k_{4}U_{k_{1},\ldots
,k_{4}}^{3}-k_{1}k_{3}U_{k_{1},\ldots ,k_{4}}^{4} \\
& \multicolumn{1}{r}{-k_{3}k_{4}U_{k_{1},\ldots
,k_{4}}^{5}+k_{2}k_{4}U_{k_{1},\ldots ,k_{4}}^{6}+k_{2}k_{3}U_{k_{1},\ldots
,k_{4}}^{7}+\left( (k_{1})^{2}-(k_{2})^{2}\right) U_{k_{1},\ldots
,k_{4}}^{9},} \\
0= & \tfrac{1}{2}(k_{4})^{2}\left( U_{k_{1},\ldots
,k_{4}}^{1}-U_{k_{1},\ldots ,k_{4}}^{5}-U_{k_{1},\ldots ,k_{4}}^{8}\right)
-k_{1}k_{4}U_{k_{1},\ldots ,k_{4}}^{4} \\
& \multicolumn{1}{r}{+k_{2}k_{4}U_{k_{1},\ldots
,k_{4}}^{7}+k_{3}k_{4}U_{k_{1},\ldots ,k_{4}}^{9}
+\tfrac{1}{2}\left( (k_{1})^{2}-(k_{2})^{2}-(k_{3})^{2}\right) 
U_{k_{1},\ldots ,k_{4}}^{10},}
\end{array}
\end{equation*}
for every system $(k_{1},\ldots ,k_{4})\in \mathbb{Z}^{4}$. Solving these
equations for $k_{2}\neq 0$, we obtain
\begin{eqnarray*}
U_{k_{1},\ldots ,k_{4}}^{1} &=&U_{k_{1},\ldots ,k_{4}}^{3}=U_{k_{1},\ldots
,k_{4}}^{4}=U_{k_{1},\ldots ,k_{4}}^{8}=U_{k_{1},\ldots
,k_{4}}^{9}=U_{k_{1},\ldots ,k_{4}}^{10} \\
&=&0, \\
U_{k_{1},\ldots ,k_{4}}^{2} &=&\tfrac{1}{2}\frac{k_{1}}{k_{2}}
U_{k_{1},\ldots ,k_{4}}^{5}, \\
U_{k_{1},\ldots ,k_{4}}^{6} &=&\tfrac{1}{2}\frac{k_{3}}{k_{2}}
U_{k_{1},\ldots ,k_{4}}^{5}, \\
U_{k_{1},\ldots ,k_{4}}^{7} &=&\tfrac{1}{2}\frac{k_{4}}{k_{2}}
U_{k_{1},\ldots ,k_{4}}^{5},
\end{eqnarray*}
and the unknowns $U_{k_{1},\ldots ,k_{4}}^{5}$ remain undetermined. If 
$k_{2}=0$ but $k_{4}\neq 0$, then the solutions to the previous equations are
\begin{eqnarray*}
U_{k_{1},\ldots ,k_{4}}^{1} &=&U_{k_{1},\ldots ,k_{4}}^{3}=U_{k_{1},\ldots
,k_{4}}^{4}=U_{k_{1},\ldots ,k_{4}}^{5} \\
&=&U_{k_{1},\ldots ,k_{4}}^{7}=U_{k_{1},\ldots ,k_{4}}^{8}=U_{k_{1},\ldots
,k_{4}}^{9}=U_{k_{1},\ldots ,k_{4}}^{10} \\
&=&0, \\
U_{k_{1},\ldots ,k_{4}}^{2} &=&\frac{k_{1}}{k_{4}}U_{k_{1},\ldots
,k_{4}}^{7},U_{k_{1},\ldots ,k_{4}}^{6}=\frac{k_{3}}{k_{4}}U_{k_{1},\ldots
,k_{4}}^{7},
\end{eqnarray*}
the unknowns $U_{k_{1},\ldots ,k_{4}}^{7}$ remaining undetermined. If 
$k_{2}=k_{4}=0$ but $k_{3}\neq 0$, then
\begin{eqnarray*}
U_{k_{1},\ldots ,k_{4}}^{4} &=&U_{k_{1},\ldots ,k_{4}}^{5}=U_{k_{1},\ldots
,k_{4}}^{7}=U_{k_{1},\ldots ,k_{4}}^{9}=U_{k_{1},\ldots ,k_{4}}^{10}=0, \\
U_{k_{1},\ldots ,k_{4}}^{1} &=&-\frac{k_{1}}{(k_{3})^{2}}\left(
k_{1}U_{k_{1},\ldots ,k_{4}}^{8}-2k_{3}U_{k_{1},\ldots ,k_{4}}^{3}\right) ,
\\
U_{k_{1},\ldots ,k_{4}}^{2} &=&\frac{k_{1}}{k_{3}}U_{k_{1},\ldots
,k_{4}}^{6},
\end{eqnarray*}
the unknowns $U_{k_{1},\ldots ,k_{4}}^{2}$, $U_{k_{1},\ldots ,k_{4}}^{6}$,
and $U_{k_{1},\ldots ,k_{4}}^{8}$ remaining undetermined. Finally, if 
$k_{2}=k_{3}=k_{4}=0$, then $U_{k_{1},\ldots ,k_{4}}^{A}=0$, $5\leq A\leq 10$, 
and the unknowns $U_{k_{1},\ldots ,k_{4}}^{A}$, $1\leq A\leq 4$ remain
undetermined. Therefore
\begin{equation*}
\begin{array}{ll}
U^{1}= & -\sum_{k_{2}=k_{4}=0,k_{3}\neq 0}\tfrac{k_{1}}{(k_{3})^{2}}\left(
k_{1}U_{k_{1},\ldots ,k_{4}}^{8}-2k_{3}U_{k_{1},\ldots ,k_{4}}^{3}\right)
\exp (ik_{j}x^{j})\smallskip \\
\multicolumn{1}{r}{} & \multicolumn{1}{r}{+U_{k_{1}000}^{1}\exp
(ik_{1}x^{1}),\medskip} \\
U^{2}= & \tfrac{1}{2}\sum\nolimits_{k_{2}\neq 0}
\tfrac{k_{1}}{k_{2}}U_{k_{1},\ldots ,k_{4}}^{5}\exp (ik_{j}x^{j})
+\sum_{k_{2}=0,k_{4}\neq 0}
\tfrac{k_{1}}{k_{4}}U_{k_{1},\ldots ,k_{4}}^{7}\exp (ik_{j}x^{j})
\smallskip
\\
\multicolumn{1}{r}{} & \multicolumn{1}{r}{+\sum_{k_{2}=k_{4}=0,k_{3}\neq 0}
\tfrac{k_{1}}{k_{3}}U_{k_{1},\ldots ,k_{4}}^{6}
\exp (ik_{j}x^{j})+U_{k_{1}000}^{2}
\exp (ik_{1}x^{1}),\medskip} \\
U^{3}= & \sum_{k_{2}=k_{4}=0}U_{k_{1},\ldots ,k_{4}}^{3}
\exp (ik_{j}x^{j}),\medskip \\
U^{4}= & U_{k_{1}000}^{4}\exp (ik_{1}x^{1}),\medskip \\
U^{5}= & \sum\nolimits_{k_{2}\neq 0}U_{k_{1},\ldots ,k_{4}}^{5}
\exp (ik_{j}x^{j}),\medskip \\
U^{6}= & \tfrac{1}{2}\sum\nolimits_{k_{2}\neq 0}\tfrac{k_{3}}{k_{2}}
U_{k_{1},\ldots ,k_{4}}^{5}\exp (ik_{j}x^{j})
+\sum_{k_{2}=0,k_{4}\neq 0}
\tfrac{k_{3}}{k_{4}}U_{k_{1},\ldots ,k_{4}}^{7}\exp (ik_{j}x^{j})
\smallskip
\\
\multicolumn{1}{r}{} & \multicolumn{1}{r}{+\sum_{k_{2}=k_{4}=0,k_{3}
\neq 0}U_{k_{1},\ldots ,k_{4}}^{6}\exp (ik_{j}x^{j}),\medskip} \\
U^{7}= & \tfrac{1}{2}\sum\nolimits_{k_{2}\neq 0}
\tfrac{k_{4}}{k_{2}}U_{k_{1},\ldots ,k_{4}}^{5}\exp (ik_{j}x^{j})
+\sum_{k_{2}=0,k_{4}\neq 0}U_{k_{1},\ldots ,k_{4}}^{7}
\exp (ik_{j}x^{j}),
\medskip \\
U^{8}= & \sum_{k_{2}=k_{4}=0,k_{3}\neq 0}U_{k_{1},\ldots ,k_{4}}^{8}
\exp (ik_{j}x^{j}),\medskip \\
U^{9}= & 0,\medskip \\
U^{10}= & 0.
\end{array}
\end{equation*}
Hence, by using the formulas (\ref{E's}) we obtain
\begin{align*}
\sum_{A=1}^{10}U^{A}E_{A}& =U_{k_{1},0,k_{3},0}^{8}\exp \left[
i(k_{1}x^{1}+k_{3}x^{3})\right] \left\{ E_{8}
-\tfrac{(k_{1})^{2}}{(k_{3})^{2}}E_{1}\right\} \\
& +U_{k_{1},0,k_{3},k_{4}}^{7}\exp \left[ i(k_{1}x^{1}+k_{3}x^{3}
+k_{4}x^{4})\right] \left\{ \tfrac{k_{1}}{k_{4}}E_{2}
+\tfrac{k_{3}}{k_{4}}E_{6}+E_{7}\right\} \\
& +U_{k_{1},0,k_{3},0}^{6}
\exp \left[ i(k_{1}x^{1}+k_{3}x^{3})\right]
\left\{ \tfrac{k_{1}}{k_{3}}E_{2}+E_{6}\right\} \\
& +U_{k_{1},\ldots ,k_{4}}^{5}\exp (ik_{j}x^{j})\left\{ 
\tfrac{1}{2}\tfrac{k_{1}}{k_{2}}E_{2}+E_{5}
+\tfrac{1}{2}\tfrac{k_{3}}{k_{2}}E_{6}
+\tfrac{1}{2}\tfrac{k_{4}}{k_{2}}E_{7}\right\} \\
& +U_{k_{1}000}^{4}\exp (ik_{1}x^{1})Y_{4}
+U_{k_{1},0,k_{3},0}^{3}
\exp \left[ i(k_{1}x^{1}+k_{3}x^{3})\right] 
\left\{ 2k_{3}E_{1}+E_{3}\right\} \\
& +U_{k_{1}000}^{2}\exp (ik_{1}x^{1})E_{2}+U_{k_{1}000}^{1}
\exp (ik_{1}x^{1})E_{1},
\end{align*}
and the vector fields
\begin{equation*}
\begin{array}{rll}
X_{1}^{k}= & \exp \left[ i(k_{1}x^{1}+k_{3}x^{3})\right] 
\left\{ \frac{\partial }{\partial y_{33}}
-\tfrac{(k_{1})^{2}}{(k_{3})^{2}}\frac{\partial }
{\partial y_{11}}\right\} , & k_{3}\neq 0, \\
X_{2}^{k}= & \exp \left[ i(k_{1}x^{1}+k_{3}x^{3}+k_{4}x^{4})\right] 
\left\{
\tfrac{k_{1}}{k_{4}}\frac{\partial }{\partial y_{12}}
+\tfrac{k_{3}}{k_{4}}
\frac{\partial }{\partial y_{23}}+\frac{\partial }{\partial y_{24}}\right\} ,
& k_{4}\neq 0, \\
X_{3}^{k}= & \exp \left[ i(k_{1}x^{1}+k_{3}x^{3})\right] 
\left\{ \tfrac{k_{1}}{k_{3}}
\frac{\partial }{\partial y_{12}}+\frac{\partial }{\partial y_{23}}
\right\} , & k_{3}\neq 0, \\
X_{4}^{k}= & \exp \left[ i(k_{1}x^{1}+k_{2}x^{2}+k_{3}x^{3}+k_{4}x^{4})
\right] & k_{2}\neq 0, \\
& \multicolumn{1}{r}{\cdot \left\{ \tfrac{1}{2}\tfrac{k_{1}}{k_{2}}
\frac{\partial }{\partial y_{12}}+\frac{\partial }{\partial y_{22}}
+\tfrac{1}{2}
\tfrac{k_{3}}{k_{2}}\frac{\partial }{\partial y_{23}}+\tfrac{1}{2}
\tfrac{k_{4}}{k_{2}}\frac{\partial }{\partial y_{24}}\right\} ,} &  \\
X_{5}^{k}= & \exp (ik_{1}x^{1})\frac{\partial }{\partial y_{14}}, &  \\
X_{6}^{k}= & \exp \left[ i(k_{1}x^{1}+k_{3}x^{3})\right] \left\{ 2k_{3}
\frac{\partial }{\partial y_{11}}+\frac{\partial }{\partial y_{13}}\right\} , &
\\
X_{7}^{k}= & \exp (ik_{1}x^{1})\frac{\partial }{\partial y_{12}}, &  \\
X_{8}^{k}= & \exp (ik_{1}x^{1})\frac{\partial }{\partial y_{11}}, &
\end{array}
\end{equation*}
with $k\in \mathbb{Z}^{4}$, span 
$T_{g}\mathcal{S}((\mathbb{R}/2\pi \mathbb{Z})^{4})$ topologically.
\end{example}
Let $\Lambda $ be a Lagrangian density on an arbitrary fibred manifold 
$p\colon E\to N$ and let $\Theta _{\Lambda }$ be the P-C form
associated to $\Lambda $. Let $X,Y\in T_{s}\mathcal{S}(N)$ be Jacobi vector
fields defined along an extremal $s\in \mathcal{S}(N)$ for the Lagrangian
density $\Lambda $. Then, 
$d[(j^{1}s)^{\ast }(i_{Y^{(1)}}i_{X^{(1)}}d\Theta _{\Lambda })]=0$ 
(e.g., see \cite{Garcia}); i.e., the $(n-1)$-form 
$i_{Y^{(1)}}i_{X^{(1)}}d\Theta _{\Lambda }$ is closed along $j^{1}s$.

The alternate bilinear mapping taking values in the space $Z^{n-1}(N)$ of
closed $(n-1)$-forms, defined by
\begin{equation*}
\begin{array}{l}
(\omega _{2})_{s}\colon T_{s}\mathcal{S}(N)\times T_{s}\mathcal{S}(N)\to Z^{n-1}(N),
\smallskip \\
(\omega _{2})_{s}(X,Y)
=(j^{1}s)^{\ast }\left( i_{Y^{(1)}}i_{X^{(1)}}d\Theta _{\Lambda }\right)
\end{array}
\end{equation*}
is called the presymplectic structure associated to $\Lambda $.
\begin{theorem}
\label{radical}Let $s$ be an extremal of a second-order Lagrangian density 
$\Lambda =Lv$ on $p\colon E\to N$ with Poincar\'{e}-Cartan form
projectable onto $J^{1}E$. Assume that the variational problem defined by 
$\Lambda $ is regular in the sense of \emph{Proposition \ref{regular1}}. For
every $x\in N$, let $R_{x}^{2}(\Lambda )\subseteq J_{x}^{2}(s^{\ast }V(p))$
be the vector subspace of $2$-jets $j_{x}^{2}X$ of $p$-vertical vector
fields along $s$ that satisfy the Jacobi equations \emph{(\ref{Jacobi1})} at
$x$. If the natural projection $p_{1}^{2}\colon R_{x}^{2}(\Lambda
)\to J_{x}^{1}(s^{\ast }V(p))$ is surjective for every $x\in N$,
then the radical of the valued $2$-form $(\omega _{2})_{s}$ vanishes.
\end{theorem}
\begin{proof}
According to (\ref{differential_P-C_bis}), we have $d\Theta _{\Lambda
}=(-1)^{i-1}dp_{\alpha }^{i}\wedge dy^{\alpha }\wedge v_{i}+dH\wedge v$. If 
$X^{(1)}=V^{\sigma }\tfrac{\partial }{\partial y^{\sigma }}
+\tfrac{\partial V^{\sigma }}{\partial x^{j}}
\tfrac{\partial }{\partial y_{j}^{\sigma }}$, 
$Y^{(1)}=W^{\sigma }\tfrac{\partial }{\partial y^{\sigma }}
+\tfrac{\partial W^{\sigma }}{\partial x^{j}}
\tfrac{\partial }{\partial y_{j}^{\sigma }}$,
with $V^{\sigma },W^{\sigma }\in C^{\infty }(N)$, then
\begin{eqnarray*}
(\omega _{2})_{s}(X,Y) 
&=&(-1)^{i-1}\left\{ \left( V^{\sigma }W^{\alpha}
-V^{\alpha }W^{\sigma }\right) 
\left( \tfrac{\partial p_{\alpha }^{i}}{\partial y^{\sigma }}
\circ j^{1}s
\right) \right. \\
&&\left. \left. 
+\left( 
\tfrac{\partial V^{\sigma }}{\partial x^{j}}W^{\alpha }
-V^{\alpha }\tfrac{\partial W^{\sigma }}{\partial x^{j}}\right)
\left( \tfrac{\partial p_{\alpha }^{i}}{\partial y_{j}^{\sigma }}
\circ j^{1}s\right) \right\} \right\vert _{j^{1}s}v_{i}.
\end{eqnarray*}
If we assume the vector field $X$ belongs to 
$\operatorname{rad}(\omega _{2})_{s}$, then by evaluating 
at $x$ the equation $(\omega _{2})_{s}(X,Y)=0$, 
$\forall Y\in T_{s}\mathcal{S}(N)$, we obtain
\begin{equation}
\begin{array}{ll}
0= & \left[ V^{\sigma }(x)W^{\alpha }(x)-V^{\alpha }(x)W^{\sigma }(x)\right]
\tfrac{\partial p_{\alpha }^{i}}{\partial y^{\sigma }}(j_{x}^{1}s) \\
& +\left[ \tfrac{\partial V^{\sigma }}{\partial x^{j}}(x)W^{\alpha}(x)
-V^{\alpha }(x)\tfrac{\partial W^{\sigma }}{\partial x^{j}}(x)\right]
\tfrac{\partial p_{\alpha }^{i}}{\partial y_{j}^{\sigma }}(j_{x}^{1}s),
\quad 1\leq i\leq n.
\end{array}
\label{rad_x}
\end{equation}
The assumption in the statement implies that given arbitrary values for 
$W^{\beta }(x)$ and $\tfrac{\partial W^{\beta }}{\partial x^{h}}(x)$, there
exists an element $j_{x}^{2}Y\in R_{x}^{2}(\Lambda )$ projecting under the
natural mapping $p_{1}^{2}\colon R_{x}^{2}(\Lambda )\to J_{x}^{1}(s^{\ast }V(p))$ 
onto the $1$-jet at $x$ the coordinates of which
coincide with these values. Accordingly, the coefficients of $W^{\beta }(x)$
and $\tfrac{\partial W^{\beta }}{\partial x^{h}}(x)$\ in (\ref{rad_x}) must
vanish, i.e.,
\begin{equation}
\begin{array}{l}
0=V^{\alpha }\left( \tfrac{\partial p_{\beta }^{i}}{\partial y^{\alpha }}
\circ j^{1}s\right) -V^{\alpha }\left( \tfrac{\partial p_{\alpha }^{i}}
{\partial y^{\beta }}\circ j^{1}s\right) 
+\tfrac{\partial V^{\alpha }}{\partial x^{j}}
\left( 
\tfrac{\partial p_{\beta }^{i}}{\partial y_{j}^{\alpha }}\circ j^{1}s
\right) ,
\smallskip \\
\multicolumn{1}{r}{1\leq i\leq n,\;1\leq \beta \leq m,\medskip} \\
0=V^{\alpha }\left( \tfrac{\partial p_{\alpha }^{i}}
{\partial y_{h}^{\beta }}
\circ j^{1}s\right) ,\qquad \qquad \qquad \qquad \quad h,i=1,\ldots
,n,\;1\leq \beta \leq m,
\end{array}
\label{rad1}
\end{equation}
as the point $x$ is arbitrary. Hence the formulas (\ref{rad1}) are the
equations for the radical of $(\omega _{2})_{s}$. If we set
\begin{equation*}
V=(V^{1},\ldots ,V^{m}),\quad O_{m}=(0,\ldots ,0),
\quad 
\Upsilon 
=\left(
\tfrac{\partial p_{\alpha }^{i}}{\partial y_{h}^{\beta }}
\right) _{1\leq h\leq n,1\leq \beta \leq m}^{1\leq i\leq n,1\leq \alpha \leq m},
\end{equation*}
then the second group of equations in (\ref{rad1}) can matricially be
written as
\begin{equation*}
\underset{n\text{ times}}{\underbrace{\left( V,\ldots ,V\right) }}\cdot
\left( \Upsilon \circ j^{1}s\right) 
=\underset{n\text{ times}}{\underbrace{\left( O_{m},\ldots ,O_{m}\right) }.}
\end{equation*}
If the variational problem defined by the density $\Lambda $ is regular in
the sense of Proposition \ref{regular1}, then\ $\det \Upsilon \neq 0$; Hence
$V=0$.
\end{proof}
\begin{criterion}
Next, we give a criterion in order to ensure that the condition of Theorem 
\ref{radical}\ holds. According to (\ref{momenta}) we have 
$p_{\alpha }^{i}=\frac{\partial \bar{L}}{\partial y_{i}^{\alpha }}$, 
where $\bar{L}$ is the first-order Lagrangian defined by (\ref{barL}), 
also see Theorem \ref{first-order-equivalent}. As is known, 
the Hessian metric of $\bar{L}$ is the section of the vector bundle 
$S^{2}V^{\ast }(p_{0}^{1})$ locally given by,
$\operatorname{Hess}(\bar{L})
=\tfrac{\partial ^{2}\bar{L}}{\partial y_{i}^{\alpha }\partial y_{j}^{\beta }}
d_{10}y_{i}^{\alpha }\otimes d_{10}y_{j}^{\beta }$. 
As mentioned in section \ref{Hamilton_formalism} there is a canonical isomorphism 
\begin{equation*}
\begin{array}{l}
I\colon (p_{0}^{1})^{\ast }\left( p^{\ast }(T^{\ast }N)\otimes V(p)\right)
\to V(p_{0}^{1}), \\
I\left( j_{x}^{1}s,(dx^{i})_{x}\otimes \left( 
\tfrac{\partial }{\partial y^{\alpha }}\right) _{s(x)}\right) 
=\left( \tfrac{\partial }{\partial y_{i}^{\alpha }}\right) _{j_{x}^{1}s},
\end{array}
\end{equation*}
and dually,
\begin{equation*}
\begin{array}{l}
I^{\ast }\colon V^{\ast }(p_{0}^{1})\to (p_{0}^{1})^{\ast }
\left( p^{\ast }(TN)\otimes V^{\ast }(p)\right) , \\
I^{\ast }\left( j_{x}^{1}s,\left( \frac{\partial }{\partial x^{i}}
\right)
_{x}\otimes \left( dy^{\alpha }\right) _{s(x)}\right)
=\left( d_{10}y_{i}^{\alpha }\right) _{j_{x}^{1}s}.
\end{array}
\end{equation*}
Hence the Hessian metric can be viewed as a symmetric bilinear form
\begin{equation*}
\operatorname{Hess}(\bar{L})_{j_{x}^{1}s}\colon \!
V_{j_{x}^{1}s}(p_{0}^{1})\times V_{j_{x}^{1}s}(p_{0}^{1})
\cong \left[ (T_{x}^{\ast }N)\otimes V_{s(x)}(p)\right] 
\!\times \!\left[ (T_{x}^{\ast }N)\otimes V_{s(x)}(p)
\right]
\to \mathbb{R},
\end{equation*}
and we can define a linear map as follows:
\begin{equation*}
\begin{array}{l}
\operatorname{Hess}(\bar{L})_{j_{x}^{1}s}^{\natural }
\colon (T_{x}^{\ast }N)\otimes (T_{x}^{\ast }N)
\otimes V_{s(x)}(p)\to V_{s(x)}^{\ast }(p), \\
\operatorname{Hess}(\bar{L})_{j_{x}^{1}s}^{\natural }
\left(
w_{1},w_{2},X_{1}\right) (X_{2})
=\operatorname{Hess}(\bar{L})_{j_{x}^{1}s}
\left(
w_{1}\otimes X_{1},w_{2}\otimes X_{2}\right) , \\
\forall w_{1},w_{2}\in T_{x}^{\ast }N,
\quad \forall X_{1},X_{2}\in V_{s(x)}(p).
\end{array}
\end{equation*}
The matrix of $\left. \operatorname{Hess}(\bar{L})_{j_{x}^{1}s}^{\natural
}\right\vert _{S^{2}(T_{x}^{\ast }N)\otimes V_{s(x)}(p)}$ is $\Upsilon
^{\natural }=\left( \!\frac{1}{1+\delta _{ij}}\left[ \tfrac{\partial
p_{\alpha }^{i}}{\partial y_{j}^{\gamma }}
+\tfrac{\partial p_{\alpha }^{j}}{\partial y_{i}^{\gamma }}\right] 
\!(j_{x}^{1}s)\!\right) _{\gamma ,i\leq j}^{\alpha }$ 
in the standard basis. Moreover, letting 
$v_{ij}^{\gamma }=\tfrac{\partial ^{2}V^{\gamma }}
{\partial x^{i}\partial x^{j}}(x)$, and
denoting by $E^{\alpha }$ the right-hand side of the formula (\ref{Jacobi1}), 
this formula, evaluated at $x$, reads as follows: 
$v_{ij}^{\gamma }\tfrac{\partial p_{\alpha }^{i}}{\partial y_{j}^{\gamma }}
(j_{x}^{1}s)=E^{\alpha }(x)$, which is a linear system with $m$ equations 
in the $\frac{m}{2}n(n+1)$ unknowns $v_{ij}^{\gamma }$, 
$1\leq i\leq j\leq n$, and the matrix of this system is precisely 
$\Upsilon ^{\natural }$. Consequently, if 
$\operatorname{Hess}(\bar{L})_{j_{x}^{1}s}^{\natural }$ 
is assumed to be surjective, then the previous system is compatible.
\end{criterion}
\begin{corollary}
The radical of the valued $2$-form $(\omega _{2})_{g}$ corresponding to the
E-H Lagrangian density along an arbitrary extremal metric $g$, vanishes.
\end{corollary}
\begin{proof}
According to Theorem \ref{radical}, in order to prove the corollary above,
we need only to verify that the projection $p_{1}^{2}\colon
R_{x}^{2}(\Lambda )\to J_{x}^{1}(s^{\ast }V(p))$ is surjective for
every $x\in N$. By considering a system of normal coordiantes for the metric
$g$ at the point $x$, and letting $v^{ab}=V^{ab}(x)$, 
$v_{ij}^{ab}=\tfrac{\partial ^{2}V^{ab}}{\partial x^{i}\partial x^{j}}(x)$, 
the equations (\ref{JacobiEH}) evaluated at $x$, are written as follows:
\begin{equation*}
\begin{array}{rl}
0= & \tfrac{1}{2}\left[ \varepsilon _{i}\left( \delta _{a\nu }\delta _{j\mu}
+\delta _{a\mu }\delta _{\nu j}\right) \delta ^{ib}
-\varepsilon _{i}\delta^{ij}\delta _{a\nu }\delta _{b\mu }
-\varepsilon _{a}\delta ^{ab}\delta_{i\nu }\delta _{j\mu }\right] 
v_{ij}^{ab}
\smallskip \\
& \multicolumn{1}{r}{+\varepsilon _{b}(R^{g})_{\mu \nu b}^{a}(x)v^{ab}} \\
= & \tfrac{\varepsilon _{i}}{2}\left( v_{i\mu }^{\nu i}+v_{i\nu }^{\mu i}
-v_{\mu \nu }^{ii}-v_{ii}^{\mu \nu }\right) 
+\varepsilon _{b}(R^{g})_{\mu \nu b}^{a}(x)v^{ab}, \\
& 1\leq \mu \leq \nu \leq n,
\end{array}
\end{equation*}
which is a system with $\frac{1}{2}n(n+1)$ equations in the
$\frac{1}{4}n^{2}(n+1)^{2}$ unknowns $v_{\mu \nu }^{ij}$, 
$1\leq i\leq j\leq n$, $1\leq \mu \leq \nu \leq n$, with 
$v_{\mu \nu }^{ij}=v_{\nu \mu }^{ij}=v_{\mu \nu }^{ji}$, and where 
the scalars $v^{ab}$, $1\leq a\leq b\leq n$, can take
arbitrary values. A particular solution to this system is obtained 
by letting,
\begin{equation}
\left.
\begin{array}{rl}
\text{(i)} & \varepsilon _{i}v_{\mu \nu }^{ii}=\varepsilon _{i}v_{ii}^{\mu
\nu }=0\smallskip \\
\text{(ii)} & \varepsilon _{i}v_{i\mu }^{\nu i}=\varepsilon _{i}v_{i\nu
}^{\mu i}=-\varepsilon _{b}(R^{g})_{\mu \nu b}^{a}(x)v^{ab}
\end{array}
\right\} ,\quad 1\leq \mu \leq \nu \leq n.  \label{eqns}
\end{equation}
The equations (\ref{eqns})-(i) hold by setting $v_{\mu \nu
}^{ii}=v_{ii}^{\mu \nu }$, $\forall i,\mu ,\nu =1,\ldots ,n,\;\mu \leq \nu $
, while the equations (\ref{eqns})-(ii) hold by setting
\begin{equation*}
\left.
\begin{array}{lll}
& v_{i\nu }^{\mu i}=v_{\mu i}^{\nu i}=0, & 2\leq i\leq n \\
& v_{1\nu }^{1\mu }=v_{1\mu }^{1\nu }=-\varepsilon _{1}\varepsilon
_{b}(R^{g})_{\mu \nu b}^{a}(x)v^{ab}, &
\end{array}
\right\} ,\quad 1\leq \mu \leq \nu \leq n.
\end{equation*}
\end{proof}

\begin{example}
\label{example3}Below, we compute the presymplectic structure associated to
Example \ref{example2}; i.e., we compute $(\omega _{2})_{g}$ for the E-H
Lagrangian density when $N=(\mathbb{R}/2\pi \mathbb{Z})^{4}$ and 
$g=\varepsilon _{i}(dx^{i})^{2}$, $\varepsilon _{1}=-1$, $\varepsilon
_{2}=\varepsilon _{3}=\varepsilon _{4}=+1$ by using the basis $X_{h}^{k}$, 
$1\leq h\leq 8$, $k\in \mathbb{Z}^{4}$ of that example. We follow some ideas
in \cite[Section 7]{Szczyrba} for our particular case.

According to the previous notations and calculations, we have
\begin{eqnarray*}
(\omega _{2})_{g}(X,Y) &=&(-1)^{i-1}\left\{ \left(
V^{kl}W^{ab}-V^{ab}W^{kl}\right) \left( 
\tfrac{\partial p_{ab}^{i}}{\partial y_{kl}}\circ j^{1}g\right) \right. \\
&&\left. \left. +\left( \tfrac{\partial V^{kl}}{\partial x^{j}}W^{ab}
-V^{ab}\tfrac{\partial W^{kl}}{\partial x^{j}}
\right) 
\left( \tfrac{\partial p_{ab}^{i}}{\partial y_{kl,j}}\circ j^{1}g
\right) 
\right\} 
\right\vert _{j^{1}g}v_{i},
\end{eqnarray*}
\begin{equation*}
X^{(1)}=V^{ab}\tfrac{\partial }{\partial y_{ab}}
+\tfrac{\partial V^{ab}}
{\partial x^{j}}\tfrac{\partial }{\partial y_{ab,j}},\;Y^{(1)}
=W^{ab}\tfrac{\partial }{\partial y_{ab}}
+\tfrac{\partial W^{ab}}{\partial x^{j}}
\tfrac{\partial }{\partial y_{ab,j}},
\end{equation*}
and from the formulas (\ref{p^j_mr}), (\ref{Y's}) it follows:
\begin{equation*}
\tfrac{\partial p_{kl}^{i}}{\partial y_{uv}}\circ j^{1}g=0.
\end{equation*}
Therefore
\begin{equation*}
(\omega _{2})_{g}(X,Y)=(-1)^{i-1}\omega _{2}^{i}(X,Y)v_{i},
\end{equation*}
where
\begin{equation*}
\omega _{2}^{i}(X,Y)=\sum_{k\leq l}\sum_{a\leq b}
\left( 
\tfrac{\partial p_{ab}^{i}}{\partial y_{kl,j}}\circ j^{1}g
\right) 
\left( \tfrac{\partial V^{kl}}{\partial x^{j}}W^{ab}
-V^{ab}\tfrac{\partial W^{kl}}{\partial x^{j}}
\right) ,
\end{equation*}
and, as a computation shows, the scalar differential forms $\omega _{2}^{i}$
are given by{\small
\begin{eqnarray*}
\omega _{2}^{1}\!\!\!\! &=&\!\!\!\!\tfrac{1}{2}\left( \!\tfrac{\partial
W^{13}}{\partial x^{3}}\!+\!\tfrac{\partial W^{14}}{\partial x^{4}}
\!+\!
\tfrac{\partial W^{12}}{\partial x^{2}}\!\right) V^{11}
\!-\!\tfrac{1}{2}\left( \!\tfrac{\partial W^{11}}{\partial x^{2}}
\!+\!\tfrac{\partial W^{33}}{\partial x^{2}}\!+\!
\tfrac{\partial W^{22}}{\partial x^{2}}
\!+\!\tfrac{\partial W^{44}}{\partial x^{2}}\!\right) V^{12} \\
&&\!\!\!\!\!\!-\tfrac{1}{2}\left( \!\tfrac{\partial W^{11}}{\partial x^{3}}
\!+\!\tfrac{\partial W^{22}}{\partial x^{3}}
\!+\!\tfrac{\partial W^{33}}{\partial x^{3}}
\!+\!\tfrac{\partial W^{44}}{\partial x^{3}}\!\right)
V^{13}\!-\!\tfrac{1}{2}\left( \!\tfrac{\partial W^{11}}{\partial x^{4}}
\!+\!\tfrac{\partial W^{33}}{\partial x^{4}}
\!+\!\tfrac{\partial W^{22}}{\partial x^{4}}
\!+\!\tfrac{\partial W^{44}}{\partial x^{4}}\!\right) V^{14} \\
&&\!\!\!\!\!\!+\left( \!\tfrac{\partial W^{12}}{\partial x^{3}}
\!+\!\tfrac{\partial W^{13}}{\partial x^{2}}
\!-\!\tfrac{\partial W^{23}}{\partial x^{1}}\!\right) V^{23}
\!+\!\tfrac{1}{2}\left( \!\tfrac{\partial W^{33}}{\partial x^{1}}
\!+\!\tfrac{\partial W^{44}}{\partial x^{1}}\!+\!\tfrac{\partial W^{12}%
}{\partial x^{2}}\!-\!\tfrac{\partial W^{13}}{\partial x^{3}}\!-\!\tfrac{%
\partial W^{14}}{\partial x^{4}}\!\right) V^{22} \\
&&\!\!\!\!\!\!+\left( \!\tfrac{\partial W^{12}}{\partial x^{4}}\!+\!\tfrac{%
\partial W^{14}}{\partial x^{2}}\!-\!\tfrac{\partial W^{24}}{\partial x^{1}}%
\!\right) V^{24}\!+\!\tfrac{1}{2}\left( \!\tfrac{\partial W^{44}}{\partial
x^{1}}\!+\!\tfrac{\partial W^{22}}{\partial x^{1}}\!-\!\tfrac{\partial W^{12}%
}{\partial x^{2}}\!+\!\tfrac{\partial W^{13}}{\partial x^{3}}\!-\!\tfrac{%
\partial W^{14}}{\partial x^{4}}\!\right) V^{33} \\
&&\!\!\!\!\!\!+\left( \!\tfrac{\partial W^{13}}{\partial x^{4}}\!+\!\tfrac{%
\partial W^{14}}{\partial x^{3}}\!-\!\tfrac{\partial W^{34}}{\partial x^{1}}%
\right) V^{34}\!+\!\tfrac{1}{2}\left( \!\tfrac{\partial W^{33}}{\partial
x^{1}}\!+\!\tfrac{\partial W^{22}}{\partial x^{1}}\!-\!\tfrac{\partial W^{12}%
}{\partial x^{2}}\!-\!\tfrac{\partial W^{13}}{\partial x^{3}}\!+\!\tfrac{%
\partial W^{14}}{\partial x^{4}}\!\right) V^{44} \\
&&\!\!\!\!\!\!-\tfrac{1}{2}\left( \!\tfrac{\partial V^{13}}{\partial x^{3}}%
\!+\!\tfrac{\partial V^{12}}{\partial x^{2}}\!+\!\tfrac{\partial V^{14}}{%
\partial x^{4}}\!\right) W^{11}\!+\!\tfrac{1}{2}\left( \!\tfrac{\partial
V^{22}}{\partial x^{2}}\!+\!\tfrac{\partial V^{44}}{\partial x^{2}}\!+\!%
\tfrac{\partial V^{11}}{\partial x^{2}}\!+\!\tfrac{\partial V^{33}}{\partial
x^{2}}\!\right) W^{12} \\
&&\!\!\!\!\!\!+\tfrac{1}{2}\left( \!\tfrac{\partial V^{11}}{\partial x^{3}}%
\!+\!\tfrac{\partial V^{33}}{\partial x^{3}}\!+\!\tfrac{\partial V^{44}}{%
\partial x^{3}}\!+\!\tfrac{\partial V^{22}}{\partial x^{3}}\!\right)
W^{13}\!+\!\tfrac{1}{2}\left( \!\tfrac{\partial V^{22}}{\partial x^{4}}\!+\!%
\tfrac{\partial V^{33}}{\partial x^{4}}\!+\!\tfrac{\partial V^{11}}{\partial
x^{4}}\!+\!\tfrac{\partial V^{44}}{\partial x^{4}}\!\right) W^{14} \\
&&\!\!\!\!\!\!+\!\left( \!\tfrac{\partial V^{23}}{\partial x^{1}}\!-\!\tfrac{%
\partial V^{13}}{\partial x^{2}}\!-\!\tfrac{\partial V^{12}}{\partial x^{3}}%
\!\right) W^{23}\!-\!\tfrac{1}{2}\left( \!\tfrac{\partial V^{33}}{\partial
x^{1}}\!+\!\tfrac{\partial V^{44}}{\partial x^{1}}\!+\!\tfrac{\partial V^{12}%
}{\partial x^{2}}\!-\!\tfrac{\partial V^{13}}{\partial x^{3}}\!-\!\tfrac{%
\partial V^{14}}{\partial x^{4}}\!\right) W^{22}\! \\
&&\!\!\!\!\!\!+\left( \!\tfrac{\partial V^{24}}{\partial x^{1}}\!-\!\tfrac{%
\partial V^{14}}{\partial x^{2}}\!-\!\tfrac{\partial V^{12}}{\partial x^{4}}%
\!\right) W^{24}\!-\!\tfrac{1}{2}\left( \!\tfrac{\partial V^{22}}{\partial
x^{1}}\!+\!\tfrac{\partial V^{44}}{\partial x^{1}}\!-\!\tfrac{\partial V^{12}%
}{\partial x^{2}}\!+\!\tfrac{\partial V^{13}}{\partial x^{3}}\!-\!\tfrac{%
\partial V^{14}}{\partial x^{4}}\!\right) W^{33} \\
&&\!\!\!\!\!\!+\left( \!\tfrac{\partial V^{34}}{\partial x^{1}}\!-\!\tfrac{%
\partial V^{14}}{\partial x^{3}}\!-\!\tfrac{\partial V^{13}}{\partial x^{4}}%
\!\right) W^{34}\!-\!\tfrac{1}{2}\left( \!\tfrac{\partial V^{22}}{\partial
x^{1}}\!+\!\tfrac{\partial V^{33}}{\partial x^{1}}\!-\!\tfrac{\partial V^{12}%
}{\partial x^{2}}\!-\!\tfrac{\partial V^{13}}{\partial x^{3}}\!+\!\tfrac{%
\partial V^{14}}{\partial x^{4}}\!\right) W^{44},
\end{eqnarray*}%
} {\small
\begin{eqnarray*}
\omega _{2}^{2}\!\!\!\! &=&\!\!\!\!\tfrac{1}{2}\left( \!-\!\tfrac{\partial
W^{12}}{\partial x^{1}}\!-\tfrac{\partial W^{23}}{\partial x^{3}}\!+\!\tfrac{%
\partial W^{44}}{\partial x^{2}}\!+\!\tfrac{\partial W^{33}}{\partial x^{2}}%
\!-\!\tfrac{\partial W^{24}}{\partial x^{4}}\!\right) V^{11}\! \\
&&\!\!\!\!+\tfrac{1}{2}\left( \!\tfrac{\partial V^{12}}{\partial x^{1}}\!+\!%
\tfrac{\partial V^{23}}{\partial x^{3}}\!-\!\tfrac{\partial V^{44}}{\partial
x^{2}}\!-\!\tfrac{\partial V^{33}}{\partial x^{2}}\!+\!\tfrac{\partial V^{24}%
}{\partial x^{4}}\!\right) W^{11} \\
&&\!\!\!\!+\!\tfrac{1}{2}\left( \!\tfrac{\partial W^{11}}{\partial x^{1}}%
\!+\!\tfrac{\partial W^{22}}{\partial x^{1}}\!-\!\tfrac{\partial W^{33}}{%
\partial x^{1}}\!-\!\tfrac{\partial W^{44}}{\partial x^{1}}\!\right)
V^{12}-\!\tfrac{1}{2}\left( \!\tfrac{\partial V^{11}}{\partial x^{1}}\!+\!%
\tfrac{\partial V^{22}}{\partial x^{1}}\!-\!\tfrac{\partial V^{33}}{\partial
x^{1}}\!-\!\tfrac{\partial V^{44}}{\partial x^{1}}\!\right) W^{12} \\
&&\!\!\!\!+\left( \!\tfrac{\partial W^{12}}{\partial x^{3}}\!+\!\tfrac{%
\partial W^{23}}{\partial x^{1}}\!-\!\tfrac{\partial W^{13}}{\partial x^{2}}%
\!\right) V^{13}-\left( \!\tfrac{\partial V^{12}}{\partial x^{3}}\!+\!\tfrac{%
\partial V^{23}}{\partial x^{1}}\!-\!\tfrac{\partial V^{13}}{\partial x^{2}}%
\!\right) W^{13} \\
&&\!\!\!\!\!+\!\left( \!\tfrac{\partial W^{24}}{\partial x^{1}}\!-\!\tfrac{%
\partial W^{14}}{\partial x^{2}}\!+\!\tfrac{\partial W^{12}}{\partial x^{4}}%
\!\right) V^{14}-\!\left( \!\tfrac{\partial V^{24}}{\partial x^{1}}\!-\!%
\tfrac{\partial V^{14}}{\partial x^{2}}\!+\!\tfrac{\partial V^{12}}{\partial
x^{4}}\!\right) W^{14} \\
&&\!\!\!\!\!+\tfrac{1}{2}\left( \!-\tfrac{\partial W^{12}}{\partial x^{1}}%
\!+\!\tfrac{\partial W^{24}}{\partial x^{4}}\!+\!\tfrac{\partial W^{23}}{%
\partial x^{3}}\!\right) V^{22}-\tfrac{1}{2}\left( \!-\tfrac{\partial V^{12}%
}{\partial x^{1}}\!+\!\tfrac{\partial V^{24}}{\partial x^{4}}\!+\!\tfrac{%
\partial V^{23}}{\partial x^{3}}\!\right) W^{22} \\
&&\!\!\!\!\!\!\!-\!\tfrac{1}{2}\left( \!\tfrac{\partial W^{11}}{\partial
x^{3}}\!+\!\tfrac{\partial W^{22}}{\partial x^{3}}\!-\!\tfrac{\partial W^{33}%
}{\partial x^{3}}\!-\!\tfrac{\partial W^{44}}{\partial x^{3}}\!\right)
V^{23}+\!\tfrac{1}{2}\left( \!\tfrac{\partial V^{11}}{\partial x^{3}}\!+\!%
\tfrac{\partial V^{22}}{\partial x^{3}}\!-\!\tfrac{\partial V^{33}}{\partial
x^{3}}\!-\!\tfrac{\partial V^{44}}{\partial x^{3}}\!\right) W^{23} \\
&&\!\!\!\!\!\!-\tfrac{1}{2}\left( \!\tfrac{\partial W^{11}}{\partial x^{4}}%
\!+\!\tfrac{\partial W^{22}}{\partial x^{4}}\!-\!\tfrac{\partial W^{33}}{%
\partial x^{4}}\!-\!\tfrac{\partial W^{44}}{\partial x^{4}}\!\right) V^{24}+%
\tfrac{1}{2}\left( \!\tfrac{\partial V^{11}}{\partial x^{4}}\!+\!\tfrac{%
\partial V^{22}}{\partial x^{4}}\!-\!\tfrac{\partial V^{33}}{\partial x^{4}}%
\!-\!\tfrac{\partial V^{44}}{\partial x^{4}}\!\right) W^{24} \\
&&\!\!\!\!\!\!+\!\tfrac{1}{2}\left( \!-\!\tfrac{\partial W^{23}}{\partial
x^{3}}\!-\!\tfrac{\partial W^{12}}{\partial x^{1}}\!+\!\tfrac{\partial W^{24}%
}{\partial x^{4}}\!-\!\tfrac{\partial W^{44}}{\partial x^{2}}\!+\!\tfrac{%
\partial W^{11}}{\partial x^{2}}\!\right) V^{33} \\
&&\!\!\!\!\!\!-\!\tfrac{1}{2}\left( \!-\!\tfrac{\partial V^{23}}{\partial
x^{3}}\!-\!\tfrac{\partial V^{12}}{\partial x^{1}}\!+\!\tfrac{\partial V^{24}%
}{\partial x^{4}}\!-\!\tfrac{\partial V^{44}}{\partial x^{2}}\!+\!\tfrac{%
\partial V^{11}}{\partial x^{2}}\!\right) W^{33} \\
&&\!\!\!\!\!\!+\left( -\!\tfrac{\partial W^{23}}{\partial x^{4}}\!-\!\tfrac{%
\partial W^{24}}{\partial x^{3}}\!+\!\tfrac{\partial W^{34}}{\partial x^{2}}%
\!\right) V^{34}\!-\!\left( -\!\tfrac{\partial V^{23}}{\partial x^{4}}\!-\!%
\tfrac{\partial V^{24}}{\partial x^{3}}\!+\!\tfrac{\partial V^{34}}{\partial
x^{2}}\!\right) W^{34}\! \\
&&\!\!\!\!\!\!+\!\tfrac{1}{2}\left( \!-\!\tfrac{\partial W^{24}}{\partial
x^{3}}\!-\!\tfrac{\partial W^{12}}{\partial x^{1}}\!-\!\tfrac{\partial W^{33}%
}{\partial x^{2}}\!+\!\tfrac{\partial W^{23}}{\partial x^{3}}\!+\!\tfrac{%
\partial W^{11}}{\partial x^{2}}\!\right) V^{44} \\
&&\!\!\!\!\!\!-\!\tfrac{1}{2}\left( \!-\!\tfrac{\partial V^{24}}{\partial
x^{3}}\!-\!\tfrac{\partial V^{12}}{\partial x^{1}}\!-\!\tfrac{\partial V^{33}%
}{\partial x^{2}}\!+\!\tfrac{\partial V^{23}}{\partial x^{3}}\!+\!\tfrac{%
\partial V^{11}}{\partial x^{2}}\!\right) W^{44},
\end{eqnarray*}%
} {\small
\begin{eqnarray*}
\omega _{2}^{3}\!\!\!\! &=&\!\!\!\!\!\tfrac{1}{2}\left( -\!\tfrac{\partial
W^{23}}{\partial x^{2}}\!+\!\tfrac{\partial W^{22}}{\partial x^{3}}\!+\!%
\tfrac{\partial W^{44}}{\partial x^{3}}\!-\!\tfrac{\partial W^{13}}{\partial
x^{1}}\!-\!\tfrac{\partial W^{34}}{\partial x^{4}}\!\right) V^{11}\! \\
&&\!\!\!\!\!-\tfrac{1}{2}\left( -\!\tfrac{\partial V^{23}}{\partial x^{2}}\!-%
\tfrac{\partial V^{22}}{\partial x^{3}}\!+\!\tfrac{\partial V^{44}}{\partial
x^{3}}\!-\!\tfrac{\partial V^{13}}{\partial x^{1}}\!-\!\tfrac{\partial V^{34}%
}{\partial x^{4}}\!\right) W^{11}\! \\
&&\!\!\!\!\!+\left( \!\tfrac{\partial W^{13}}{\partial x^{2}}\!-\tfrac{%
\partial W^{12}}{\partial x^{3}}\!+\!\tfrac{\partial W^{23}}{\partial x^{1}}%
\!\right) V^{12}\!\!-\left( \!\tfrac{\partial V^{13}}{\partial x^{2}}\!-%
\tfrac{\partial V^{12}}{\partial x^{3}}\!+\!\tfrac{\partial V^{23}}{\partial
x^{1}}\!\right) W^{12}\! \\
&&\!\!\!\!\!+\tfrac{1}{2}\left( \!\tfrac{\partial W^{33}}{\partial x^{1}}%
\!-\!\tfrac{\partial W^{44}}{\partial x^{1}}\!-\!\tfrac{\partial W^{22}}{%
\partial x^{1}}\!+\!\tfrac{\partial W^{11}}{\partial x^{1}}\!\right) V^{13}-%
\tfrac{1}{2}\left( \!\tfrac{\partial V^{33}}{\partial x^{1}}\!-\!\tfrac{%
\partial V^{44}}{\partial x^{1}}\!-\!\tfrac{\partial V^{22}}{\partial x^{1}}%
\!+\!\tfrac{\partial V^{11}}{\partial x^{1}}\!\right) W^{13}\! \\
&&\!\!\!\!\!-\left( \!\tfrac{\partial W^{14}}{\partial x^{3}}\!+\!\tfrac{%
\partial W^{13}}{\partial x^{4}}\!+\!\tfrac{\partial W^{34}}{\partial x^{1}}%
\!\right) V^{14}\!\!+\left( \!\tfrac{\partial V^{14}}{\partial x^{3}}\!+\!%
\tfrac{\partial V^{13}}{\partial x^{4}}\!+\!\tfrac{\partial V^{34}}{\partial
x^{1}}\!\right) W^{14}\! \\
&&\!\!\!\!\!+\tfrac{1}{2}\left( \!\tfrac{\partial W^{34}}{\partial x^{4}}%
\!-\!\tfrac{\partial W^{23}}{\partial x^{2}}\!+\!\tfrac{\partial W^{11}}{%
\partial x^{3}}\!-\!\tfrac{\partial W^{13}}{\partial x^{1}}\!-\!\tfrac{%
\partial W^{44}}{\partial x^{3}}\!\right) V^{22}\! \\
&&\!\!\!\!\!-\tfrac{1}{2}\left( \!\tfrac{\partial V^{34}}{\partial x^{4}}%
\!-\!\tfrac{\partial V^{23}}{\partial x^{2}}\!+\!\tfrac{\partial V^{11}}{%
\partial x^{3}}\!-\!\tfrac{\partial V^{13}}{\partial x^{1}}\!-\!\tfrac{%
\partial V^{44}}{\partial x^{3}}\!\right) W^{22}\! \\
&&\!\!\!\!\!+\tfrac{1}{2}\left( \!\tfrac{\partial W^{22}}{\partial x^{2}}%
\!-\!\tfrac{\partial W^{33}}{\partial x^{2}}\!+\!\tfrac{\partial W^{44}}{%
\partial x^{2}}\!-\!\tfrac{\partial W^{11}}{\partial x^{2}}\!\right)
V^{23}\!-\!\tfrac{1}{2}\left( \!\tfrac{\partial V^{22}}{\partial x^{2}}\!-\!%
\tfrac{\partial V^{33}}{\partial x^{2}}\!+\!\tfrac{\partial V^{44}}{\partial
x^{2}}\!-\!\tfrac{\partial V^{11}}{\partial x^{2}}\!\right) W^{23} \\
&&\!\!\!\!\!+\left( -\!\tfrac{\partial W^{34}}{\partial x^{2}}\!+\!\tfrac{%
\partial W^{24}}{\partial x^{3}}\!-\!\tfrac{\partial W^{23}}{\partial x^{4}}%
\!\right) V^{24}\!-\!\left( \!-\!\tfrac{\partial V^{34}}{\partial x^{2}}\!+\!%
\tfrac{\partial V^{24}}{\partial x^{3}}\!-\!\tfrac{\partial V^{23}}{\partial
x^{4}}\!\right) W^{24}\! \\
&&\!\!\!\!\!+\tfrac{1}{2}\left( \!\tfrac{\partial W^{23}}{\partial x^{2}}%
\!-\!\tfrac{\partial W^{13}}{\partial x^{1}}\!+\!\tfrac{\partial W^{34}}{%
\partial x^{4}}\!\right) V^{33}\!-\!\tfrac{1}{2}\left( \!\tfrac{\partial
V^{23}}{\partial x^{2}}\!-\!\tfrac{\partial V^{13}}{\partial x^{1}}\!+\!%
\tfrac{\partial V^{34}}{\partial x^{4}}\!\right) W^{33}\! \\
&&\!\!\!\!\!\!+\tfrac{1}{2}\left( -\!\tfrac{\partial W^{33}}{\partial x^{4}}%
\!+\!\tfrac{\partial W^{44}}{\partial x^{4}}\!+\!\tfrac{\partial W^{22}}{%
\partial x^{4}}\!-\!\tfrac{\partial W^{11}}{\partial x^{4}}\!\right)
\!V^{34}\!-\!\tfrac{1}{2}\left( \!-\!\tfrac{\partial V^{33}}{\partial x^{4}}%
\!+\!\tfrac{\partial V^{44}}{\partial x^{4}}\!+\!\tfrac{\partial V^{22}}{%
\partial x^{4}}\!-\!\tfrac{\partial V^{11}}{\partial x^{4}}\!\right) \!W^{34}
\\
&&\!\!\!\!\!+\tfrac{1}{2}\left( \!\tfrac{\partial W^{11}}{\partial x^{3}}%
\!-\!\tfrac{\partial W^{13}}{\partial x^{1}}\!-\!\tfrac{\partial W^{34}}{%
\partial x^{4}}\!-\!\tfrac{\partial W^{22}}{\partial x^{3}}\!+\!\tfrac{%
\partial W^{23}}{\partial x^{2}}\!\right) V^{44}\! \\
&&\!\!\!\!\!-\tfrac{1}{2}\left( \!\tfrac{\partial V^{11}}{\partial x^{3}}%
\!-\!\tfrac{\partial V^{13}}{\partial x^{1}}\!-\!\tfrac{\partial V^{34}}{%
\partial x^{4}}\!-\!\tfrac{\partial V^{22}}{\partial x^{3}}\!+\!\tfrac{%
\partial V^{23}}{\partial x^{2}}\!\right) W^{44},
\end{eqnarray*}%
} {\small
\begin{eqnarray*}
\omega _{2}^{4}\!\!\!\! &=&\!\!\!\!\!\!\tfrac{1}{2}\left( \!\tfrac{\partial
W^{33}}{\partial x^{4}}\!-\!\tfrac{\partial W^{34}}{\partial x^{3}}\!+\!%
\tfrac{\partial W^{22}}{\partial x^{4}}\!-\!\tfrac{\partial W^{24}}{\partial
x^{2}}\!-\!\tfrac{\partial W^{14}}{\partial x^{1}}\!\right) V^{11}\! \\
&&\!\!\!\!\!\!-\tfrac{1}{2}\left( \!\tfrac{\partial V^{33}}{\partial x^{4}}%
\!-\!\tfrac{\partial V^{34}}{\partial x^{3}}\!+\!\tfrac{\partial V^{22}}{%
\partial x^{4}}\!-\!\tfrac{\partial V^{24}}{\partial x^{2}}\!-\!\tfrac{%
\partial V^{14}}{\partial x^{1}}\!\right) W^{11}\! \\
&&\!\!\!\!\!\!+\left( \!-\tfrac{\partial W^{12}}{\partial x^{4}}\!+\!\tfrac{%
\partial W^{24}}{\partial x^{1}}\!+\!\tfrac{\partial W^{14}}{\partial x^{2}}%
\!\right) V^{12}\!-\!\left( \!-\tfrac{\partial V^{12}}{\partial x^{4}}\!+\!%
\tfrac{\partial V^{24}}{\partial x^{1}}\!+\!\tfrac{\partial V^{14}}{\partial
x^{2}}\!\right) W^{12}\! \\
&&\!\!\!\!\!\!+\left( \!-\tfrac{\partial W^{13}}{\partial x^{4}}\!+\!\tfrac{%
\partial W^{34}}{\partial x^{1}}\!+\!\tfrac{\partial W^{14}}{\partial x^{3}}%
\!\right) V^{13}\!-\!\left( \!-\tfrac{\partial V^{13}}{\partial x^{4}}\!+\!%
\tfrac{\partial V^{34}}{\partial x^{1}}\!+\!\tfrac{\partial V^{14}}{\partial
x^{3}}\!\right) W^{13} \\
&&\!\!\!\!\!\!+\tfrac{1}{2}\left( \!\tfrac{\partial W^{11}}{\partial x^{1}}%
\!-\!\tfrac{\partial W^{22}}{\partial x^{1}}\!-\!\tfrac{\partial W^{33}}{%
\partial x^{1}}\!+\!\tfrac{\partial W^{44}}{\partial x^{1}}\!\right) V^{14}-%
\tfrac{1}{2}\left( \!\tfrac{\partial V^{11}}{\partial x^{1}}\!-\!\tfrac{%
\partial V^{22}}{\partial x^{1}}\!-\!\tfrac{\partial V^{33}}{\partial x^{1}}%
\!+\!\tfrac{\partial V^{44}}{\partial x^{1}}\!\right) W^{14}\! \\
&&\!\!\!\!\!\!+\tfrac{1}{2}\left( \!-\tfrac{\partial W^{24}}{\partial x^{4}}%
\!+\!\tfrac{\partial W^{34}}{\partial x^{3}}\!-\!\tfrac{\partial W^{14}}{%
\partial x^{1}}\!-\!\tfrac{\partial W^{33}}{\partial x^{4}}\!+\!\tfrac{%
\partial W^{11}}{\partial x^{4}}\!\right) V^{22}\! \\
&&\!\!\!\!\!\!-\tfrac{1}{2}\left( \!-\tfrac{\partial V^{24}}{\partial x^{4}}%
\!+\!\tfrac{\partial V^{34}}{\partial x^{3}}\!-\!\tfrac{\partial V^{14}}{%
\partial x^{1}}\!-\!\tfrac{\partial V^{33}}{\partial x^{4}}\!+\!\tfrac{%
\partial V^{11}}{\partial x^{4}}\!\right) W^{22}\! \\
&&\!\!\!\!\!\!+\left( \!-\tfrac{\partial W^{34}}{\partial x^{2}}\!-\!\tfrac{%
\partial W^{24}}{\partial x^{3}}\!+\!\tfrac{\partial W^{23}}{\partial x^{4}}%
\!\right) V^{23}\!-\!\left( \!-\tfrac{\partial V^{34}}{\partial x^{2}}\!-\!%
\tfrac{\partial V^{24}}{\partial x^{3}}\!+\!\tfrac{\partial V^{23}}{\partial
x^{4}}\!\right) W^{23}\! \\
&&\!\!\!\!\!\!+\tfrac{1}{2}\left( \!-\tfrac{\partial W^{11}}{\partial x^{2}}%
\!+\!\tfrac{\partial W^{22}}{\partial x^{2}}\!+\!\tfrac{\partial W^{33}}{%
\partial x^{2}}\!-\!\tfrac{\partial W^{44}}{\partial x^{2}}\!\right)
\!V^{24}\!-\!\tfrac{1}{2}\left( \!-\tfrac{\partial V^{11}}{\partial x^{2}}%
\!+\!\tfrac{\partial V^{22}}{\partial x^{2}}\!+\!\tfrac{\partial V^{33}}{%
\partial x^{2}}\!-\!\tfrac{\partial V^{44}}{\partial x^{2}}\!\right) \!W^{24}
\\
&&\!\!\!\!\!\!+\tfrac{1}{2}\left( \!-\tfrac{\partial W^{14}}{\partial x^{1}}%
\!-\!\tfrac{\partial W^{34}}{\partial x^{3}}\!+\!\tfrac{\partial W^{24}}{%
\partial x^{2}}\!+\!\tfrac{\partial W^{11}}{\partial x^{4}}\!-\!\tfrac{%
\partial W^{22}}{\partial x^{4}}\!\right) V^{33}\! \\
&&\!\!\!\!\!\!-\tfrac{1}{2}\left( \!-\tfrac{\partial V^{14}}{\partial x^{1}}%
\!-\!\tfrac{\partial V^{34}}{\partial x^{3}}\!+\!\tfrac{\partial V^{24}}{%
\partial x^{2}}\!+\!\tfrac{\partial V^{11}}{\partial x^{4}}\!-\!\tfrac{%
\partial V^{22}}{\partial x^{4}}\!\right) W^{33}\! \\
&&\!\!\!\!\!\!+\tfrac{1}{2}\!\left( \tfrac{\partial W^{22}}{\partial x^{3}}%
\!+\!\tfrac{\partial W^{33}}{\partial x^{3}}\!-\!\tfrac{\partial W^{11}}{%
\partial x^{3}}\!-\!\tfrac{\partial W^{44}}{\partial x^{3}}\!\right)
\!V^{34}\!-\tfrac{1}{2}\!\left( \!\tfrac{\partial V^{22}}{\partial x^{3}}%
\!+\!\tfrac{\partial V^{33}}{\partial x^{3}}\!-\!\tfrac{\partial V^{11}}{%
\partial x^{3}}\!-\!\tfrac{\partial V^{44}}{\partial x^{3}}\!\right) \!W^{34}
\\
&&\!\!\!\!\!\!+\tfrac{1}{2}\left( \!-\tfrac{\partial W^{14}}{\partial x^{1}}%
\!+\!\tfrac{\partial W^{34}}{\partial x^{3}}\!+\!\tfrac{\partial W^{24}}{%
\partial x^{2}}\!\right) V^{44}\!-\!\tfrac{1}{2}\left( \!-\tfrac{\partial
V^{14}}{\partial x^{1}}\!+\!\tfrac{\partial V^{34}}{\partial x^{3}}\!+\!%
\tfrac{\partial V^{24}}{\partial x^{2}}\!\right) W^{44}.
\end{eqnarray*}%
} Hence%
\begin{equation*}
\begin{array}{ll}
\omega _{2}^{i}(X_{1}^{k},X_{1}^{l})=0, & \omega
_{2}^{i}(X_{2}^{k},X_{2}^{l})=0, \\
\omega _{2}^{i}(X_{3}^{k},X_{2}^{l})=0, & \omega
_{2}^{i}(X_{7}^{k},X_{2}^{l})=0, \\
\omega _{2}^{i}(X_{3}^{k},X_{3}^{l})=0, & \omega
_{2}^{i}(X_{5}^{k},X_{3}^{l})=0, \\
\omega _{2}^{i}(X_{5}^{k},X_{5}^{l})=0, & \omega
_{2}^{i}(X_{7}^{k},X_{3}^{l})=0, \\
\omega _{2}^{i}(X_{7}^{k},X_{5}^{l})=0, & \omega
_{2}^{i}(X_{7}^{k},X_{7}^{l})=0, \\
\omega _{2}^{i}(X_{8}^{k},X_{8}^{l})=0, &
\end{array}%
\quad 1\leq i\leq 4,
\end{equation*}%
\begin{equation*}
\begin{array}{ll}
\omega _{2}^{i}(X_{6}^{k},X_{5}^{l})=0, & \omega
_{2}^{i}(X_{8}^{k},X_{5}^{l})=0,%
\end{array}%
\quad i\neq 4,
\end{equation*}%
\begin{equation*}
\omega _{2}^{4}(X_{6}^{k},X_{5}^{l})=-ik_{3}(k_{1}+l_{1})\exp \left( i\left[
(k_{1}+l_{1})x^{1}+k_{3}x^{3}\right] \right) ,
\end{equation*}%
\begin{equation*}
\omega _{2}^{4}(X_{8}^{k},X_{5}^{l})=-\tfrac{i(k_{1}+l_{1})}{2}\exp \left( i%
\left[ (k_{1}+l_{1})x^{1}\right] \right) ,
\end{equation*}%
\begin{equation*}
\begin{array}{ll}
\omega _{2}^{i}(X_{2}^{k},X_{1}^{l})=0, & \omega
_{2}^{i}(X_{3}^{k},X_{1}^{l})=0, \\
\omega _{2}^{i}(X_{5}^{k},X_{2}^{l})=0, & \omega
_{2}^{i}(X_{6}^{k},X_{2}^{l})=0, \\
\omega _{2}^{i}(X_{6}^{k},X_{3}^{l})=0, & \omega
_{2}^{i}(X_{7}^{k},X_{6}^{l})=0, \\
\omega _{2}^{i}(X_{8}^{k},X_{2}^{l})=0, & \omega
_{2}^{i}(X_{8}^{k},X_{3}^{l})=0, \\
\omega _{2}^{i}(X_{8}^{k},X_{7}^{l})=0, &
\end{array}%
\quad i\neq 2,
\end{equation*}%
\begin{eqnarray*}
\omega _{2}^{2}(X_{2}^{k},X_{1}^{l}) &=&c_{21}\exp \left( i\left[ \left(
k_{1}+l_{1}\right) x^{1}+\left( k_{3}+l_{3}\right) x^{3}+k_{4}x^{4}\right]
\right) , \\
c_{21} &=&\tfrac{i}{2}\tfrac{\left( (k_{4})^{2}+k_{1}l_{1}-k_{3}l_{3}\right)
\left( (l_{1})^{2}-(l_{3})^{2}\right) +\left( (k_{1})^{2}+(k_{3})^{2}\right)
\left( (l_{1})^{2}+(l_{3})^{2}\right) }{k_{4}(l_{3})^{2},},
\end{eqnarray*}%
\begin{eqnarray*}
\omega _{2}^{2}(X_{3}^{k},X_{1}^{l}) &=&c_{31}\exp \left( i\left[ \left(
k_{1}+l_{1}\right) x^{1}+\left( k_{3}+l_{3}\right) x^{3}\right] \right) , \\
c_{31} &=&\tfrac{i}{2}\tfrac{\left( k_{1}l_{1}-k_{3}l_{3}\right) \left(
(l_{1})^{2}-(l_{3})^{2}\right) +\left( (k_{1})^{2}+(k_{3})^{2}\right) \left(
(l_{1})^{2}+(l_{3})^{2}\right) }{k_{3}(l_{3})^{2}},
\end{eqnarray*}%
\begin{equation*}
\omega _{2}^{2}(X_{5}^{k},X_{2}^{l})=2il_{1}\exp \left( i\left[ \left(
k_{1}+l_{1}\right) x^{1}+l_{3}x^{3}+l_{4}x^{4}\right] \right) ,
\end{equation*}%
\begin{eqnarray*}
\omega _{2}^{2}(X_{6}^{k},X_{2}^{l}) &=&c_{62}\exp \left( i\left[ \left(
k_{1}+l_{1}\right) x^{1}+\left( k_{3}+l_{3}\right) x^{3}+l_{4}x^{4}\right]
\right) , \\
c_{62} &=&-i\tfrac{(k_{3}k_{1}l_{1}-\left( k_{3}\right)
^{2}l_{3}-2l_{1}l_{3}+k_{3}\left( l_{1}\right) ^{2}+k_{3}\left( l_{3}\right)
^{2}+k_{3}\left( l_{4}\right) ^{2})}{l_{4}},
\end{eqnarray*}%
\begin{eqnarray*}
\omega _{2}^{2}(X_{6}^{k},X_{3}^{l}) &=&c_{63}\exp \left( i\left[ \left(
k_{1}+l_{1}\right) x^{1}+\left( k_{3}+l_{3}\right) x^{3}\right] \right) , \\
c_{63} &=&-i\tfrac{(k_{3}k_{1}l_{1}-(k_{3})^{2}l_{3}+k_{3}\left(
l_{1}\right) ^{2}-2l_{1}l_{3}+k_{3}\left( l_{3}\right) ^{2})}{l_{3}},
\end{eqnarray*}%
\begin{equation*}
\omega _{2}^{2}(X_{7}^{k},X_{6}^{l})=i(k_{1}+l_{1})l_{3}\exp \left( \left[
i(\left( k_{1}+l_{1}\right) x^{1}+l_{3}x^{3})\right] \right) ,
\end{equation*}%
\begin{equation*}
\omega _{2}^{2}(X_{8}^{k},X_{2}^{l})=-\tfrac{i}{2}\tfrac{k_{1}l_{1}+\left(
l_{1}\right) ^{2}+\left( l_{3}\right) ^{2}+\left( l_{4}\right) ^{2}}{l_{4}}%
\exp \left( i\left[ (\left( k_{1}+l_{1}\right) x^{1}+l_{3}x^{3}+l_{4}x^{4}%
\right] \right) ,
\end{equation*}%
\begin{equation*}
\omega _{2}^{2}(X_{8}^{k},X_{3}^{l})=-\tfrac{i}{2}\tfrac{%
k_{1}l_{1}+(l_{1})^{2}+(l_{3})^{2}}{l_{3}}\exp \left( i\left[ \left(
k_{1}+l_{1}\right) x^{1}+l_{3}x^{3}\right] \right) ,
\end{equation*}%
\begin{equation*}
\omega _{2}^{2}(X_{8}^{k},X_{7}^{l})=-\tfrac{i}{2}(k_{1}+l_{1})\exp \left[
i(k_{1}+l_{1})x^{1}\right] ,
\end{equation*}%
\begin{equation*}
\begin{array}{rl}
\omega _{2}^{i}(X_{4}^{k},X_{1}^{l})= & c_{41}^{i}\exp \left( i\left[ \left(
k_{1}+l_{1}\right) x^{1}+k_{2}x^{2}+\left( k_{3}+l_{3}\right)
x^{3}+k_{4}x^{4}\right] \right) , \\
& 1\leq i\leq 4,%
\end{array}%
\end{equation*}%
\begin{equation*}
\begin{array}{rl}
c_{41}^{1}= & -\tfrac{i}{4}\tfrac{%
k_{1}(l_{1})^{2}+k_{1}(l_{3})^{2}-2l_{1}(l_{3})^{2}}{l_{3}^{2}}, \\
c_{41}^{2}\!= & \tfrac{i}{4}\tfrac{k_{1}(l_{1})^{3}-k_{3}(l_{1})^{2}l_{3}+%
\left[ (k_{1})^{2}+(k_{3})^{2}\right] \left[ (l_{1})^{2}+(l_{3})^{2}\right]
-k_{1}l_{1}\left( l_{3}\right) ^{2}+k_{3}(l_{3})^{3}+(k_{4})^{2}\left[
(l_{1})^{2}-(l_{3})^{2}\right] }{k_{2}(l_{3})^{2}}, \\
c_{41}^{3}= & -\tfrac{i}{4}\tfrac{%
k_{3}(l_{3})^{2}-2(l_{1})^{2}l_{3}+k_{3}(l_{1})^{2}}{(l_{3})^{2}}, \\
c_{41}^{4}= & -\tfrac{i}{4}k_{4}\tfrac{(l_{1})^{2}-(l_{3})^{2}}{(l_{3})^{2}},%
\end{array}%
\end{equation*}%
\begin{equation*}
\omega _{2}^{1}(X_{4}^{k},X_{2}^{l})=\tfrac{i}{2}\tfrac{k_{2}l_{1}}{l_{4}}%
\exp \left( i\left[ \left( k_{1}+l_{1}\right) x^{1}+k_{2}x^{2}+\left(
k_{3}+l_{3}\right) x^{3}+\left( k_{4}+l_{4}\right) x^{4}\right] \right) ,
\end{equation*}%
\begin{eqnarray*}
\omega _{2}^{2}(X_{4}^{k},X_{2}^{l}) &=&c_{42}\exp \left( i\left[ \left(
k_{1}+l_{1}\right) x^{1}+k_{2}x^{2}+\left( k_{3}+l_{3}\right) x^{3}+\left(
k_{4}+l_{4}\right) x^{4}\right] \right) , \\
c_{42} &=&-\tfrac{i}{2}\tfrac{\left( l_{1}+k_{1}\right) l_{1}-\left(
l_{3}+k_{3}\right) l_{3}-\left( l_{4}+k_{4}\right) l_{4}}{l_{4}},
\end{eqnarray*}%
\begin{align*}
\omega _{2}^{3}(X_{4}^{k},X_{2}^{l})& =-\tfrac{i}{2}\tfrac{k_{2}l_{3}}{l_{4}}%
\exp \left( i\left[ \left( k_{1}\!+\!l_{1}\right)
x^{1}\!+\!k_{2}x^{2}\!+\!\left( k_{3}\!+\!l_{3}\right) x^{3}\!+\!\left(
k_{4}\!+\!l_{4}\right) x^{4}\right] \right) , \\
\omega _{2}^{4}(X_{4}^{k},X_{2}^{l})& =-\tfrac{i}{2}k_{2}\exp \left( i\left[
\left( k_{1}+l_{1}\right) x^{1}+k_{2}x^{2}+\left( k_{3}+l_{3}\right)
x^{3}+\left( k_{4}+l_{4}\right) x^{4}\right] \right) ,
\end{align*}%
\begin{align*}
\omega _{2}^{1}(X_{4}^{k},X_{3}^{l})& =\tfrac{i}{2}\tfrac{k_{2}l_{1}}{l_{3}}%
\exp \left( i\left[ \left( k_{1}+l_{1}\right) x^{1}+k_{2}x^{2}+\left(
k_{3}+l_{3}\right) x^{3}+k_{4}x^{4}\right] \right) , \\
\omega _{2}^{2}(X_{4}^{k},X_{3}^{l})& =c_{43}\exp \left( i\left[ \left(
k_{1}+l_{1}\right) x^{1}+k_{2}x^{2}+\left( k_{3}+l_{3}\right)
x^{3}+k_{4}x^{4}\right] \right) , \\
c_{43}& =-\tfrac{i}{2}\tfrac{\left( l_{1}+k_{1}\right) l_{1}-\left(
k_{3}+l_{3}\right) l_{3}}{l_{3}}, \\
\omega _{2}^{3}(X_{4}^{k},X_{3}^{l})& =-\tfrac{i}{2}k_{2}\exp \left( i\left[
\left( k_{1}+l_{1}\right) x^{1}+k_{2}x^{2}+\left( k_{3}+l_{3}\right)
x^{3}+k_{4}x^{4}\right] \right) , \\
\omega _{2}^{4}(X_{4}^{k},X_{3}^{l})& =0,
\end{align*}%
\begin{equation*}
\begin{array}{rl}
\omega _{2}^{i}(X_{4}^{k},X_{4}^{l})= & c_{44}^{i}\exp \left( i\left[ \left(
k_{1}\!+\!l_{1}\right) x^{1}\!+\!\left( k_{2}\!+\!l_{2}\right)
x^{2}\!+\!\left( k_{3}\!+\!l_{3}\right) x^{3}\!+\!\left(
k_{4}\!+\!l_{4}\right) x^{4}\right] \right) , \\
& 1\leq i\leq 4, \\
c_{44}^{1}= & -\tfrac{i}{4}\tfrac{(k_{1}l_{2}-l_{1}k_{2})(l_{2}+k_{2})}{%
k_{2}l_{2}}, \\
c_{44}^{2}\!= & \!\tfrac{i}{4}\tfrac{\left( k_{1}+l_{1}\right) \left(
l_{2}k_{1}-k_{2}l_{1}\right) +\left( l_{3}+k_{3}\right) \left(
k_{2}l_{3}-l_{2}k_{3}\right) +\left( l_{4}+k_{4}\right) \left(
k_{2}l_{4}-l_{2}k_{4}\right) }{k_{2}l_{2}}, \\
c_{44}^{3}= & -\tfrac{i}{4}\tfrac{(l_{3}k_{2}-k_{3}l_{2})(l_{2}+k_{2})}{%
l_{2}k_{2}}, \\
c_{44}^{4}= & -\tfrac{i}{4}\tfrac{(l_{4}k_{2}-k_{4}l_{2})(l_{2}+k_{2})}{%
l_{2}k_{2}},%
\end{array}%
\end{equation*}%
\begin{eqnarray*}
\omega _{2}^{i}(X_{5}^{k},X_{1}^{l}) &=&0,\quad 1\leq i\leq 3, \\
\omega _{2}^{4}(X_{5}^{k},X_{1}^{l}) &=&\tfrac{i}{2}\tfrac{%
(l_{1})^{3}+k_{1}(l_{1})^{2}+k_{1}(l_{3})^{2}-l_{1}(l_{3})^{2}}{(l_{3})^{2}}%
\exp \left( i\left[ (k_{1}+l_{1})x^{1}+l_{3}x^{3}\right] \right)
\end{eqnarray*}%
\begin{eqnarray*}
\omega _{2}^{1}(X_{5}^{k},X_{4}^{l}) &=&-\tfrac{i}{2}l_{4}\exp \left( i\left[
(k_{1}+l_{1})x^{1}+l_{2}x^{2}+l_{3}x^{3}+l_{4}x^{4}\right] \right) , \\
\omega _{2}^{2}(X_{5}^{k},X_{4}^{l}) &=&i\tfrac{l_{1}l_{4}}{l_{2}}\exp
\left( i\left[ (k_{1}+l_{1})x^{1}+l_{2}x^{2}+l_{3}x^{3}+l_{4}x^{4}\right]
\right) , \\
\omega _{2}^{3}(X_{5}^{k},X_{4}^{l}) &=&0, \\
\omega _{2}^{4}(X_{5}^{k},X_{4}^{l}) &=&\tfrac{i}{2}(k_{1}-l_{1})\exp \left(
i\left[ (k_{1}+l_{1})x^{1}+l_{2}x^{2}+l_{3}x^{3}+l_{4}x^{4}\right] \right) ,
\end{eqnarray*}%
\begin{equation*}
\begin{array}{rl}
\omega _{2}^{i}(X_{6}^{k},X_{4}^{l})= & c_{64}^{i}\exp \left( i\left[
(k_{1}+l_{1})x^{1}+l_{2}x^{2}+(k_{3}+l_{3})x^{3}+l_{4}x^{4}\right] \right) ,
\\
& 1\leq i\leq 4, \\
c_{64}^{1}= & \tfrac{i}{2}(k_{3}l_{1}+k_{3}-l_{3}), \\
c_{64}^{2}= & -\tfrac{i\left[
k_{1}k_{3}l_{1}-(k_{3})^{2}l_{3}+k_{3}(l_{1})^{2}-2l_{3}l_{1}
+k_{3}(l_{3})^{2}+k_{3}(l_{4})^{2}\right] }{2l_{2}}, \\
c_{64}^{3}= & \tfrac{i\left[ k_{1}-l_{1}+k_{3}l_{3}-2(k_{3})^{2}\right] }{2},
\\
c_{64}^{4}= & \tfrac{i}{2}k_{3}l_{4},%
\end{array}%
\end{equation*}%
\begin{eqnarray*}
\omega _{2}^{1}(X_{6}^{k},X_{6}^{l}) 
&=&i\left[ (k_{3})^{2}-(l_{3})^{2}\right] 
\exp \left( i\left[ (k_{1}+l_{1})x^{1}+(k_{3}+l_{3})x^{3}\right]
\right) , \\
\omega _{2}^{2}(X_{6}^{k},X_{6}^{l}) 
&=&0, \\
\omega _{2}^{3}(X_{6}^{k},X_{6}^{l}) 
&=&-i\left( k_{1}+l_{1}\right) \left(
k_{3}-l_{3}\right) \exp \left( i\left[ (k_{1}+l_{1})x^{1}+(k_{3}+l_{3})x^{3}
\right] \right) , \\
\omega _{2}^{4}(X_{6}^{k},X_{6}^{l}) 
&=&0,
\end{eqnarray*}
\begin{eqnarray*}
\omega _{2}^{1}(X_{7}^{k},X_{4}^{l}) &=&-\tfrac{i}{2}l_{2}
\exp \left( i\left[
(k_{1}+l_{1})x^{1}+l_{2}x^{2}+l_{3}x^{3}+l_{4}x^{4}\right] \right) , \\
\omega _{2}^{2}(X_{7}^{k},X_{4}^{l}) &=&\tfrac{i(k_{1}+l_{1})}{2}
\exp \left(
i\left[ (k_{1}+l_{1})x^{1}+l_{2}x^{2}+l_{3}x^{3}+l_{4}x^{4}\right] \right) ,
\\
\omega _{2}^{3}(X_{7}^{k},X_{4}^{l}) &=&0, \\
\omega _{2}^{4}(X_{7}^{k},X_{4}^{l}) &=&0,
\end{eqnarray*}%
\begin{eqnarray*}
\omega _{2}^{i}(X_{8}^{k},X_{4}^{l}) &=&c_{84}^{i}
\exp \left( i\left[
(k_{1}+l_{1})x^{1}+l_{2}x^{2}+l_{3}x^{3}+l_{4}x^{4}\right] \right) , \\
c_{84}^{j} &=&\tfrac{il_{j}}{4},\text{ }j=1,3,4, \\
c_{84}^{2} &=&-\tfrac{i\left(
k_{1}l_{1}+(l_{1})^{2}+(l_{3})^{2}+(l_{4})^{2}\right) }{4l_{2}}.
\end{eqnarray*}

From the previous formulas it follows the closed $3$-form 
$\omega _{2}(X_{a}^{k},X_{b}^{l})$ is exact except in the following cases, where 
$[\omega _{3}]$ denotes the cohomology class of a closed $3$-form $\omega _{3} $:
\begin{equation*}
\begin{array}{lll}
\left[ \omega _{2}(X_{5}^{k},X_{4}^{l})\right] = & ik_{1}\left[ v_{4}\right] , 
& k_{1}+l_{1}=l_{2}=l_{3}=0,\medskip \\
\left[ \omega _{2}(X_{6}^{k},X_{4}^{l})\right] = & -\tfrac{ik_{3}\left(
(k_{3})^{2}-k_{1}\right) }{l_{2}}\left[ v_{2}\right] , &
k_{1}+l_{1}=k_{3}+l_{3}=l_{4}=0,l_{2}\neq 0,\medskip \\
\left[ \omega _{2}(X_{6}^{k},X_{4}^{l})\right] = & \tfrac{i\left(
2k_{1}+k_{3}l_{3}-2(k_{3})^{2}\right) }{2}\left[ v_{3}\right] ,
& k_{1}+l_{1}=l_{2}=l_{4}=0,\medskip \\
\left[ \omega _{2}(X_{6}^{k},X_{4}^{l})\right] = & \tfrac{ik_{3}l_{4}}{2}%
\left[ v_{4}\right] , & k_{1}+l_{1}=k_{3}+l_{3}=l_{2}=0,\medskip \\
\left[ \omega _{2}(X_{8}^{k},X_{4}^{l})\right] = & \tfrac{il_{1}}{4}\left[
v_{1}\right] , & l_{2}=l_{3}=l_{4}=0,\medskip \\
\left[ \omega _{2}(X_{8}^{k},X_{4}^{l})\right] = & \tfrac{il_{3}}{4}\left[
v_{3}\right] , & k_{1}+l_{1}=l_{2}=l_{4}=0,\medskip \\
\left[ \omega _{2}(X_{8}^{k},X_{4}^{l})\right] = & \tfrac{il_{4}}{4}\left[
v_{4}\right] , & k_{1}+l_{1}=l_{2}=l_{3}=0,%
\end{array}%
\end{equation*}
\begin{equation*}
\begin{array}{rll}
\left[ \omega _{2}(X_{4}^{k},X_{1}^{l})\right] = & 8i\pi ^{3}
\tfrac{(k_{1})^{2}}{k_{2}}\left[ v_{2}\right] , & k_{2}\neq 0, \\
\left[ \omega _{2}(X_{4}^{k},X_{1}^{l})\right] = & 8\pi ^{3}\bar{c}_{41}^{1}
\left[ v_{1}\right] , & k_{1}+l_{1}\neq 0,k_{2}=k_{4}=0, \\
\bar{c}_{41}^{1}= & -\tfrac{i}{4}
\tfrac{k_{1}(l_{1})^{2}+k_{1}(k_{3})^{2}-2l_{1}(k_{3})^{2}}
{\left( k_{3}\right) ^{2}}, 
&  \\
\left[ \omega _{2}(X_{4}^{k},X_{1}^{l})\right] = & 8\pi ^{3}\bar{c}_{41}^{3}
\left[ v_{3}\right] , & k_{1}+l_{1}=k_{2}=k_{4}=0, \\
\bar{c}_{41}^{3}= & -\tfrac{i}{4}\tfrac{%
k_{3}(l_{3})^{2}-2(k_{1})^{2}l_{3}+k_{3}(k_{1})^{2}}{(l_{3})^{2}}, &  \\
\left[ \omega _{2}(X_{4}^{k},X_{1}^{l})\right] = & 8\pi ^{3}\bar{c}_{41}^{4}
\left[ v_{4}\right] , & k_{2}=0,k_{4}\neq 0, \\
\bar{c}_{41}^{4}= & -\tfrac{i}{4}k_{4}
\tfrac{(k_{1})^{2}-(k_{3})^{2}}{(k_{3})^{2}}. &
\end{array}
\end{equation*}
\end{example}

\end{document}